\newif\ifreport

\reportfalse

\ifreport
\documentclass[journal]{IEEEtran}
\else
\documentclass[conference]{IEEEtran}
\renewcommand\baselinestretch{1}
\fi

\usepackage{enumitem}
\usepackage{multirow}
\usepackage{booktabs}
\usepackage{hyperref}
\usepackage{cite}
\usepackage{graphicx}
\usepackage{amsmath}
\DeclareMathOperator{\Tr}{Tr}
\usepackage{amssymb}
\usepackage{amsthm}
\usepackage{color}
\theoremstyle{definition}

\newtheorem{theorem}{Theorem}
\newtheorem{lemma}{Lemma}

\IEEEoverridecommandlockouts

\begin{document}

\title{\vspace{-0mm}Joint Beam and Channel Tracking for Two-Dimensional Phased Antenna Arrays\vspace{0mm}}
% \title{\vspace{-2mm}Super Fast Beam and Channel Tracking in 2D Phased Antenna Arrays\vspace{-10mm}}
% \title{Super Fast Beam Tracking in Phased \\ Antenna Arrays}

 \author{Yu Liu$^*$, Jiahui Li$^*$, Yin Sun$^\S$, Shidong Zhou$^*$\\
 $^*$Dept. of EE, Tsinghua University, Beijing, 100084, China\\
 $^\S$Dept. of ECE, Auburn University, Auburn AL, 36849, U.S.A}

\maketitle
\begin{abstract}
Analog beamforming is a low-cost architecture for millimeter-wave (mmWave) mobile communications. However, it has two disadvantages for serving fast mobility users: (i) the mmWave beam in the wireless channel and the beam steered by analog beamforming have small angular spreads which are difficult to align with each other and (ii) the receiver can only observe the mmWave channel in one beam direction and rely on beam-probing algorithms to check other directions. In this paper, we develop a beam probing and tracking algorithm that can efficiently track fast-moving mmWave beams in three-dimensional (3D) space. This algorithm has several salient features: (1) fading channel supportive: it can simultaneously track the channel coefficient and two-dimensional (2D) beam direction in fading channel environments; (2) low probing overhead: it achieves the minimum probing requirement for joint beam and channel tracking; (3) fast tracking speed and high tracking accuracy: its tracking error converges to the minimum Cram\'{e}r-Rao lower bound (CRLB) in static scenarios in theory and it outperforms several existing tracking algorithms with lower tracking error and faster tracking speed in simulations.
% Millimeter wave (mmWave) is an attractive candidate for high-speed mobile communications in the future. However, due to the propagation characteristics of mmWave, beam and and and and alignment becomes a key challenge for serving users with fast moving speeds. In this paper, we develop a joint beam and channel tracking algorithm that can track beams from the horizontal and vertical directions by using two-dimensional (2D) phased antenna arrays. A general sequence of optimal trial beamforming parameters is obtained to achieve the minimum Cramer-Rao lower bound (CRLB) of joint beam and channel tracking asymptotically as antenna number grows to infinity. This sequence is proved to be asymptotically optimal in different conditions, e.g., channel coefficients, path directions, and antenna array sizes. We prove that the proposed algorithm converges to the minimum CRLB in static scenarios. Simulation results show that our algorithm outperforms several existing algorithms in tracking accuracy and speed band.
\end{abstract}

\bstctlcite{BSTcontrol}

\vspace{-0mm}%2mm
\section{Introduction}\label{sec_intro}
Due to the low hardware cost and energy consumption, analog beamforming is often used in mmWave mobile communications to provide large array gains \cite{Heath2016overview,Molisch2017Hybrid}. However, the beam steered by analog beamforming has small angular spreads. Slight misalignment can cause severe energy loss. Accurate alignment can be achieved by beam training at the expense of large pilot overhead in static or quasi-static scenarios. Nevertheless, this price is unacceptable in fast-moving environments. Therefore, efficient beam tracking is important for serving fast mobility users in mmWave communication.

Some beam tracking methods has been proposed \cite{Zhu2016,Alkhateeb2015Limited,Alkhateeb2015Compressed}, utilizing historical observations and estimations to obtain current estimate. Despite this, the analog beamforming vectors are not optimized in those tracking algorithms, resulting in a waste of transmission energy. A beam tracking algorithm is proposed in \cite{JLiAnalogbeamtracking2017}, trying to optimize the analog beamforming vectors, assuming the channel coefficient is known. In \cite{JLiJoint2018}, the authors start to jointly track the channel coefficient and beam direction with optimal analog beamforming vectors. The theorems of convergence and optimality are established for joint tracking. However, all these algorithms are based on uniform linear array (ULA) antennas, which can only support one-dimensional (1D) beam tracking. While in several mobile scenarios, e.g., unmanned aerial vehicle (UAV) scenarios \cite{Brown2016Promise}, the beam may also come from different horizontal and vertical directions. Hence, we need to dynamically track the two-dimensional (2D) beam direction with 2D phased antenna arrays.
%dense urban area \cite{Rappaport2013CompactBroadband}

This problem is challenging due to the following three reasons: (i) with analog beamforming, we can only obtain part of the system information through one observation. (ii) We need to jointly track channel coefficient and 2D beam direction and the analog beamforming vectors also need to be adjusted. Therefore, it is a dynamic joint optimization problem with sequential analog beamforming vectors and these analog beamforming vectors also need to be optimized. (iii) Compared with 1D beam direction, more analog beamforming vectors are required when tracking 2D beam direction. As a result, the optimization dimension greatly increases.

In this paper, we design a joint beam and channel tracking algorithm for 2D phased antenna arrays to handle the problem above. The main contributions and results are summarized as follows:
\begin{itemize}
\item This algorithm can achieve the minimum probing overhead for joint beam and channel tracking.
\item In static scenarios, we get the performance bound, i.e., the minimum CRLB by optimizing the analog beamforming vectors under some constraints. A general way to generate the optimal analog beamforming vectors is proposed with a sequence of parameters. These parameters are proved to be asymptotically optimal in different conditions, e.g., channel coefficients, and path directions, as the number of antennas grows to infinity.
   % A sequence of general asymptotically optimal beamforming parameters is obtained and proven to be optimal even for different channel gains, different beam directions and different antenna numbers when antenna number is sufficiently large.
\item We prove that our algorithm can converge to the minimum CRLB with high probability in static scenarios.
\item Simulation results show that our algorithm approaches the minimum CRLB quickly in static scenarios. In dynamic scenarios, our algorithm can achieve lower tracking error and faster tracking speed compared with several existing algorithms.
\end{itemize}

\vspace{-0mm}%-3mm
\section{System Model}\label{sec_model}
%\fi
%\subsection{System Model}
\vspace{-0mm}%-1mm
We consider a mmWave receiver equipped with a planar phased antenna array\footnote{Note that tracking is needed at both the transmitter and receiver. However, considering the transmitter-receiver reciprocity, the beam and channel tracking of both sides have similar designs. Hence, we focus on beam and channel tracking on the receiver side.}, as shown in Fig. \ref{antenna}. The planar array consists of $M \times N$ antenna elements that are placed in a rectangular area, with a distance $d_1$ ($d_2$) between neighboring antenna elements along $x$-axis ($y$-axis)\footnote{To obtain different resolutions in horizontal direction and vertical direction, the antenna numbers along different directions may not be the same, i.e., $M \neq N$ \cite{Roh2014mmwave}. To suppress sidelobe, the antennas may be unequally spaced, i.e., $d_1 \neq d_2$ \cite{Rajarshi2012Position}.}. The antenna elements are connected to the same RF chain through different phase shifters. The system is time-slotted. To estimate and track the direction of the incoming beam, the transmitter sends $q$ pre-determined pilot symbols $s_p$ in each time slot, where $\lvert s_p \rvert^2=\text{E}_p$ is the transmit power of each pilot symbol.

In mmWave channels, only a few paths exist due to the weak scattering effect \cite{Heath2016overview}. Because the angle spread is small and the mmWave system is usually configured with a large number of antennas, the interaction between multi-paths is relatively weak. In other words, the incoming beam paths are usually sparse in space, making it possible to track each path independently \cite{Zhenyu2016Enabling}. Hence, we focus on the method for tracking one path. Different paths can be tracked separately by using the same method.
\begin{figure}[!t]
	\centering
	\vspace{-5mm}
	\includegraphics[width=5.5cm]{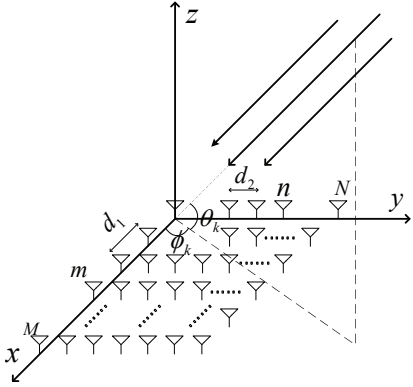}
	\vspace{-4mm}
	\caption{2D phased antenna array.}
	\vspace{-5mm}
	\label{antenna}
\end{figure}In time-slot $k$, the direction of the incoming beam path is denoted by ($\theta_k, \varphi_k$), where $\theta_k\in[0,\pi/2]$ is the elevation angle of arrival (AoA) and $\varphi_k\in[-\pi,\pi)$ is the azimuth AoA. The channel vector of this path is
\vspace{-1mm}
\begin{equation}\label{eq_channelvector}
\textbf{h}_k = \beta_k \textbf{a}(\textbf{x}_k),
\vspace{-1mm}
\end{equation}
%be the ,
where $\beta_k=\beta_k^\text{re}+j\beta_k^\text{im}$ is the complex channel coefficient, $\textbf{x}_k=\left[x_{k,1},x_{k,2} \right]^\text{T}=\left[\frac{M d_1 \cos(\theta_k) \cos(\varphi_k)}{\lambda},\frac{N d_2 \cos(\theta_k) \sin(\varphi_k)}{\lambda} \right]^\text{T}$ is the direction parameter vector determined by ($\theta_k, \varphi_k$),
\vspace{-1.2mm}
\begin{equation}\label{eq_steeringvector}
\textbf{a}(\textbf{x}_k)=\left[ a_{11}(\textbf{x}_k)~ \cdots~a_{1N}(\textbf{x}_k)~a_{21}(\textbf{x}_k) ~\cdots~a_{MN}(\textbf{x}_k) \right]^\text{T}
\vspace{-1.2mm}
\end{equation}is the steering vector with $a_{mn}(\textbf{x}_k)=e^{j 2 \pi \left(\frac{m-1}{M} x_{k,1} + \frac{n-1}{N} x_{k,2}\right)}$ $\,(m=1,\cdots,M;\,n=1,\cdots,N)$, and $\lambda$ is the wavelength.

Let $\textbf{w}_{k,i}$ be the analog beamforming vector for receiving the $i$-th $(i=1,\cdots,q)$ pilot symbol in time-slot $k$, given by
\vspace{-1.2mm}
\begin{equation}\label{eq_bf}
\textbf{w}_{k,i} = \frac{1}{\sqrt{MN}}\textbf{a}\left(\textbf{x}_k+{\boldsymbol{\Delta}}_{k,i}\right),
%\textbf{w}_{k,i} = \frac{1}{\sqrt{MN}}\textbf{a}\left(\left[x_{k,1}+\delta_{k,i1},x_{k,2}+\delta_{k,i2}\right]^\text{T}\right),
\vspace{-1.2mm}
\end{equation}
where $\widetilde{\boldsymbol{\Delta}}_{k,i}$ is the direction parameter offset corresponding to $\textbf{w}_{k,i}$. After phase shifting and combining, the observation at the baseband output of RF chain is given by
\vspace{-1mm}
\begin{equation}\label{eq_observation}
y_{k,i} = \textbf{w}_{k,i}^\text{H}\textbf{h}(\textbf{x}_k)s_p + z_{k,i}=s_p \beta_k \textbf{w}_{k,i}^\text{H} \textbf{a}(\textbf{x}_k)+z_{k,i},
\vspace{-1mm}
\end{equation}
$\boldsymbol{\psi}_k \triangleq \big[\beta_k^\text{re},\beta_k^\text{im},x_{k,1},
x_{k,2}\big]^\text{T}$
where ${z}_{k,i}\sim\mathcal{CN}(0,\sigma^2)$ is an \emph{i.i.d.} circularly symmetric complex Gaussian random variable. Define $\boldsymbol{\psi}_k \triangleq \big[\beta_k^\text{re},\beta_k^\text{im},x_{k,1},$
$x_{k,2}\big]^\text{T}$ as the channel parameter vector in time-slot $k$, $\textbf{W}_k\triangleq \left[\textbf{w}_{k,1},\ldots,\textbf{w}_{k,q}\right]$ as the analog beamforming matrix, and $\textbf{z}_k \triangleq \left[z_{k,1},\ldots,z_{k,q}\right]$ as the noise vector. Then the conditional probability density function of the observation vector $\textbf{y}_k\triangleq \left[y_{k,1},\ldots,y_{k,q}\right]^\text{T}$\! is given by
\vspace{-1.5mm}
\begin{equation}\label{eq_pdf}
p(\textbf{y}_k| \boldsymbol{\psi}_k, \textbf{W}_k) = {\frac{1}{\pi^{q} \sigma^{2q}} e^{- \frac {{\left\| \textbf{y}_k-s_p \beta_l \textbf{W}_k^\text{H} \textbf{a}(\textbf{x}) \right\|}_2^2} {\sigma^2}}}.
\vspace{-1.5mm}
\end{equation}

%$\hat {\boldsymbol{\psi}}_k \triangleq \big[\hat{\beta}_k^\text{re},\hat{\beta}_k^\text{im},\hat x_{k,1},\hat {x}_{k,2}\big]^\text{T}$

In time-slot $k$, the receiver needs to choose an analog beamforming matrix $\textbf{W}_k$ and obtain an estimate $\hat {\boldsymbol{\psi}}_k \triangleq \big[\hat{\beta}_k^\text{re},\hat{\beta}_k^\text{im},$
$\hat x_{k,1},\hat {x}_{k,2}\big]^\text{T}$ of the channel parameter vector $\boldsymbol{\Psi}_k$. From a control system perspective, ${\boldsymbol{\psi}}_k$ is the system state, $\hat{{\boldsymbol{\psi}}}_k$ is the estimate of the system state, the analog beamforming matrix $\textbf{W}_k$ is the control action and $\textbf{y}_k$ is a non-linear noisy observation determined by the system state and control action.
%satisfy the following conditions: (i)~The information that is available for determining the beamforming vector $\textbf{w}_n$ includes the history of received signals $(y_i: i = 1, \ldots, n-1)$ and the history of beamforming vectors $(\textbf{w}_i: i = 1, \ldots, n-1)$. (ii)~The beamforming vectors $\textbf{w}_1, \textbf{w}_2, \ldots$ satisfy (\ref{eq_bf}).

%The goal of this paper is to develop an efficient beam tracking algorithm that can achieve high accuracy with a low pilot overhead.

\vspace{-1mm}
\section{Problem Formulation and Optimal Beamforming Matrix}\label{sec_problem}
%%In this section, we first formulate the problem and give the lower bound of it. {\color{blue}Then we develop a general way to generate the optimal beamforming vectors and obtain the minimum CRLB.}
\subsection{Problem Formulation}\label{sbsec_ProblemBound}

Let $\zeta = (\textbf{W}_1, \textbf{W}_2, \ldots, \hat {\boldsymbol{\psi}}_1, \hat {\boldsymbol{\psi}}_2, \ldots)$ denote a beam and channel tracking scheme. We consider a particular set $\Xi$ of \emph{causal} beam tracking policies: in time-slot $k$, the analog beamforming matrix $\textbf{W}_k$ and estimate $\hat {\boldsymbol{\psi}}_k$ are based on the previously used analog beamforming matrix $\textbf{W}_1,\cdots, \textbf{W}_{k-1}$ and historical observations $\textbf{y}_1,\cdots, \textbf{y}_{k-1}$. Hence, in $k$-th time-slot, the beam and channel tracking problem is formulated as:
\vspace{-1.5mm}
\begin{align}\label{eq_problem}
  \underset{\zeta \in \Xi}{\min} ~& \frac {1}{MN} \,\mathbb{E} \left[{\left\|\hat{\textbf{h}}_k - \textbf{h}_k\right\|}_2^2 \right] \\
\label{eq_constrant1} \text{s.t.} ~& \mathbb{E}\left[\hat{\textbf{h}}_k\right] = \textbf{h}_k,\\
~&\eqref{eq_channelvector}-\eqref{eq_observation}, \nonumber
\vspace{-1.5mm}
\end{align}where the constraint \eqref{eq_constrant1} ensures that $\hat{\textbf{h}}_k \triangleq \hat {\beta}_k \textbf{a} \left(\hat {\textbf{x}}_k\right)$ is an unbiased estimation of the channel vector $\textbf{h}_k = \beta_k \textbf{a}\left(\textbf{x}_{k}\right)$ and the constraints \eqref{eq_channelvector}-\eqref{eq_observation} ensure the steering vector form of analog beamforming vectors.

Problem \eqref{eq_problem} is difficult to solve optimally due to several reasons: (i) it is a constrained partially observed Markov decision process (C-POMDP) that is usually quite difficult to solve. (ii) The analog beamforming matrix $\textbf{W}_k$ and the estimate $\hat{\boldsymbol{\psi}}_k$ need to be optimized. However, both the optimization of $\textbf{W}_{k}$ and $\hat{\boldsymbol{\psi}}_k$ are non-convex problems.

Before giving some theoretical results of problem \eqref{eq_problem}, we will first study the pilot overhead needed for beam and channel tracking in 2D phased antenna arrays.
%{\red which introduces multiple stable equilibrium points to the tracking .}
%\begin{align}\label{eq_problem}
%	& \underset{\begin{matrix}F_n(\cdot), W_{mn}(\cdot)  \\ m = 1,\ldots, M \end{matrix}}{\min} ~\mathbb{E}\left[ \left( \hat{x}_{n} - x \right)^2 \right] \\
%	&~~~~~~~\text{s.t.}\ \ \hat{x}_{n} = F_n(y_1, \ldots, y_n, \textbf{w}_1, \ldots, \textbf{w}_n), \mathbb{E}\left[ \hat{x}_{n} \right] = x, \nonumber \\
%	&~~~~~~~~~~~~\textbf{w}_n = \frac{1}{\sqrt{M}}\left[ e^{jw_{1n}}, e^{jw_{2n}}, \ldots, e^{jw_{Mn}} \right]^\text{H}, \nonumber \\
%&~~~~~~~~~~~~w_{mn} = W_{mn}(y_1, \ldots, y_{n-1}, \textbf{w}_1, \ldots, \textbf{w}_{n-1}), \nonumber \\
%&~~~~~~~~~~~~m = 1,\ldots, M. \nonumber
%\end{align}
% which is different from the traditional MMSE problem in full-digital MIMO systems \cite{Biguesh2006Training}, since only one combined signal $y_n$ is obtained in each time-slot, which is dependent on the analog beamforming vector $\textbf{w}_n$. Therefore, it is a joint optimization problem of both the analog beamforming vector $\textbf{w}_n$ and the estimate $\hat{x}_{n}$. Moreover, due to the constant-modulus constraint in (\ref{eq_constraint}), both the optimization of the analog beamforming vector $\textbf{w}_n$ and the optimization of the estimate $\hat{x}_{n}$ are non-convex \cite{Rial2016Hybrid,Heath2016overview, zhang2016mobile,palacios2016tracking}.
%We derive the optimum MSE that can be achieved for this problem.
\vspace{-1mm}
\subsection{How Many Pilots Are Needed? }\label{sbsec_LeastPilot}
%Since $\textbf{W}_k = \left[\textbf{w}_{k,1},\textbf{w}_{k,2},\ldots,\textbf{w}_{k,q}\right]$ is composed of $q$ beamforming vectors, we should determine the pilot number $q$ first before optimizing $\textbf{W}$ to get the minimum CRLB in \eqref{eq_CMMSE}.
%With one observation by using one trial beamforming vector, we can obtain two real number equations. When tracking the normalized beam direction $\textbf{x}_k$ and channel coefficient $\beta_k$ simultaneously, four real variables need to be estimated and four independent real number equations are required in each time-slot.
\vspace{-0.5mm}
According to \cite{JLiJoint2018}, two pilots in each time-slot are sufficient to jointly track the channel coefficient and 1D beam direction. When tracking the horizontal and vertical beam direction
simultaneously, four pilots  are feasible by separately using two pilots to track each dimension of the 2D beam direction. However, with four pilots, the channel coefficient is updated twice in each time-slot, possibly leading to redundancy. Hence, we can jointly track channel coefficient and 2D beam direction to further reduce pilot overhead.

When tracking the channel parameters jointly, four real variables (i.e., the real part $\beta_k^{\text{re}}$ and imaginary part $\beta_k^{\text{im}}$ of channel coefficient $\beta_k$ and the two direction parameters ${x}_{k,1}, {x}_{k,2}$) need to be estimated. Then the following lemma is proposed to help determine the smallest $q$:
\begin{lemma}\label{IndObservations}
\emph{If the analog beamforming vectors are steering vectors, i.e.,} $\textbf{w}_{k,i} \!=\!\! \frac{1}{\sqrt{MN}}\textbf{a}\!\left(\textbf{x}\!+\!{\boldsymbol{\Delta}}_{k,i}\right)$\emph{, then at least} $q$ \emph{observations are needed  to estimate} $q+1$ \emph{real variables in time-slot} $k$.
\end{lemma}
\begin{proof}
See Appendix \ref{proof_IndObservations}.
\end{proof}

Lemma \ref{IndObservations} tells us at least three observations are required in each time-slot to estimate four real variables. Hence, the smallest pilot number in each time-slot is $q = 3$, i.e., the analog beamforming matrix $\textbf{W}_k = \left[\textbf{w}_{k,1},\textbf{w}_{k,2},\textbf{w}_{k,3}\right]$.

\subsection{Lower Bound of Tracking Error}\label{sbsec_Bound}
\vspace{-0.5mm}
The huge challenge to solve problem \eqref{eq_problem} optimally makes it hard to complete in just one paper. Therefore, we perform some theoretical analysis for static scenarios as the first step in this paper.

Consider the problem of tracking a static beam, where $ {\boldsymbol{\psi}}_k\! = \!{\boldsymbol{\psi}}\! \triangleq \! \left[\beta^\text{re},\beta^\text{im},x_{1},x_{2}\right]^\text{T}$ for all time-slots. The Cram\'{e}r-Rao lower bound theory gives the lower bound of the unbiased estimation error according to \cite{Sengijpta1993Fundamental}. Based on this, we introduce the following lemma to obtain the lower bound of tracking error:
\begin{lemma}\label{MSEOpt}
\emph{The MSE of channel vector in \eqref{eq_problem} is lower bounded as follows:}
\vspace{-1.5mm}
\begin{align}\label{MSELB}
&~\frac{1}{{MN}} \mathbb{E}\left[\left\| \hat{\textbf{h}}_k - \textbf{h}_k \right\|_2^2 \right]\\
 \ge&~ \frac{1}{{MN}}\Tr \left\{ {{{\left(\sum\limits_{l = 1}^k {{\bf{I}}(\boldsymbol{\psi} ,{{\bf{W}}_l})} \right)}^{-1}}\sum\limits_{m = 1}^M {\sum\limits_{n = 1}^N {\left({\textbf{v}_{m,n}^\text{H}}{\textbf{v}_{m,n}}\right)} } } \right\},\nonumber
\end{align}
\vspace{-0mm}\emph{where} $\textbf{v}_{m,n} \triangleq \left[1,j,{j 2 \pi \frac {m-1}{M} \beta}, {j 2 \pi \frac {n-1}{N} \beta} \right]$ \emph{and the Fisher information matrix} $\textbf{I}(\boldsymbol{\psi}, \textbf{W}_k)$ \emph{is given by}
\vspace{-1.5mm}
\begin{small}
\begin{align}\label{eq_fisher}
&\!\textbf{I}(\boldsymbol{\psi}, \textbf{W}_k)\triangleq  \!- \mathbb{E}\left[\frac {\partial \text{log} \, p \left(\textbf{y}_k |\boldsymbol{\psi},\textbf{W}_k \right)}{\partial \boldsymbol{\psi}} \cdot \frac {\partial \text{log}\, p \left(\textbf{y}_k |\boldsymbol{\psi},\textbf{W}_k \right)}{\partial \boldsymbol{\psi}^\text{T}}\right]\\
=&\frac {2 {\lvert s_p \rvert}^2} {{\sigma}^2}\!
\left[
\begin{matrix}
{{\lVert \textbf{g}_k \rVert}_2^2} & 0 & {\text{Re} \{ {\textbf{g}}_k^\text{H} \tilde {\textbf{g}}_{k1}\} } & {\text{Re}  \{ {\textbf{g}}_k^\text{H} \tilde {\textbf{g}}_{k2}\} }\\
0 & \!{{\lVert \textbf{g}_k \rVert}_2^2} & {\text{Im} \{ \textbf{g}_k^\text{H} \tilde {\textbf{g}}_{k1} \}} & {\text{Im} \{ \textbf{g}_k^\text{H} \tilde {\textbf{g}}_{k2} \}}\\
{\text{Re} \{ \textbf{g}_k^\text{H} \tilde {\textbf{g}}_{k1} \}} & \!{\text{Im} \{ \textbf{g}_k^\text{H} \tilde {\textbf{g}}_{k1} \}} &  {\lVert \tilde {\textbf{g}}_{k1} \rVert}_2^2 & {\text{Re} \{{\tilde {\textbf{g}}_{k1}}^\text{H}{\tilde {\textbf{g}}_{k2} \}}}\\
{\text{Re} \{ \textbf{g}_k^\text{H} \tilde {\textbf{g}}_{k2} \}} & {\text{Im} \{ \textbf{g}_k^\text{H} \tilde {\textbf{g}}_{k2} \}} & \!{\text{Re} \{{\tilde {\textbf{g}}_{k1}}^\text{H}{\tilde {\textbf{g}}_{k2}} \}} & {\lVert \tilde {\textbf{g}}_{k2} \rVert}_2^2
\end{matrix}
\right],\nonumber
\vspace{-1.5mm}
\end{align}
\end{small}with $\textbf{g}_k\!\!=\!\!\textbf{W}_k^\text{H}\textbf{a}\left(\textbf{x}\right)$, $\tilde {\textbf{g}}_{k1}\!\!=\!\!\beta \textbf{W}_k^\text{H}\frac {\partial \textbf{a}\left(\textbf{x}\right)}{\partial x_1}$, and $\tilde {\textbf{g}}_{k2}\!\!=\!\!\beta \textbf{W}_k^\text{H}\frac {\partial \textbf{a}\left(\textbf{x}\right)}{\partial x_2}$.
\end{lemma}
\begin{proof}
See Appendix \ref{proof_MSEOpt}.
\end{proof}

The CRLB in \eqref{MSELB} is a function of the analog beamforming matrices $\textbf{W}_1,\ldots,\textbf{W}_k$. It is hard to optimize so many beamforming matrices simultaneously. Suppose that $\textbf{W}_1=\textbf{W}_2=\ldots=\textbf{W}_k$. Then we can get the \emph{minimum CRLB} under this constriant, given by
\vspace{-1.5mm}
\begin{small}
\begin{align}\label{eq_CMMSE}
\!{I}_{\min}(\boldsymbol{\psi})=&\min_{\textbf{W}_1,\ldots,\textbf{W}_k}\!\frac{1}{{MN}}\!\!\Tr\!\! \left\{\!\! {{{\!\left(\sum\limits_{l = 1}^k {{\bf{I}}(\boldsymbol{\psi} ,{{\bf{W}}_l})} \!\!\right)}^{\!\!-1}}\!\!\!\sum\limits_{m = 1}^M\! {\sum\limits_{n = 1}^N \!\!{\left(\!{\textbf{v}_{m,n}^\text{H}\!}{\textbf{v}_{m,n}}\!\!\right)} } } \!\!\!\right\} \nonumber\\
=& \min_{\textbf{W}}\!\frac{1}{{MN}}\!\Tr\! \left\{\! {{{\left( k\textbf{I} \left( \boldsymbol{\psi}, \textbf{W} \right) \right)}^{\! - 1}}\!\sum\limits_{m = 1}^M {\sum\limits_{n = 1}^N {{\!\left({\textbf{v}_{m,n}^\text{H}\!}{\textbf{v}_{m,n}}\!\!\right)} } } } \!\!\right\}\!\!.\!\!\!
%\overset{(b)}{=}& \min_{\textbf{W}} u\left(\boldsymbol{\psi},\textbf{W}\right),%=&u\left(\boldsymbol{\psi},\textbf{W}^*\right),\nonumber
\vspace{-1.5mm}
\end{align}
\end{small}Solving problem \eqref{eq_CMMSE} yields the optimal analog beamforming matrix $\textbf{W}^*=\left[\textbf{w}_1^*,\textbf{w}_2^*,\textbf{w}_3^*\right]$:
\vspace{-1.5mm}
\begin{equation}\label{eq_bfo}
\textbf{w}_{i}^* = \frac{1}{\sqrt{MN}}\textbf{a}\left(\textbf{x}+ \boldsymbol{\Delta}_{i}^*\right),i=1,2,3,
\vspace{-1.5mm}
\end{equation}where $\boldsymbol{\Delta}_{1}^*,\boldsymbol{\Delta}_{2}^*,\boldsymbol{\Delta}_{3}^*$ denote the optimal direction parameter offsets. Hence, let $\textbf{W}_1^{*}=\textbf{W}_2^{*}=\cdots\textbf{W}_k^{*}=\textbf{W}^{*}$ and we can obtain the minimum CRLB by \eqref{eq_CMMSE}.

\subsection{Asymptotically Optimal Analog Beamforming Matrix}\label{sbsec_OptBeamforming}
\vspace{-0.5mm}
Let us consider the optimal analog beamforming matrix $\textbf{W}^*$. In \eqref{eq_CMMSE}, three 2D direction parameter offsets need to be optimized. It is hard to get analytical results for such a six-dimensional non-convex problem. Numerical search is a feasible way to handle the problem. However, these optimal offsets may be related to some system parameters, e.g., channel coefficient $\beta$, direction parameter vector $\textbf{x}$ and antenna array size $M,\,N$. Once these system parameters change, numerical search has to be re-conducted, leading to high complexity. To overcome this challenge, we explore the properties of $\boldsymbol{\Delta}_{1}^{*},\boldsymbol{\Delta}_{2}^{*},\boldsymbol{\Delta}_{3}^{*}$ and obtain the following lemma:
\begin{lemma}\label{UnifiedOptShift}
\emph{The optimal direction parameter offsets $\boldsymbol{\Delta}_{1}^{*},\boldsymbol{\Delta}_{2}^{*},\boldsymbol{\Delta}_{3}^{*}$ have the following three properties}:

\emph{1)} $\boldsymbol{\Delta}_{1}^{*},\boldsymbol{\Delta}_{2}^{*},\boldsymbol{\Delta}_{3}^{*}$ \emph{are invariant to the channel coefficient} $\beta$;

\emph{2)} $\boldsymbol{\Delta}_{1}^{*},\boldsymbol{\Delta}_{2}^{*},\boldsymbol{\Delta}_{3}^{*}$ \emph{are invariant to the direction parameter vector} $\textbf{x}$;

\emph{3)} $\boldsymbol{\Delta}_{1}^{*},\boldsymbol{\Delta}_{2}^{*},\boldsymbol{\Delta}_{3}^{*}$ \emph{converge to constant values as} $\emph{M},\,\emph{N} \to +\infty$:
\vspace{-1mm}
\begin{align}\label{eq_MainLobe}
\widetilde{\boldsymbol{\Delta}}_{i}^{*} \triangleq {\lim\limits_{M,N \to +\infty}}\boldsymbol{\Delta}_{i}^{*},i=1,2,3. \nonumber
\end{align}
\vspace{-3mm}
\end{lemma}
\begin{proof}
See Appendix \ref{proof_UnifiedOptShift}.
\end{proof}

Lemma \ref{UnifiedOptShift} reveals that $\boldsymbol{\Delta}_{1}^{*},\boldsymbol{\Delta}_{2}^{*},\boldsymbol{\Delta}_{3}^{*}$ are only related to array size $M,\,N$. Hence, the numerical search complexity can be reduced to one for a particular array size $M,\,N$. Even if $\boldsymbol{\Delta}_{1}^{*},\boldsymbol{\Delta}_{2}^{*},\boldsymbol{\Delta}_{3}^{*}$ may change for different array sizes, we can adopt $\widetilde{\boldsymbol{\Delta}}_{1}^{*},\,\widetilde{\boldsymbol{\Delta}}_{2}^{*},\widetilde{\boldsymbol{\Delta}}_{3}^{*}$ to take the place of $\boldsymbol{\Delta}_{1}^{*},\boldsymbol{\Delta}_{2}^{*},\boldsymbol{\Delta}_{3}^{*}$ as long as $M$ and $N$ are sufficiently large. Therefore, the numerical search times are reduced to one.

By numerical search in the main lobe of the direction parameter vector:
\vspace{-1mm}
\begin{equation}\label{eq_MainLobe}
\begin{aligned} \mathcal{B}\left(\textbf{x}\right) \triangleq  \left(x_1-1,x_1+1\right) \times \left(x_2-1,x_2+1\right),
\end{aligned}
\vspace{-1mm}
\end{equation}we can obtain the asymptotically optimal direction parameter offsets $\widetilde{\boldsymbol{\Delta}}_{1}^{*},\,\widetilde{\boldsymbol{\Delta}}_{2}^{*},\widetilde{\boldsymbol{\Delta}}_{3}^{*}$ in TABLE \ref{Tab_asymptotically optimal_parameters} and Fig. \ref{fig_offsets}.
\renewcommand\arraystretch{1.4}
\begin{table}[!t!]
% \vspace{-5mm}
\centering
\caption{Asymptotically optimal offsets.}
\vspace{-3mm}
\begin{tabular} {c|c|c}
\hline
\hline
%$\tilde{\delta}_{k,11}^{*}$ & $\tilde{\delta}_{k,12}^{*}$ & $\tilde{\delta}_{k,21}^{*}$ & $\tilde{\delta}_{k,22}^{*}$ & $\tilde{\delta}_{k,31}^{*}$ & $\tilde{\delta}_{k,32}^{*}$ \\
$\widetilde{\boldsymbol{\Delta}}_{1}^{*}$ & $\widetilde{\boldsymbol{\Delta}}_{2}^{*}$ & $\widetilde{\boldsymbol{\Delta}}_{3}^{*}$\\
\hline
$\left[0.0963, 0.5098\right]^\text{T}$ & $\left[-0.5098,-0.0963\right]^\text{T}$ & $\left[0.2906,-0.2906\right]^\text{T}$\\
\hline
\hline
\end{tabular}\label{Tab_asymptotically optimal_parameters}
\vspace{-4mm}
\end{table}
\begin{figure}[!t!]
% \vspace{-5mm}
\centering
\includegraphics[width=7.5cm]{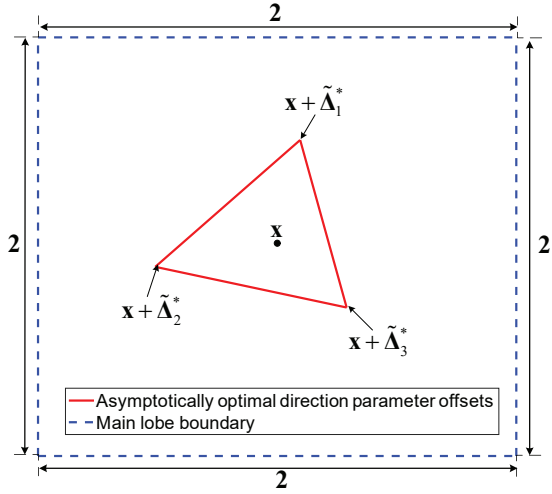}
\vspace{-4mm}
\caption{Asymptotically optimal offsets.}
\vspace{-2mm}
\label{fig_offsets}
\end{figure}With these offsets, a general way to generate the asymptotically optimal analog beamforming matrix $\widetilde{\textbf{W}}_k^{*}=[\tilde{\textbf{w}}_{k,1}^{*},\tilde{\textbf{w}}_{k,2}^{*},\tilde{\textbf{w}}_{k,3}^{*}]$ is obtained to achieve the minimum CRLB as below:
\vspace{-1mm}
\begin{equation}\label{eq_obf}
\tilde{\textbf{w}}_{k,i}^{*} = \frac{1}{\sqrt{MN}}\textbf{a}\left(\textbf{x}+\widetilde{\boldsymbol{\Delta}}_{i}^{*}\right),i=1,2,3.
\end{equation}
\vspace{-1mm}

By adopting $\widetilde{\boldsymbol{\Delta}}_{1}^{*},\,\widetilde{\boldsymbol{\Delta}}_{2}^{*},\widetilde{\boldsymbol{\Delta}}_{3}^{*}$ to smaller size antenna arrays, we compare the minimum CRLB and the CRLB corresponding to $\widetilde{\boldsymbol{\Delta}}_{1}^{*},\,\widetilde{\boldsymbol{\Delta}}_{2}^{*},\widetilde{\boldsymbol{\Delta}}_{3}^{*}$ in TABLE \ref{Tab_asymptotically optimal_parameters}.
\begin{figure}[!t]
\centering
\includegraphics[width=3.33in]{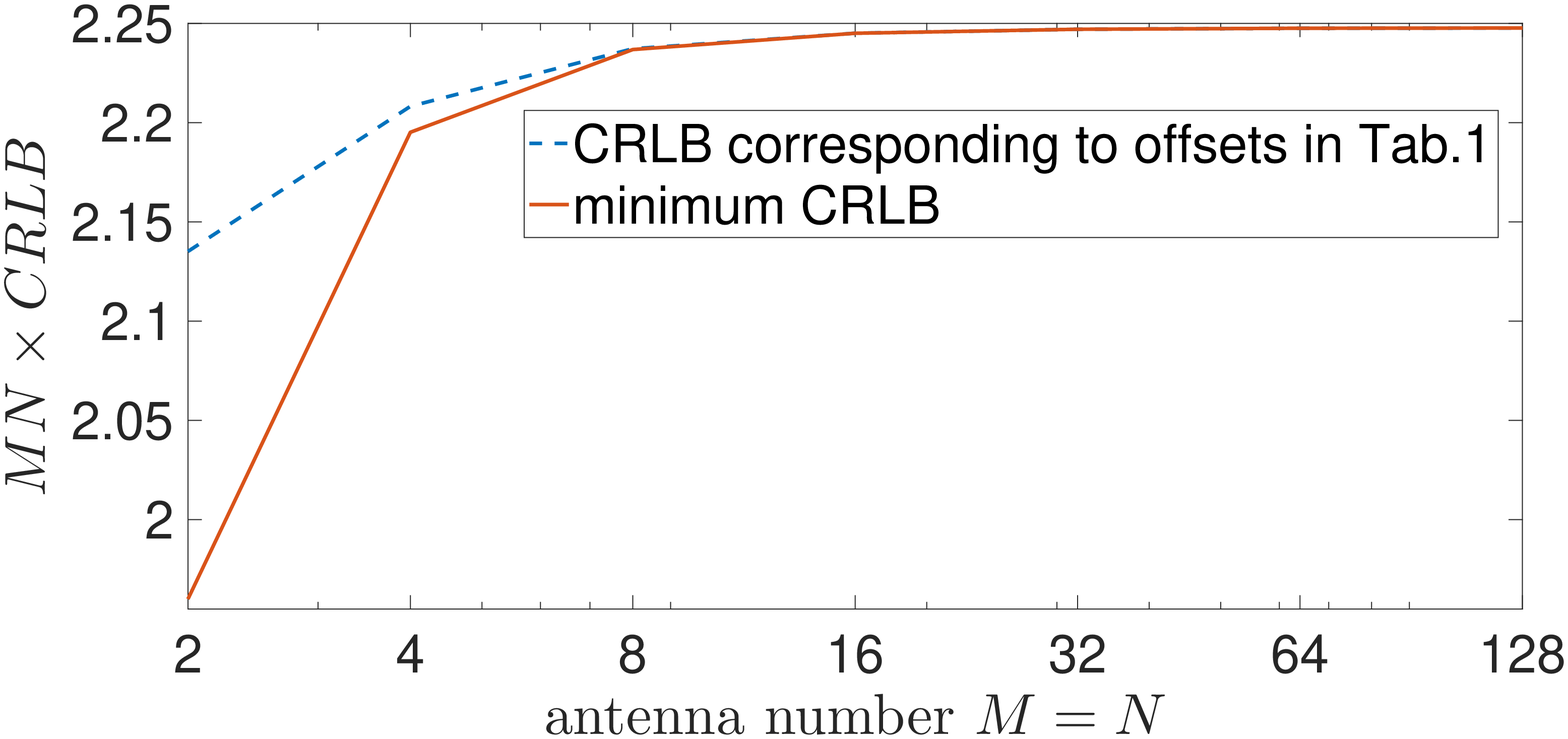}
\vspace{-3mm}
\caption{Performance of offsets in TABLE \ref{Tab_asymptotically optimal_parameters}}
\vspace{-5mm}
\label{fig_Extensibility.eps}
\end{figure}As illustrated in Fig. \ref{fig_Extensibility.eps}, when antenna number $M=N \geq 8$, we can approach the minimum CRLB with a relative error less than $0.1\%$ by using $\widetilde{\boldsymbol{\Delta}}_{1}^{*},\,\widetilde{\boldsymbol{\Delta}}_{2}^{*},\widetilde{\boldsymbol{\Delta}}_{3}^{*}$. Therefore, with $\widetilde{\boldsymbol{\Delta}}_{1}^{*},\,\widetilde{\boldsymbol{\Delta}}_{2}^{*},\widetilde{\boldsymbol{\Delta}}_{3}^{*}$, the minimum CRLB is obtained
for different beam directions, different channel coefficients and different antenna numbers when $M=N \geq 8$.
\section{Asymptotically Optimal Joint Beam and Channel Tracking}\label{sec_algorithm&AsymptoticallyOptimality}
\subsection{Joint Beam and Channel Tracking}\label{sbsec_Tracking}
\vspace{-0.5mm}
The proposed tracking algorithm is similar to that in \cite{JLiJoint2018}. The main difference is that we need $M \times N$ pilots to estimate the initial direction parameter offsets and three analog beamforming vectors to track the time-varying beam direction.
\\
\textbf{Joint Beam and Channel Tracking}:
% which consists of two stages:
%Recursive Analog Beam Tracking Algorithm:}
\begin{itemize}[leftmargin=*]
	\item[1)]\emph{\textbf{Coarse Beam Sweeping:}}
	%\end{itemize} %Use the exhaustive beam sweeping algorithm \cite{Hur2013Millimeter}.
	As shown in Fig. \ref{FrameStrcucture}, $M \times N$ pilots are received successively. The analog beamforming vector corresponding to the observation $\check{y}_{m,n}$ is  $\check{\textbf{w}}_{m,n}=\frac{1}{\sqrt{MN}} \textbf{a}\left(\left[\left(2m-1-M\right)\frac{d_1}{\lambda},\left(2n-1-N\right)\frac{d_2}{\lambda}\right]^\text{T}\right),m=1,\cdots ,$ $M,n=1,\cdots,N$. The initial estimate $\hat {\boldsymbol{\psi}}_0 = \big[\hat{\beta}_0^\text{re},\hat{\beta}_0^\text{im},\hat{x}_{0,1},$ $\hat{x}_{0,2}\big]^\text{T}$ is obtained by:
	\vspace{-3mm}
	\begin{equation}\label{eq_RBCE}
	\begin{aligned}  \hat{\textbf{x}}_0 =  \underset{\hat{\textbf{x}} \in \chi}{\text{argmax}} \lvert \textbf{a}\left(\hat{\textbf{x}}\right)^\text{H} \check {\textbf{W}} \check {\textbf{y}} \rvert , \quad  \hat {\beta}_0 = \left[\check {\textbf{W}}^\text{H} \textbf{a}\left(\hat{\textbf{x}}_0\right) \right]^{+} \check {\textbf{y}},
	\end{aligned}
	\vspace{-3mm}
	\end{equation}
	%\begin{itemize}
	where\begin{small}
		$\chi \!=\! \!\left\{\!\left[\frac{\left(2m-1-M_0\right)M d_1}{\lambda M_0},\frac{\left(2n-1-N_0\right)N d_2}{\lambda N_0}\!\right]^\text{T}\!\bigg|\!\begin{smallmatrix}\begin{aligned}& m\!=\!1,\ldots ,M_0\\& n\!=\!1,\ldots,N_0
		\end{aligned}
		\end{smallmatrix}\!\!\right\}$\end{small}, $M_0 \times N_0$ is the codebook size with $M_0 \geq M$ and $N_0 \geq N$, $\check {\textbf{y}}=\left[\check{y}_{11},\check{y}_{12}\cdots ,\check{y}_{MN}\right]^\text{T}$, $\check {\textbf{W}}=\left[\check{\textbf{w}}_{11},\check{\textbf{w}}_{12},\cdots,\check{\textbf{w}}_{MN}\right]$, and $\textbf{X}^\text{+}=\left(\textbf{X}^\text{H} \textbf{X}\right)^{-1} \textbf{X}^\text{H}$.
	%\end{itemize}
	%\begin{itemize}
	\begin{figure}[!t]
		\centering
		\vspace{-0mm}
		\includegraphics[width=3.33in]{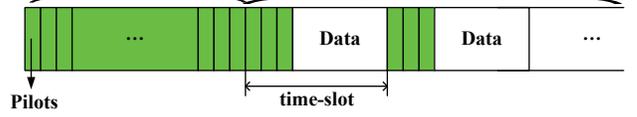}
		\vspace{-3mm}
		\caption{Frame structure.}
		\vspace{-5mm}
		\label{FrameStrcucture}
	\end{figure}
	\item[2)] \emph{\textbf{Beam and channel tracking:}}
	In time-slot $k$, three pilots are received by using analog beamforming vectors given below:
	\vspace{-2mm}
	\begin{equation}\label{eq_BF}
	\begin{aligned}\textbf{w}_{k,i} = \frac {1}{\sqrt{MN}} \textbf{a} \left(\hat{\textbf{x}}_{k-1} + \widetilde{\boldsymbol{\Delta}}_{i}^{*}\right), i =1,2,3,
	\end{aligned}
	\vspace{-2mm}
	\end{equation}where $\hat{\textbf{x}}_{k} \triangleq \left[\hat{x}_{k,1},\hat{x}_{k,2}\right]^\text{T}$ and $\widetilde{\boldsymbol{\Delta}}_{i}^{*}\,(i=1,2,3)$ are given by TABLE \ref{Tab_asymptotically optimal_parameters}. The estimate $\hat {\boldsymbol{\psi}}_k = \left[\hat{\beta}_k^\text{re},\hat{\beta}_k^\text{im},\hat x_{k,1},\hat {x}_{k,2}\right]$ is updated by
	\begin{small}
		\vspace{-2mm}
		\begin{equation}\label{eq_Tracking}
		\begin{aligned} \!\hat {\boldsymbol{\psi}}_k \!=\! \hat {\boldsymbol{\psi}}_{k-1}\!+\! \frac{2}{\sigma^2} b_k \textbf{I}\left(\!\hat{\boldsymbol{\psi}}_{k\!-\!1}, \textbf{W}_k\!\right)^\text{\!-1} \!\left[\!
		\begin{matrix}
		{\text{Re}\left\{s_p^\text{H} \textbf{e}_k^\text{H}\left(\textbf{y}_k\!-\!\hat{\textbf{y}}_k \right)\right\}}\\
		{\text{Im}\left\{s_p^\text{H} \textbf{e}_k^\text{H}\left(\textbf{y}_k\!-\!\hat{\textbf{y}}_k \right)\right\}}\\
		{\text{Re}\left\{s_p^\text{H} \tilde{\textbf{e}}_{k1}^\text{H}\left(\textbf{y}_k\!-\!\hat{\textbf{y}}_k \right)\right\}}\\
		{\text{Re}\left\{s_p^\text{H} \tilde{\textbf{e}}_{k2}^\text{H}\left(\textbf{y}_k\!-\!\hat{\textbf{y}}_k \right)\right\}}
		\end{matrix}
		\!\!\right]\!,\!\!
		\end{aligned}
		\vspace{-2mm}
		\end{equation}
	\end{small}where ${\textbf{e}}_k = \textbf{W}_k^\text{H} \textbf{a}\left(\hat{\textbf{x}}_{k-1}\right)$, $\hat{\textbf{y}}_k = s_p \hat{\beta}_{k-1}\textbf{W}_k^\text{H} \textbf{a}\left(\hat{\textbf{x}}_{k-1}\right)$, $\tilde{\textbf{e}}_{k1} = \hat{\beta}_{k-1} \textbf{W}_k^\text{H} \frac{\partial \textbf{a}\left(\hat{\textbf{x}}_{k-1}\right)}{\partial x_1}$ and $\tilde{\textbf{e}}_{k2} = \hat{\beta}_{k-1} \textbf{W}_k^\text{H} \frac{\partial \textbf{a}\left(\hat{\textbf{x}}_{k-1}\right)}{\partial x_2}$. Here, $b_k$ is the step size and will be specified later.
\end{itemize}

\subsection{Asymptotic Optimality Analysis}\label{sbsec_asy_opt}
In the tracking procedure \eqref{eq_Tracking}, there exist multiple stable points and these stable points correspond to the local optimal points for our proposed algorithm. To study these stable points, we rewrite \eqref{eq_Tracking} as \eqref{eq_rtracking}:
\vspace{-1mm}
\begin{equation}\label{eq_rtracking}
\hat {\boldsymbol{\psi}}_k = \hat {\boldsymbol{\psi}}_{k-1}+  b_k \left(\textbf{f} \left(\hat {\boldsymbol{\psi}}_{k-1},{\boldsymbol{\psi}}_{k}\right) + \hat{\textbf{z}}_k \right),
\vspace{-1mm}
\end{equation}
where $\textbf{f} \left(\hat {\boldsymbol{\psi}}_{k-1},{\boldsymbol{\psi}}_{k}\right)$  is defined as follows:
\begin{small}
	\begin{align}\label{eq_f}
	&\textbf{f} \left(\hat {\boldsymbol{\psi}}_{k-1},{\boldsymbol{\psi}}_{k}\right)  \triangleq \mathbb{E} \left[\textbf{I}\left(\hat{\boldsymbol{\psi}}_{k-1}, \textbf{W}_k\right)^\text{-1} \frac{\partial \text{log} \, p \left( \textbf{y}_k \mid \hat{\boldsymbol{\psi}}_{k-1},\textbf{W}_k \right)}{\partial \hat {\boldsymbol{\psi}}_{k-1}} \right]\nonumber\\
	=& \frac{2 {\lvert s_p \rvert}^2}{\sigma^2}\textbf{I}\!\left(\!\hat{\boldsymbol{\psi}}_{k\!-\!1}, \textbf{W}_k\!\right)^\text{\!-1}
	\!\left[
	\begin{matrix}
	{\text{Re}\!\left\{\textbf{e}_k^\text{H}\left(\beta_k \textbf{W}_k^\text{H} \textbf{a}\left(\textbf{x}_k\right) \!-\!\hat{\beta}_{k\!-\!1} \textbf{e}_k \right)\!\right\}}\\
	{\text{Im}\!\left\{\textbf{e}_k^\text{H}\!\left(\beta_k \textbf{W}_k^\text{H} \textbf{a}\left(\textbf{x}_k\right)\!-\!\hat{\beta}_{k\!-\!1} \textbf{e}_k \right)\!\right\}}\\
	{\text{Re}\!\left\{\tilde{\textbf{e}}_{k1}^\text{H}\!\left(\beta_k \textbf{W}_k^\text{H} \textbf{a}\left(\textbf{x}_k\right)\!-\!\hat{\beta}_{k\!-\!1} \textbf{e}_k \right)\!\right\}}\\
	{\text{Re}\!\left\{\tilde{\textbf{e}}_{k2}^\text{H}\!\left(\beta_k \textbf{W}_k^\text{H} \textbf{a}\left(\textbf{x}_k\right)\!-\!\hat{\beta}_{k\!-\!1} \textbf{e}_k \right)\!\right\}}
	\end{matrix}
	\!\right]\!,\!
	\end{align}
	\vspace{-1mm}
\end{small}and $\hat{\textbf{z}}_k$ is given by
\vspace{-1mm}
\begin{small}
	\begin{align}\label{eq_z}
	\hat{\textbf{z}}_k & \triangleq \textbf{I}\left(\hat{\boldsymbol{\psi}}_{k-1}, \textbf{W}_k\right)^\text{-1} \frac{\partial \text{log} \, p \left( \textbf{y}_k \mid \hat{\boldsymbol{\psi}}_{k-1},\textbf{W}_k \right)}{\partial \hat {\boldsymbol{\psi}}_{k-1}} - \textbf{f} \left(\hat {\boldsymbol{\psi}}_{k-1},{\boldsymbol{\psi}}_{k}\right)\nonumber\\
	&= \frac{2}{\sigma^2}\textbf{I}\left(\hat{\boldsymbol{\psi}}_{k-1}, \textbf{W}_k\right)^\text{-1}
	\left[
	\begin{matrix}
	{\text{Re}\left\{s_p^\text{H} \textbf{e}_k^\text{H} \textbf{z}_k\right\}}\\
	{\text{Im}\left\{s_p^\text{H} \textbf{e}_k^\text{H} \textbf{z}_k\right\}}\\
	{\text{Re}\left\{s_p^\text{H} \tilde{\textbf{e}}_{k1}^\text{H} \textbf{z}_k\right\}}\\
	{\text{Re}\left\{s_p^\text{H} \tilde{\textbf{e}}_{k2}^\text{H} \textbf{z}_k\right\}}
	\end{matrix}
	\right].
	\end{align}
	\vspace{-1mm}
\end{small}

A stable point $\hat {\boldsymbol{\psi}}_{k-1}$ of $\textbf{f} \left(\hat {\boldsymbol{\psi}}_{k-1},{\boldsymbol{\psi}}_{k}\right)$ satisfies two conditions: 1) $\textbf{f} \left(\hat {\boldsymbol{\psi}}_{k-1},{\boldsymbol{\psi}}_{k}\right) = 0$; 2) $\frac{\partial \textbf{f} \left(\hat {\boldsymbol{\psi}}_{k-1},{\boldsymbol{\psi}}_{k}\right)}{\partial \hat {\boldsymbol{\psi}}_{k-1}^\text{T} }$ is negative definite. Hence, we define the stable points set in time-slot $k$ as : $\mathcal{S}_{k} \triangleq \left\{\hat {\boldsymbol{\psi}}_{k-1}:\textbf{f} \left(\hat {\boldsymbol{\psi}}_{k-1},{\boldsymbol{\psi}}_{k}\right) = 0,\frac{\partial \textbf{f} \left(\hat {\boldsymbol{\psi}}_{k-1},{\boldsymbol{\psi}}_{k}\right)}{\partial \hat {\boldsymbol{\psi}}_{k-1}^\text{T}} \prec \textbf{0} \right\}$.

The channel parameter ${\boldsymbol{\psi}}_{k}$ is a stable point: when $\hat {\boldsymbol{\psi}}_{k-1}={\boldsymbol{\psi}}_{k}$,

1) $\beta_k \textbf{W}_k^\text{H} \textbf{a}\left(\textbf{x}_k\right) =\hat{\beta}_{k-1} \textbf{e}_k$ in \eqref{eq_f}. Hence, $\textbf{f} \left( {\boldsymbol{\psi}}_{k},{\boldsymbol{\psi}}_{k}\right) = 0$;
\vspace{2mm}

2)$\frac{\partial \textbf{f} \left( {\boldsymbol{\psi}}_{k},{\boldsymbol{\psi}}_{k}\right)}{\partial \hat {\boldsymbol{\psi}}_{k-1}^\text{T} }=-\textbf{J}_4$ by derivation, where $\textbf{J}_4$ is a $4 \times 4$ identity matrix. Thus,
$\frac{\partial \textbf{f} \left( {\boldsymbol{\psi}}_{k},{\boldsymbol{\psi}}_{k}\right)}{\partial \hat {\boldsymbol{\psi}}_{k-1}^\text{T} }$ is negative definite.\\
Therefore, ${\boldsymbol{\psi}}_k$ is a stable point.

Other stable points in $\mathcal{S}_k$ correspond to the local optimal points of the beam and channel tracking problem, which are without the main lobe $\mathcal{B}(\textbf{x})$. Except for the channel parameter vector ${\boldsymbol{\psi}}_{k}$, the antenna array gain of other stable points in $\mathcal{S}_k$ is quite low, resulting in low tracking accuracy. Therefore, one key challenge is to ensure that the tracking algorithm converges to ${\boldsymbol{\psi}}_k$ rather than other stable points.

In static scenarios, where $\mathcal{S}_{k} =\mathcal{S} \triangleq \big\{\hat {\boldsymbol{\psi}}_{k-1}:\textbf{f} \left(\hat {\boldsymbol{\psi}}_{k-1},{\boldsymbol{\psi}}\right) = 0,\frac{\partial \textbf{f} \left(\hat {\boldsymbol{\psi}}_{k-1},{\boldsymbol{\psi}}\right)}{\partial \hat {\boldsymbol{\psi}}_{k-1}^\text{T}} \prec \textbf{0} \big\}$, the corresponding theorems are developed to study the convergence of our algorithm. We adopt the diminishing step-size in \eqref{eq_stepsize}, given by \cite{nevel1973stochastic,borkar2008stochastic,kushner2003stochastic}
\vspace{-1mm}
\begin{equation}\label{eq_stepsize}
\begin{aligned}
b_k=\frac{\epsilon}{k+K_0},k=1,2,\cdots
\end{aligned}
\vspace{-1mm}
\end{equation}where $K_0 \geq 0$ and $\epsilon > 0$.
%To prove that the JRBCT algorithm converges to ${\boldsymbol{\psi}}$ instead of other stable points, we develop three theorems:
\begin{theorem}[\textbf{Convergence to a Unique Stable Point}]\label{Converge to unique stable point}
	\emph{If $b_k$ is given by \eqref{eq_stepsize} with $\epsilon > 0$ and $K_0 \geq 0$, then $\hat {\boldsymbol{\psi}}_k$ converges to a unique stable point with probability one.}
\end{theorem}
\begin{proof}
	See Appendix \ref{proof_Converge to unique stable point}.
\end{proof}

Therefore, for the general step-size in \eqref{eq_stepsize}, $\hat {\boldsymbol{\psi}}_k$ converges to a unique stable point.

\begin{theorem}[\textbf{Convergence to Direction parameter vector x}]\label{Converge to real beam direction}
	\emph{If (i) the initial estimate of $\textbf{x}$ \emph{is within the main lobe, i.e.,} $\hat {\textbf{x}}_0 \in \mathcal{B}\left(\textbf{x}\right)$, and (ii) $b_k$ is given by \eqref{eq_stepsize} with $\epsilon > 0$, then there exist some $K_0 \geq 0$ and $C > 0$ such that}
	\begin{equation}\label{eq_ConvergeProbablity}
	\begin{aligned}
	P\left(\hat {\textbf{x}}_k \to \textbf{x} \mid \hat {\textbf{x}}_0 \in \mathcal{B}\left(\textbf{x}\right)\right) \geq 1-8e^{-\frac{C \lvert s_p \rvert^2}{\epsilon^2 \sigma^2}}.
	\end{aligned}
	\vspace{-1mm}
	\end{equation}
\end{theorem}
\begin{proof}
	See Appendix \ref{proof_Converge to real beam direction}.
\end{proof}

At the coarse beam sweeping stage of our proposed algorithm, the initial estimation $\hat {\textbf{x}}_0$ within main lobe $\mathcal{B}\left(\textbf{x}\right)$ in \eqref{eq_MainLobe} can be obtained with high probability. Under the condition $\hat {\textbf{x}}_0 \in \mathcal{B}\left(\textbf{x}\right)$, Theorem \ref{Converge to real beam direction} tells us the probability of $\hat {\textbf{x}}_k \to \textbf{x}$ is related to $\frac{ \lvert s_p \rvert^2}{\epsilon^2 \sigma^2}$. Hence, we can reduce the step-size and increase the transmit SNR $\frac{\lvert s_p \rvert^2}{ \sigma^2}$ to make sure that $\hat {\textbf{x}}_k \to \textbf{x}$ with probability one.
\begin{theorem}[\textbf{Convergence to x with minimum CRLB}]\label{Converge to with minimum CRLB}
	\emph{If (i) $\hat {\boldsymbol{\psi}}_k \to {\boldsymbol{\psi}}$ and (ii) $b_k$ is given by \eqref{eq_stepsize} with $\epsilon = 1$ and any  $K_0 \geq 0$, then $\hat{\textbf{h}}_k - \textbf{h}_k$ is asymptotically Gaussian and}
	\begin{equation}\label{eq_ConvergeWithMinCRLB}
	\begin{aligned}
	\mathop {\lim }\limits_{k \to \infty } \frac{k}{MN} \mathbb{E} \left[{\left\| \hat{\textbf{h}}_k - \textbf{h} \right\|}_2^2 \big| \hat{\boldsymbol{\psi}}_k \to \boldsymbol{\psi} \right] = {I}_{\min}(\boldsymbol{\psi}).
	\end{aligned}
	\vspace{-1mm}
	\end{equation}
\end{theorem}
\begin{proof}
	See Appendix \ref{proof_Converge to with minimum CRLB}.
\end{proof}

Theorem \ref{Converge to with minimum CRLB} tells us $\epsilon$ should not be too large or too small. By Theorem \ref{Converge to with minimum CRLB}, if $\epsilon = 1$, then we achieve the minimum CRLB asymptotically with high probability.

\section{Numerical Results}\label{sec_simulation}
\vspace{-0.5mm}
We compare the proposed algorithm with four other algorithms: the compressed sensing algorithm in \cite{Alkhateeb2015Compressed}, the IEEE 802.11ad algorithm in \cite{IEEE80211ad}, the extended Kalman filter (EKF) method in
\cite{Vutha2016Tracking} and the joint beam and channel tracking algorithm in \cite{JLiJoint2018} (using two pilots to track each dimension of the 2D beam direction). In each time-slot, three pilots are transmitted for all the algorithms to ensure fairness. When adopting the joint beam and channel tracking algorithm by using four pilots, we use a buffer to store the received pilots and update the estimate when receiving four new pilots. Based on the model in Section II, the parameters are set as: $M\!=\!N\!=\!8$, the antenna spacing $d_1\!=\!d_2\!=\!\frac{\lambda}{2}$, the codebook size $M_0 = 2M, N_0=2N$, the pilot symbol $s_p$ = 1, and the transmit $\text{SNR} = \frac{ \lvert s_p \rvert^2}{\sigma^2}=0 \text{dB}$.
%\begin{figure}[!t]
%\centering
%\vspace{-3mm}
%\includegraphics[width=6.3cm]{20171025_static_NMSEvsOverhead_SNR=10dB.eps}
%%{20170509_static_NMSEvsOverhead_SNR=10dB_report.eps}
%\vspace{-2mm}
%\caption{$\text{MSE}_{\textbf{h},n}$ vs. time-slot number $n$ in static beam tracking scenarios.}
%\vspace{-5mm}
%\label{fig_static_mse}
%\end{figure}

In static scenarios, the AoA ($\theta$,$\phi$) as defined in Section II is chosen evenly and randomly in $\theta \in \left[0,\frac{\pi}{2}\right]$, $\phi \in \left[-\pi,\pi\right)$. The channel coefficient is set as a constant $\beta_k={(1+1j)}/{\sqrt{2}}$. The step-size $b_k$ is set as $b_k=1/k$. Simulation results are averaged over 1000 random system realizations. Fig. \ref{fig_static_mse}
\begin{figure}[!t]
\centering
\vspace{0mm}
\includegraphics[width=3.33in]{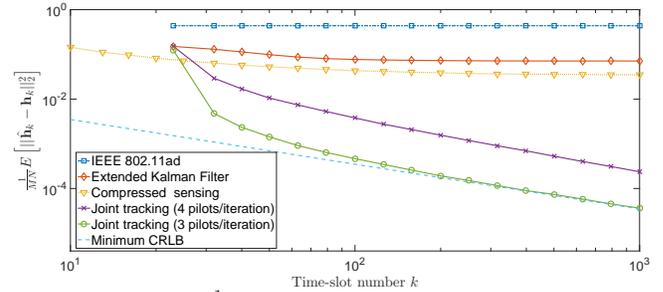}
\vspace{-4mm}
\caption{$\frac{1}{MN} \text{MSE}_{\textbf{h}_k}$ in static tracking scenarios.}
\vspace{-3mm}
\label{fig_static_mse}
\end{figure}indicates that the channel vector MSE of our proposed algorithm approaches the minimum CRLB quickly and achieves much lower tracking error than other algorithms.
%\begin{figure}[!t]
%\centering
%\vspace{-3mm}
%\includegraphics[width=6.5cm]{20171025_static_NMSEvsOverhead_SNR=10dB.eps}
%%{20170509_static_NMSEvsOverhead_SNR=10dB_report.eps}
%\vspace{-2mm}
%\caption{$\text{MSE}_{\textbf{h},n}$ vs. time-slot number $n$ in static beam tracking scenarios.}
%\vspace{-5mm}
%\label{fig_static_mse}
%\end{figure}

In dynamic scenarios, the AoA ($\theta$,$\phi$) as defined in Section II is modeled as a random walk process, i.e., $\theta_{k+1} = \theta_ k+ \Delta \theta$, $\phi_{k+1} = \phi_k+ \Delta \phi$; $\Delta \theta,\Delta \phi \sim \mathcal{CN}(0,\delta^2)$. The initial AoA values are chosen evenly and randomly in $\theta_0 \in \left[0,\frac{\pi}{2}\right]$, $\phi_0 \in \left[-\pi,\pi\right)$. The channel coefficient is modeled as Rician fading with a K-factor $\kappa$=15dB, according to the channel model in \cite{Samimi2016FadingModel}. As for the step-size $b_k$, we adopt the constant step-size. Numerical results show that when $b_k = 0.7$, the joint beam and channel tracking algorithm can track beams with higher velocity. Therefore, the step-size is set as a constant $b_k = 0.7$. Fig. \ref{fig_dynamic_mse}
\begin{figure}[!t]
\centering
\vspace{-0mm}
\includegraphics[width=3.33in]{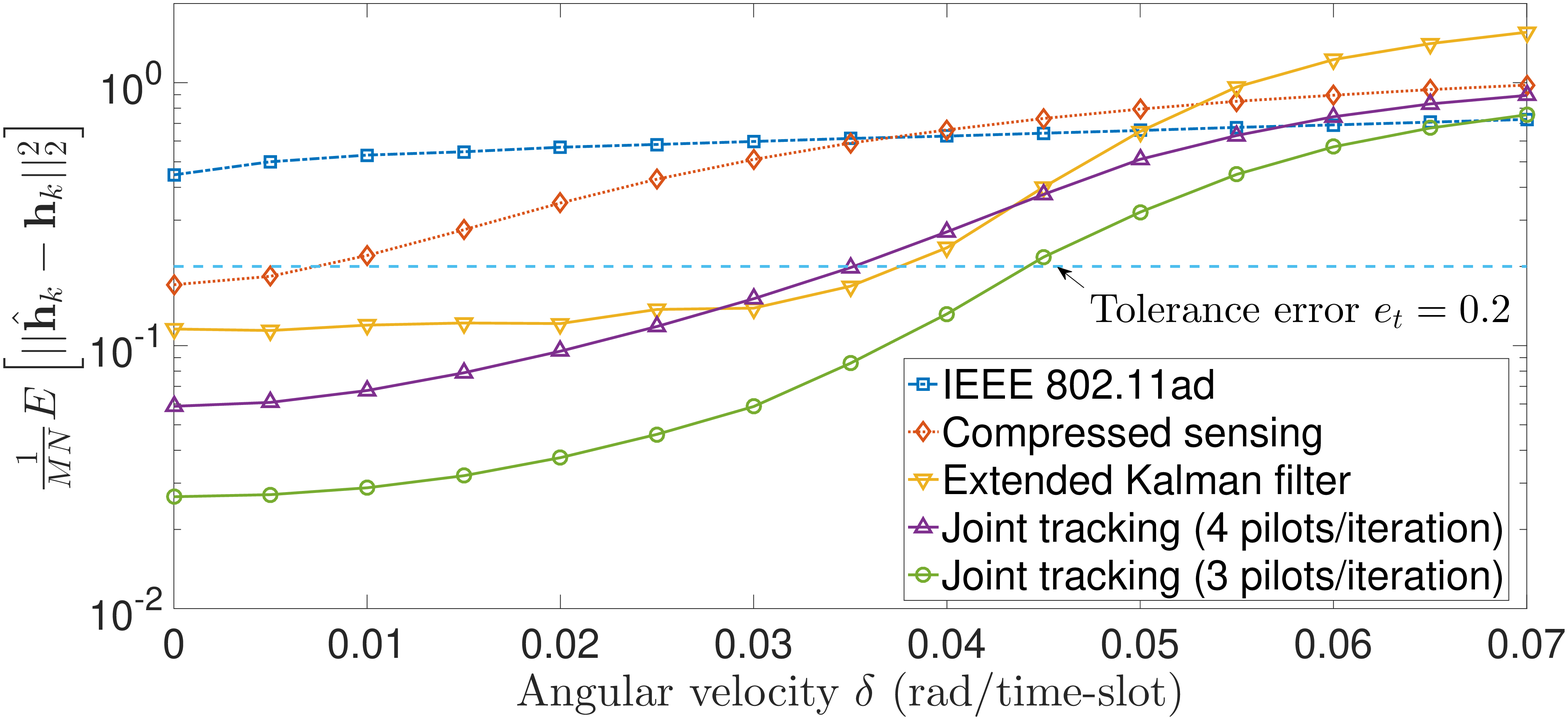}
%{20170509_static_NMSEvsOverhead_SNR=10dB_report.eps}
\vspace{-2mm}
\caption{$\frac{1}{MN} \text{MSE}_{\textbf{h}_k}$ in dynamic tracking scenarios.}
\vspace{-5mm}
\label{fig_dynamic_mse}
\end{figure}indicates the proposed algorithm can achieve higher tracking accuracy than the other four algorithms. In addition, if we set a tolerance error $e_t$, e.g., $e_t=0.2$, then our algorithm can support higher angular velocities.

\vspace{-2mm}
\section{Future Work Remarks}\label{future work remarks}
\vspace{-0.5mm}
%We have developed an analog beam tracking algorithm, and established its convergence and asymptomatic optimality. Our theoretical and simulation results show that the proposed algorithm can achieve much faster tracking speed, lower beam tracking error, and higher data rate than several state-of-the-art algorithms, with the same pilot overhead. In our future work, we will consider hybrid beamforming systems with multiple RF chains and two-dimensional antenna arrays, based on the methodology developed in the current paper.
In this paper, we have developed a joint beam and channel tracking algorithm for 2D phased antenna arrays. A general sequence of optimal analog beamforming parameters is obtained to achieve the minimum CRLB. The work is a first step to beam and channel tracking with 2D phased antenna arrays. In our future work, we will focus on the following aspects: i) establishing the corresponding theorems in dynamic scenarios; ii) jointly tracking multiple paths; iii) tracking at both the transmitter and receiver.
\vspace{-2mm}
%
%
%
%We have developed a joint beam and channel tracking algorithm for 2D phased antenna arrays and give a general sequence of optimal beamforming parameters. The theorem of convergence and asymptomatic optimality is established for this algorithm. Simulation results show that our algorithm can achieve lower tracking error than several existing algorithms.
%achieve faster tracking speed and lower tracking error than several existing algorithms.

%\input{reference}

%\ifreport

\bibliographystyle{IEEEtran}
\bibliography{IEEEabrv,reference}

\appendices
\section{Proof of Lemma~\ref{IndObservations}}\label{proof_IndObservations}
Since the effect of noise can be reduced to zero by multiple observations, we ignore the observation noise in the proof for the sake of simplicity.

If the analog beamforming vectors are steering vectors, i.e., $\textbf{w}_{k,i} = \frac{1}{\sqrt{MN}}\textbf{a}\left(\textbf{x}+\boldsymbol{\Delta}_{k,i}\right)$, where $\boldsymbol{\Delta}_{k,i}=\left[\delta_{k,i1},\delta_{k,i2}\right]^\text{T}$ denotes the $i$-th direction parameter offset, then we get the complex observation equation for the $i$-th observation:
\begin{align}\label{eq_aObservationI}
{y_{k,i}}= \frac{{{s_p}\beta }}{{\sqrt {MN} }}\sum\limits_{m = 1}^M {\sum\limits_{n = 1}^N {{e^{-j2\pi \left(\frac{{(m - 1)\delta_{k,i1}}}{M} + \frac{{(n - 1)\delta_{k,i2}}}{N}\right)}}} },
\end{align}which contains two real equations, i.e., an amplitude equation and a phase angle equation. From \eqref{eq_aObservationI}, we can obtain the phase angle equation:
\begin{align}\label{eq_aObservationIP}
\angle({y_{k,i}})=\angle ({s_p}\beta){\rm{ - }}\pi \left[ {\frac{M - 1}{M}\delta_{k,i1} + \frac{N - 1}{N}{\delta_{k,i2}}} \right].
\end{align}Then the relationship between the phase angles of two different observations $y_{k,i}$ and $y_{k,j}\,(i \neq j)$ is given by
\begin{small}
	\begin{align}\label{eq_aObservationDep}
	\angle ({y_{k,i}}) \!-\!\angle ({y_{k,j}})
	\!=\!\pi \left[\! {\frac{{M \!-\! 1}}{M}({\delta_{k,j1}} \!-\! {\delta_{k,i1}}) \!+\! \frac{{N \!-\! 1}}{N}(\delta_{k,{j2}} \!-\! {\delta_{k,i2}})} \!\right],\nonumber
	\end{align}
\end{small}where $\delta_{k,{i1}} - {\delta_{k,j1}}$ and $\delta_{k,{i2}} - {\delta_{k,j2}}$ are determined by the direction parameter offsets and unrelated to the channel parameter vector $\boldsymbol{\psi}_k$.

Hence, the phase angles of any two observations $y_{k,i}$ and $y_{k,j}$ are correlated. After $q$ observations, we can obtain $q$ independent amplitude equations and only 1 independent phase angle equation, which are $q+1$ independent real equations in total.

When estimating $q+1$ real variables, at least $q+1$ independent real equations are required. Therefore, at least $q$ observations are needed to obtain $q+1$ independent real equations and estimate $q+1$ real variables, which completes the proof.

\section{Proof of Lemma~\ref{MSEOpt}}\label{proof_MSEOpt}
In problem \eqref{eq_problem}, the constraint \eqref{eq_constrant1} ensures that $\hat{\textbf{h}}_k$ is an unbiased estimation of ${\textbf{h}}_k$. In static scenarios, where $\textbf{h}_k = \textbf{h} \triangleq \beta \textbf{a}(\textbf{x})$, we consider each element of the channel vector ${\textbf{h}}$. Given ${h_{mn}}({{\boldsymbol{\psi }}}) = \beta {e^{j2\pi \left( {\frac{{m - 1}}{M}{x_1} + \frac{{n - 1}}{N}{x_2}} \right)}}$, we have $\mathbb{E} \left[{h_{mn}}(\hat{{\boldsymbol{\psi }}})\right]= {h_{mn}}({{\boldsymbol{\psi }}})$ since $\mathbb{E}\left[\hat{\textbf{h}}_k\right] = \textbf{h}$. According to section 
3.8 of \cite{Sengijpta1993Fundamental}, if a function $f(\hat {\textbf{x}})$ is an unbiased estimation of $f({{\textbf{x}}})$, i.e., $\mathbb{E} \left[f(\hat{{\bf{x}}})\right]= f({{\textbf{x}}})$, then we can obtain that 
\vspace{-1mm}
\begin{equation}\label{eq_fLB}
\begin{aligned}
\operatorname{Var}[f(\hat {\textbf{x}})] \ge \frac{{\partial f( {\textbf{x}})}}{{\partial { {\textbf{x}}}}}{\textbf{I}( {\textbf{x}})^{ - 1}}{\left(\frac{{\partial f({ {\textbf{x}}})}}{{\partial { {\textbf{x}}}}}\right)^\text{H}},
\end{aligned}
\end{equation}where $\textbf{I}(\textbf{x})$ is the corresponding Fisher information matrix.

The partial derivative of ${h_{mn}}({{\boldsymbol{\psi }}})$ is given as follows:
\vspace{-1mm}
\begin{equation}\label{eq_fmnd}
\begin{aligned}
{\left\{ {\begin{array}{*{20}{l}}
		{\frac{{\partial {h_{mn}}({{\boldsymbol{\psi }}})}}{{\partial {\beta ^{re}}}} = {e^{j2\pi \left( {\frac{{m - 1}}{M}{x_1} + \frac{{n - 1}}{N}{x_2}} \right)}}}\\
		{\frac{{\partial {h_{mn}}({{\boldsymbol{\psi }}})}}{{\partial {\beta ^{im}}}} = j{e^{j2\pi \left( {\frac{{m - 1}}{M}{x_1} + \frac{{n - 1}}{N}{x_2}} \right)}}}\\
		{\frac{{\partial {h_{mn}}({{\boldsymbol{\psi }}})}}{{\partial {x_1}}} = j2\pi \frac{{m - 1}}{M}\beta {e^{j2\pi \left( {\frac{{m - 1}}{M}{x_1} + \frac{{n - 1}}{N}{x_2}} \right)}}}\\
		{\frac{{\partial {h_{mn}}({{\boldsymbol{\psi }}})}}{{\partial {x_2}}} = j2\pi \frac{{n - 1}}{N}\beta {e^{j2\pi \left( {\frac{{m - 1}}{M}{x_1} + \frac{{n - 1}}{N}{x_2}} \right)}}}
		\end{array}} \right.}.
\end{aligned}
\end{equation}

Combining \eqref{eq_problem}, \eqref{eq_fLB} and \eqref{eq_fmnd}, we have
\vspace{-1mm}
\begin{align}\label{eq_CMMSEP}
&\frac{1}{{MN}}\mathbb{E}\left[\left\|{\hat{\textbf{h}}_k} - {{ \textbf{h}}_k}\right\|_2^2\right]\nonumber\\
= &\frac{1}{{MN}}\sum\limits_{m = 1}^M {\sum\limits_{n = 1}^N {\mathbb{E}\left[{\big| h_{mn}(\hat{\boldsymbol{\psi}})-h_{mn}({\boldsymbol{\psi}})\big|}^2  \right]} } \\
\overset{(a)}{\ge} & \frac{1}{{MN}}\sum\limits_{m = 1}^M {\sum\limits_{n = 1}^N \left({\textbf{v}_{mn}{{\left(\sum\limits_{l = 1}^k {{\bf{I}}(\psi ,{{\bf{W}}_l})} \right)}^{ - 1}}{\textbf{v}_{mn}^\text{H}}}\right) } \nonumber\\
= &\frac{1}{{MN}}\Tr \left\{ {{{\left(\sum\limits_{l = 1}^k {{\bf{I}}(\psi ,{{\bf{W}}_l})} \right)}^{ - 1}}\sum\limits_{m = 1}^M {\sum\limits_{n = 1}^N {\left({\textbf{v}_{mn}^\text{H}}{\textbf{v}_{mn}}\right)} } } \right\},\nonumber
\end{align}where Step (a) is obtained by substituting \eqref{eq_fmnd} into \eqref{eq_fLB}.

Hence, Lemma \ref{MSEOpt} is proved.

\section{Proof of Lemma~\ref{UnifiedOptShift}}\label{proof_UnifiedOptShift}
Lemma \ref{UnifiedOptShift} is proved in three steps:

\emph{\textbf{Step 1}: We prove that $\boldsymbol{\Delta}_{1}^*,\boldsymbol{\Delta}_{2}^*,\boldsymbol{\Delta}_{3}^*$ are unrelated to channel coefficient $\beta$.}

The basic method is block matrix inversion: the Fisher information matrix in \eqref{eq_fisher} is divided into four $2 \times 2$ matrices as follows:
\vspace{-1mm}
\begin{equation}\label{eq_FisherBlocks}
\begin{aligned}
\textbf{I}(\boldsymbol{\psi}, \textbf{W}_k) = \frac{{2{{\left| {{s_p}} \right|}^2}}}{{{\sigma ^2}}}\left[\!{\begin{array}{*{20}{c}}
	{{\bf{A}}(M,N)}&{{\bf{B}}(M,N,\beta )}\\
	{{{\bf{B}}^\text{T}}(M,N,\beta )}&{{\bf{D}}(M,N,\beta )}
	\end{array}} \!\right]\!,
\end{aligned}
\vspace{-1mm}
\end{equation}where $\textbf{A}(M,N)$, $\textbf{B}(M,N,\beta)$, $\textbf{D}(M,N,\beta)$ are defined as:
\vspace{-1mm}
\begin{align}\label{eq_Block}
\left\{\begin{array}{*{20}{l}}
{\bf{A}}(M,N) \triangleq  \left[ {\begin{array}{*{20}{c}}
	{\left\| {{{\bf{g}}_{{k}}}} \right\|_2^2}&0\\
	0&{\left\| {{{\bf{g}}_{{k}}}} \right\|_2^2}
	\end{array}} \right]  \\
{\bf{B}}(M,N,\beta ) \triangleq \left[ {\begin{array}{*{20}{c}}
	{{\mathop{\rm Re}\nolimits} \{ {\bf{g}}_{{k}}^{\text{H}}{{{\bf{\tilde g}}}_{{{k,1}}}}\} }&{{\mathop{\rm Re}\nolimits} \{ {\bf{g}}_{{k}}^{\text{H}}{{{\bf{\tilde g}}}_{{{k,2}}}}\} }\\
	{{\rm{Im}}\{ {\bf{g}}_{{k}}^{\text{H}}{{{\bf{\tilde g}}}_{{{k,1}}}}\} }&{{\rm{Im}}\{ {\bf{g}}_{{k}}^{\text{H}}{{{\bf{\tilde g}}}_{{{k,2}}}}\} }
	\end{array}} \right]. \\
{\bf{D}}(M,N,\beta) \triangleq \left[\! {\begin{array}{*{20}{c}}
	{\left\| {{{{\bf{\tilde g}}}_{{{k,1}}}}} \right\|_2^2}&{{\mathop{\rm Re}\nolimits} \{ {\bf{\tilde g}}_{{{k,1}}}^{\text{H}}{{{\bf{\tilde g}}}_{{{k,2}}}}\} }\\
	{{\mathop{\rm Re}\nolimits} \{ {\bf{\tilde g}}_{{{k,1}}}^{\text{H}}{{{\bf{\tilde g}}}_{{{k,2}}}}\} }&{\left\| {{{{\bf{\tilde g}}}_{{{k,2}}}}} \right\|_2^2}
	\end{array}} \!\right]
\end{array}\right.\!
\end{align}
\vspace{-1mm}\\
Then the inverse matrix of \eqref{eq_FisherBlocks} is given in
\vspace{-1mm}
\begin{small}
	\begin{equation}\label{eq_InverseFisherBlocks}
	\begin{aligned}
	{\bf{I}}{(\bf{\psi} ,{{\bf{W}}_k})^{-1}}= \frac{{{\sigma ^2}}}{{2{{\left| {{s_p}} \right|}^2}}}\left\{ {{{\bf{I}}_{i{p_1}}}(M,N) + {{\bf{I}}_{i{p_2}}}\left( {M,N,\beta } \right)} \right\},
	\end{aligned}
	\end{equation}
\end{small}where ${{\bf{I}}_{i{p_1}}}\left(M,N\right)$ and ${{\bf{I}}_{i{p_2}}}\left(M,N,\beta\right)$ are defined as
\begin{small}
	\begin{equation}\label{eq_Ip}
	\begin{aligned}
	\!\!\!\left\{\begin{array}{*{20}{l}}
	\!\!\!{{\bf{{ I}}}_{i{p_1}}}(M,N) \triangleq \left[ {\begin{array}{*{20}{c}}
		{{\bf{A}}^{-1}}&{\bf{0}}\\
		{\bf{0}}&{\bf{0}}
		\end{array}} \right]\\
	\!\!\!{{\bf{I}}_{i{p_2}}}\left(M,N,\beta\right) \!\triangleq \!\left[\! {\begin{matrix}
		{{\bf{A}}{^{\!-1}}{\bf{B}}}\\
		{{\bf{ - }}{{\bf{J}}_2}}
		\end{matrix}} \!\right]\!\!{{\left( {{\bf{D}}\!-\!{{\bf{B}}^{\rm{T}}}{\bf{A}}{^{\!-1}}{\bf{B}}} \right)}^{\!-1}}\!\left[ {\begin{matrix}
		{{{\bf{B}}^{\rm{T}}}{\bf{A}}{^{\!-1}}}& \!\!\!\!{{\bf{\!\!-}}{{\bf{J}}_2}}
		\end{matrix}} \!\right].
	\end{array}\!\!\right.\!\!\!\!
	\end{aligned}
	\end{equation}
\end{small}$\textbf{J}_2$ is $2 \times 2$ identity matrix. By combining $\textbf{A}(M,N)$, $\textbf{B}(M,N,\beta)$, and $\textbf{D}(M,N,\beta)$ in \eqref{eq_Block}, ${{\left( {{\bf{D}} - {{\bf{B}}^{\rm{T}}}{\bf{A}}{^{-1}}{\bf{B}}} \right)}}/{\lvert\beta \rvert}^2$ can be converted to a matrix $\textbf{I}_s(M,N)$ as shown in \eqref{eq_Is},
\begin{figure*}[!t]
	\normalsize
	\vspace{-5mm}
	\begin{align}\label{eq_Is}
	\frac{{\bf{D}}\!-\!{{\bf{B}}^{\rm{T}}}{\bf{A}}{^{\!-1}}{\bf{B}}}{{\lvert \beta \rvert}^2} &=\frac{1}{{\lvert \beta \rvert}^2}\left\{\left[ {\begin{matrix}
		{\left\| {{{{\bf{\tilde g}}}_{{{k,1}}}}} \right\|_2^2}&{{\mathop{\rm Re}\nolimits} \{ {\bf{\tilde g}}_{{{k,1}}}^{\text{H}}{{{\bf{\tilde g}}}_{{{k,2}}}}\} }\\
		{{\mathop{\rm Re}\nolimits} \{ {\bf{\tilde g}}_{{{k,1}}}^{\text{H}}{{{\bf{\tilde g}}}_{{{k,2}}}}\} }&{\left\| {{{{\bf{\tilde g}}}_{{{k,2}}}}} \right\|_2^2}
		\end{matrix}} \right]-\frac{1}{\left\|\bf{g}_{k}\right\|_2^2}\left[ {\begin{matrix}
		{  \left\|{\bf{g}}_{{k}}^{\text{H}}{\bf{\tilde g}}_{k,1}\right\|_2^2 }&{{\mathop{\rm Re}\nolimits} \{\bf{\tilde{g}}_{k,1}^\text{H}{\bf{g}_{k}} \bf{g}_{k}^\text{H}{\bf{\tilde{g}}_{k,2}}\} }\\
		{{\mathop{\rm Re}\nolimits} \{\bf{\tilde{g}}_{k,1}^\text{H}{\bf{g}_{k}} \bf{g}_{k}^\text{H}{\bf{\tilde{g}}_{k,2}}\} }&{  \left\|{\bf{g}}_{{k}}^{\text{H}}{\bf{\tilde g}}_{k,2}\right\|_2^2 }
		\end{matrix}} \right]\right\}\\
	&\overset{(a)}{=}\frac{1}{{\lvert \beta \rvert}^2}{\lvert \beta \rvert}^2\left\{\left[ {\begin{matrix}
		{\left\| {{{{\bf{\breve g}}}_{{{k,1}}}}} \right\|_2^2}&{{\mathop{\rm Re}\nolimits} \{ {\bf{\breve g}}_{{{k,1}}}^{\text{H}}{{{\bf{\breve g}}}_{{{k,2}}}}\} }\\
		{{\mathop{\rm Re}\nolimits} \{ {\bf{\breve g}}_{{{k,1}}}^{\text{H}}{{{\bf{\breve g}}}_{{{k,2}}}}\} }&{\left\| {{{{\bf{\breve g}}}_{{{k,2}}}}} \right\|_2^2}
		\end{matrix}} \right]-\frac{1}{\left\|\bf{g}_{k}\right\|_2^2}\left[ {\begin{matrix}
		{  \left\|{\bf{g}}_{{k}}^{\text{H}}{\bf{\breve g}}_{k,1}\right\|_2^2 }&{{\mathop{\rm Re}\nolimits} \{\bf{\breve{g}}_{k,1}^\text{H}{\bf{g}_{k}} \bf{g}_{k}^\text{H}{\bf{\breve{g}}_{k,2}}\} }\\
		{{\mathop{\rm Re}\nolimits} \{\bf{\breve{g}}_{k,1}^\text{H}{\bf{g}_{k}} \bf{g}_{k}^\text{H}{\bf{\breve{g}}_{k,2}}\} }&{  \left\|{\bf{g}}_{{k}}^{\text{H}}{\bf{\breve g}}_{k,2}\right\|_2^2 }
		\end{matrix}} \right]\right\}\nonumber\\
	&=\left\{\left[ {\begin{matrix}
		{\left\| {{{{\bf{\breve g}}}_{{{k,1}}}}} \right\|_2^2}&{{\mathop{\rm Re}\nolimits} \{ {\bf{\breve g}}_{{{k,1}}}^{\text{H}}{{{\bf{\breve g}}}_{{{k,2}}}}\} }\\
		{{\mathop{\rm Re}\nolimits} \{ {\bf{\breve g}}_{{{k,1}}}^{\text{H}}{{{\bf{\breve g}}}_{{{k,2}}}}\} }&{\left\| {{{{\bf{\breve g}}}_{{{k,2}}}}} \right\|_2^2}
		\end{matrix}} \right]-\frac{1}{\left\|\bf{g}_{k}\right\|_2^2}\left[ {\begin{matrix}
		{  \left\|{\bf{g}}_{{k}}^{\text{H}}{\bf{\breve g}}_{k,1}\right\|_2^2 }&{{\mathop{\rm Re}\nolimits} \{\bf{\breve{g}}_{k,1}^\text{H}{\bf{g}_{k}} \bf{g}_{k}^\text{H}{\bf{\breve{g}}_{k,2}}\} }\\
		{{\mathop{\rm Re}\nolimits} \{\bf{\breve{g}}_{k,1}^\text{H}{\bf{g}_{k}} \bf{g}_{k}^\text{H}{\bf{\breve{g}}_{k,2}}\} }&{  \left\|{\bf{g}}_{{k}}^{\text{H}}{\bf{\breve g}}_{k,2}\right\|_2^2 }
		\end{matrix}} \right]\right\} \triangleq \textbf{I}_s(M,N).\nonumber
	\end{align}
	\vspace*{-0pt}
\end{figure*}where Step (a) is due to the definition of $\breve{\textbf{g}}_{k,1}$ and $\breve{\textbf{g}}_{k,2}$:
\vspace{-1mm}
\begin{equation}\label{eq_gkb}
\left\{{\begin{array}{*{20}{l}}
	\breve{\textbf{g}}_{k,1} \triangleq \frac{1}{\beta} \tilde{\textbf{g}}_{k,1} =\textbf{W}_k^\text{H}\frac {\partial \textbf{a}\left(\textbf{x}\right)}{\partial x_1},\\
	\breve{\textbf{g}}_{k,2} \triangleq \frac{1}{\beta} \tilde{\textbf{g}}_{k,2} =\textbf{W}_k^\text{H}\frac {\partial \textbf{a}\left(\textbf{x}\right)}{\partial x_2}.
	\end{array}}\right.
\end{equation}
\vspace{-1mm}
In \eqref{eq_Is}, $\textbf{I}_s(M,N)$ is unrelated to channel coefficient $\beta$ because none of $\textbf{g}_k$, $\breve{\textbf{g}}_{k,1}$, and $\breve{\textbf{g}}_{k,2}$ in \eqref{eq_Is} is related to $\beta$. By combining \eqref{eq_Ip} and \eqref{eq_Is}, we can rewrite \eqref{eq_Ip} as follows:
\begin{small}
	\begin{equation}\label{eq_rIp}
	\begin{aligned}
	\!\!\!\left\{\!\!\!\begin{array}{*{20}{l}}
	{{\bf{{ I}}}_{i{p_1}}}(M,N) \triangleq \left[ {\begin{array}{*{20}{c}}
		{{\bf{A}}^{-1}}&{\bf{0}}\\
		{\bf{0}}&{\bf{0}}
		\end{array}} \right]\\
	{{\bf{I}}_{i{p_2}}}\left(M,N,\beta\right) \!\triangleq\! \left[\! {\begin{matrix}
		{{\bf{A}}{^{\!-1}}{\bf{B}}}\\
		{{\bf{ - }}{{\bf{J}}_2}}
		\end{matrix}} \!\right]{{\left( \lvert \beta \rvert^2 \textbf{I}_s(M,N)\right)}^{\!-1}}\left[ {\begin{matrix}
		{{{\bf{B}}^{\rm{T}}}{\bf{A}}{^{\!-1}}}& \!\!\!\!{{\bf{\!\!-}}{{\bf{J}}_2}}
		\end{matrix}} \!\right]\!.
	\end{array}\!\!\right.\!\!\!\!
	\end{aligned}
	\end{equation}
\end{small}

Except for the inverse of the Fisher information matrix, the other parts in \eqref{eq_CMMSE} can be converted to \eqref{eq_sumrank1},
\begin{figure*}[!t]
	\vspace{-6mm}
	\begin{equation}\label{eq_sumrank1}
	\begin{aligned}
	T(M,N,\beta) \triangleq \sum\limits_{m = 1}^M {\sum\limits_{n = 1}^N {\left( {{\textbf{v}_{m,n}^{\rm{H}}}{\textbf{v}_{m,n}}} \right)} } 
	\!=\!MN\!\left[\!\!\! {\begin{array}{*{20}{c}}
		{1}&{j}&{j\pi \beta \frac{M - 1}{M}}&{j\pi \beta \frac{N-1}{N}}\\
		{-j}&{1}&{\pi \beta \frac{M - 1}{M}}&{\pi \beta \frac{N-1}{N}}\\
		{ - j\pi \bar \beta \frac{M - 1}{M}}&{\pi \bar \beta \frac{M - 1}{M}}&{\frac{2}{3}{\pi ^2}{{\left| \beta  \right|}^2}\frac{{(M - 1)(2M - 1)}}{M^2}}&{{\pi ^2}{{\left| \beta  \right|}^2}\frac{(M - 1)(N - 1)}{MN}}\\
		{ - j\pi \bar \beta \frac{N-1}{N}}&{\pi \bar \beta \frac{N - 1}{N}}&{{\pi ^2}{{\left| \beta  \right|}^2}\frac{(M - 1)(N - 1)}{MN}}&{\frac{2}{3}{\pi ^2}{{\left| \beta  \right|}^2}M\frac{{(N - 1)(2N - 1)}}{N^2}}
		\end{array}} \!\!\!\right]\!.\!\!
	\end{aligned}
	\end{equation}
	\hrulefill
	\vspace*{-5pt}
\end{figure*}where $\bar \beta$ denotes the conjugate of $\beta$. Therefore, we rewrite \eqref{eq_CMMSE} as:
\vspace{-1mm}
\begin{small}
	\begin{align}\label{eq_CMMSETemp}
	{I}_{\min}(\boldsymbol{\psi}) &=\frac{1}{MN} \Tr\left\{ {{{\left( {k{\bf{I}}(\psi ,{{\bf{W}}^*})} \right)}^{ - 1}}\sum\limits_{m = 1}^M {\sum\limits_{n = 1}^N {({\textbf{v}_{mn}^{\rm{H}}}\textbf{v}_{mn})} } } \right\}\nonumber\\
	&=\frac{1}{kMN}\frac{{{\sigma ^2}}}{{2{{\left| {{s_p}} \right|}^2}}}\Tr\left\{ {{\left( {{\bf{I}}(\psi ,{{\bf{W}}^*})} \right)}^{ - 1}}{\bf{T}}(M,N,\beta) \right\}\nonumber\\
	&=\frac{1}{kMN}\frac{{{\sigma ^2}}}{{2{{\left| {{s_p}} \right|}^2}}} \Tr\left\{ {{{\bf{I}}_{i{p_1}}}(M,N){\bf{T}}(M,N,\beta )}\right\}\\
	&+ \frac{1}{kMN}\frac{{{\sigma ^2}}}{{2{{\left| {{s_p}} \right|}^2}}} \Tr\left\{ {{{\bf{I}}_{i{p_2}}}(M,N){\bf{T}}(M,N,\beta )}\right\}\nonumber\\
	&\overset{(a)}{=}\frac{1}{kMN}\frac{{{\sigma ^2}}}{{2{{\left| {{s_p}} \right|}^2}}}\!\left\{\! {\frac{{2MN}}{{\left\| {{{\bf{g}}_{\bf{k}}}} \right\|_2^2}} \!+\! \Tr\left\{ {{{\bf{I}}_{i{p_2}}}(M,N){\bf{T}}(M,N,\beta )} \right\}} \!\right\}\nonumber\\
	&=\frac{1}{k}\frac{{{\sigma ^2}}}{{2{{\left| {{s_p}} \right|}^2}}}\!\left\{\! {\frac{{2}}{{\left\| {{{\bf{g}}_{\bf{k}}}} \right\|_2^2}} \!+\! \frac{1}{MN}\Tr\left\{ {{{\bf{I}}_{i{p_2}}}(M,N){\bf{T}}(M,N,\beta )} \right\}} \!\right\},\nonumber
	\end{align}
\end{small}where step (a) is by combining \eqref{eq_Ip} and \eqref{eq_sumrank1}.

To calculate $\Tr\left\{ {{{\bf{I}}_{i{p_2}}}(M,N){\bf{T}}(M,N,\beta )} \right\}$ in \eqref{eq_CMMSETemp}, we split ${\bf{T}}(M,N,\beta)$ in \eqref{eq_sumrank1} into two parts \eqref{eq_Tparts}:
\vspace{-1mm}
\begin{small}
	\begin{equation}\label{eq_Tparts}
	\begin{aligned}
	\!\!{\bf{T}}(M,N,\beta)\!=\! MN\!\!\left\{\!{\left[\!\! {\begin{array}{*{20}{c}}
			{{{\bf{b}}_{\text{T}}}}\\
			{{{\bf{c}}_{\text{T}}}}
			\end{array}} \!\right] {{{\left[\! {\begin{array}{*{20}{c}}
						{{{\bf{b}}_{\text{T}}}}\\
						{{{\bf{c}}_{\text{T}}}}
						\end{array}} \!\!\right]}^{\text{H}}}}\!\!+\! \left[\!\! {\begin{array}{*{20}{c}}
			{\bf{0}}&{\!\!\bf{0}}\\
			{\bf{0}}&{\!\!{{\bf{T}}_\text{D}}(M,N,\beta )}
			\end{array}} \!\!\right]} \!\right\},\!\!
	\end{aligned}
	\end{equation}
\end{small}
\vspace{-1mm}where ${{\bf{b}}_{\rm{T}}}$, ${{\bf{c}}_{\rm{T}}}$ and ${{\bf{T}}_\text{D}}(M,N,\beta)$ are defined as:
\begin{small}
	\begin{equation}\label{eq_rTD}
	\begin{aligned}
	\!\left\{\!{\begin{array}{*{20}{l}}
		{{\bf{b}}_{\rm{T}}} \triangleq  {\left[ {1, - j} \right]^{\rm{T}}}\\
		{{\bf{c}}_{\bf{T}}} \triangleq {\left[ {j\pi \bar \beta \frac{{M - 1}}{M},j\pi \bar \beta \frac{{N - 1}}{N}} \right]^{\rm{T}}}\\
		{{\bf{T}}_{\rm{D}}}(M,N,\beta) \!\triangleq\! \frac{1}{3}{\pi ^2}{\left| \beta  \right|^2}\!\left[\!\!{\begin{array}{*{20}{c}}
			{\frac{{(M - 1)(M - 3)}}{{{M^2}}}}&0\\
			0&{\frac{{(N - 1)(N - 3)}}{{{N^2}}}}
			\end{array}}\!\!\right]\\
		\end{array}}\right.\!\!
	\end{aligned}
	\end{equation}
\end{small}Hence, $\Tr\left\{ {{{\bf{I}}_{i{p_2}}}(M,N){\bf{T}}(M,N,\beta )} \right\}$ can be converted to
\begin{align}\label{eq_rtrcal}
&\quad \Tr\left\{ {{{\bf{I}}_{i{p_2}}}(M,N){\bf{T}}(M,N,\beta )} \right\}\nonumber\\
=&\quad \, MN\left( {\Tr\left\{ {{{\bf{I}}_{i{p_2}}}(M,N)\left[ {\begin{array}{*{20}{c}}
			{{{\bf{b}}_{\rm{T}}}}\\
			{{{\bf{b}}_{\rm{T}}}}
			\end{array}} \right] {{{\left[ {\begin{array}{*{20}{c}}
						{{{\bf{b}}_{\rm{T}}}}\\
						{{{\bf{c}}_{\rm{T}}}}
						\end{array}} \right]}^{\text{H}}}} } \right\}}\right) \\
&+MN\left( { \Tr\left\{ {{{\bf{I}}_{i{p_2}}}(M,N)\left[ {\begin{array}{*{20}{c}}
			{\bf{0}}&{\bf{0}}\\
			{\bf{0}}&{{{\bf{T}}_{\rm{D}}}(M,N,\beta )}
			\end{array}} \right]} \right\}} \right).\nonumber
\end{align}

Calculate the first part and second part separately in \eqref{eq_rtrcal}, we obtain that
\vspace{-1mm}
\begin{small}
	\begin{align}\label{eq_rtrcalsp1t}
	%\begin{array}{l}
	&\quad\Tr \left\{ {{{\bf{I}}_{i{p_2}}}(M,N)\left[ {\begin{matrix}
			{{{\bf{b}}_{\text{T}}}}\\
			{{{\bf{c}}_{\text{T}}}}
			\end{matrix}} \right] {{{\left[ {\begin{matrix}
						{{{\bf{b}}_{\text{T}}}}\\
						{{{\bf{c}}_{\text{T}}}}
						\end{matrix}} \right]}^{\text{H}}}} } \right\}\nonumber\\
	&=\Tr\left\{ {\left[ {\begin{matrix}
			{{{\bf{A}}^{\!-1}}{\bf{B}}}\\
			{{\bf{-}}{{\bf{J}}_2}}
			\end{matrix}} \right]{\left({{\lvert\beta \rvert}^2 \textbf{I}_s(M,N)}\right)^{\!-1}}\left[ {\begin{matrix}
			{{{\bf{B}}^{\rm{T}}}{{\bf{A}}^{\!-1}}}&\!\!\!\!\!{{\bf{-}}{{\bf{J}}_2}}
			\end{matrix}} \right]\left[ {\begin{matrix}
			{{{\bf{b}}_{\rm{T}}}}\\
			{{{\bf{c}}_{\rm{T}}}}
			\end{matrix}} \right]{{\left[ {\begin{matrix}
					{{{\bf{b}}_{\rm{T}}}}\\
					{{{\bf{c}}_{\rm{T}}}}
					\end{matrix}} \right]}^{\rm{H}}}} \right\}\nonumber\\
	&=\Tr\left\{ {{{\left[ {\begin{matrix}
					{{{\bf{b}}_{\rm{T}}}}\\
					{{{\bf{c}}_{\rm{T}}}}
					\end{matrix}} \right]}^{\text{H}}}\left[ {\begin{matrix}
			{{{\bf{A}}^{\!-1}}{\bf{B}}}\\
			{{\bf{-}}{{\bf{J}}_2}}
			\end{matrix}} \right]{\left({{\lvert\beta \rvert}^2 \textbf{I}_s(M,N)}\right)^{-1}}\left[ {\begin{matrix}
			{{{\bf{B}}^{\rm{T}}}{{\bf{A}}^{\!-1}}}&\!\!\!\!\!{{\bf{-}}{{\bf{J}}_2}}
			\end{matrix}} \right]\left[ {\begin{matrix}
			{{{\bf{b}}_{\text{T}}}}\\
			{{{\bf{c}}_{\text{T}}}}
			\end{matrix}} \right]} \right\}\nonumber\\
	&\overset{(a)}{=}\Tr\left\{ \left({\beta{\bf{a}}_{s}(M,N)}\right)^{\rm{H}}{\left({\lvert \beta \rvert}^2 {{ {\bf{I}}_s}(M,N)} \right)}^{-1}{{\beta \bf{a}}_s}(M,N) \right\}\nonumber\\
	&=\Tr\left\{\left({{\bf{a}}_{s}^{\rm{H}}(M,N)}\right){\left( {{ {\bf{I}}_s}(M,N)} \right)}^{-1}{{\bf{a}}_s}(M,N) \right\},
	\end{align}
\end{small}
\begin{align}\label{eq_rtrcalsp2t}
& { \Tr\left\{ {{{\bf{I}}_{i{p_2}}}(M,N)\left[ {\begin{array}{*{20}{c}}
			{\bf{0}}&{\bf{0}}\\
			{\bf{0}}&{{{\bf{T}}_{\rm{D}}}(M,N,\beta )}
			\end{array}} \right]} \right\}}\nonumber \\
\overset{(b)}{=}& { \Tr\left\{ {\left({{\lvert\beta \rvert}^2 \textbf{I}_s(M,N)}\right)^{-1}} {{\bf{T}}_{\rm{D}}}(M,N,\beta) \right\}}\\
=&\frac{1}{3}{\pi ^2}  \Tr\!\left\{ {{{\bf{I}}_s}{{(M,N)}^{ - 1}}\left[\! {\begin{array}{*{20}{c}}
		{\frac{{(M - 1)(M - 3)}}{{{M^2}}}}&0\\
		0&{\frac{{(N - 1)(N - 3)}}{{{N^2}}}}
		\end{array}} \!\!\!\right]} \right\}.\nonumber
\end{align}
In \eqref{eq_rtrcalsp1t}, Step (a) is due to the definition of ${\bf{a}}_s (M,N)$:
\begin{equation}\label{eq_as}
\begin{aligned}
{\bf{a}}_s (M,N) &\triangleq \frac{1}{\bar \beta} \left[ {\begin{matrix}
	{{{\bf{B}}^{\rm{T}}}{{\bf{A}}^{\!-1}}}&\!\!\!\!\!{{\bf{-}}{{\bf{J}}_2}}
	\end{matrix}} \right]\left[ {\begin{matrix}
	{{{\bf{b}}_{\text{T}}}}\\
	{{{\bf{c}}_{\text{T}}}}
	\end{matrix}} \right]\\
&= \frac{1}{\bar \beta}\left( {\frac{1}{{\left\| {{{\bf{g}}_k}} \right\|_2^2}}{{[{\bf{g}}_k^{\rm{H}}{{{\bf{\tilde g}}}_{k1}},{\bf{g}}_k^{\rm{H}}{{{\bf{\tilde g}}}_{k2}}]}^\text{H}} - {{\bf{c}}_{\rm{T}}}}\right)\\
&\overset{(c)}{=} \frac{1}{\bar \beta}\left( \frac{\bar \beta}{{\left\| {{{\bf{g}}_k}} \right\|_2^2}} \left[\begin{matrix}{\bf{\breve g}}_{k1}^\text{H}{\bf{g}}_k\\{\bf{\breve g}}_{k2}^\text{H}{\bf{g}}_k\end{matrix}\right]
-j \pi \bar \beta \left[\begin{matrix}\frac{M-1}{M}\\ \frac{N-1}{N}\end{matrix}\right]^\text{T}\right)\\
&= \left( \frac{1}{{\left\| {{{\bf{g}}_k}} \right\|_2^2}} \left[\begin{matrix}{\bf{\breve g}}_{k1}^\text{H}{\bf{g}}_k\\{\bf{\breve g}}_{k2}^\text{H}{\bf{g}}_k\end{matrix}\right]
-j \pi \left[\begin{matrix}\frac{M-1}{M}\\ \frac{N-1}{N}\end{matrix}\right]\right),
\end{aligned}
\end{equation}where Step (c) is due to the combination of \eqref{eq_gkb} and \eqref{eq_rTD}. In \eqref{eq_as}, $\textbf{a}_s(M,N)$ is unrelated to $\beta$ because none of $\textbf{g}_k$, $\breve{\textbf{g}}_{k,1}$, and $\breve{\textbf{g}}_{k,2}$ in \eqref{eq_as} is related to $\beta$. In \eqref{eq_rtrcalsp2t}, Step (b) is obtained by substituting \eqref{eq_Ip} into \eqref{eq_rtrcalsp2t}.

Substituting \eqref{eq_rtrcalsp1t} and \eqref{eq_rtrcalsp2t} into \eqref{eq_rtrcal}, we can obtain:
\begin{align}\label{eq_calsp}
\quad &\Tr\left\{ {{{\bf{I}}_{i{p_2}}}(M,N){\bf{T}}(M,N,\beta )} \right\}\nonumber\\
=&MN\Tr\left\{ {{\bf{a}}_{\rm{s}}^\text{H}(M,N){{\bf{I}}_s}{{(M,N)}^{ - 1}}{{\bf{a}}_{\rm{s}}}(M,N)} \right\}\\
+&\frac{\pi ^2MN}{3}\Tr\!\left\{\! {{{\bf{I}}_s}{{(M,N)}^{ \!- 1}}\!\left[\!\! {\begin{array}{*{20}{c}}
		{\frac{{(M - 1)(M - 3)}}{{{M^2}}}}&0\\
		0&{\frac{{(N - 1)(N - 3)}}{{{N^2}}}}
		\end{array}} \!\!\right]} \!\right\}\nonumber,
\end{align}which reveal that $\Tr\left\{ {{{\bf{I}}_{i{p_2}}}(M,N){\bf{T}}(M,N,\beta)}\right\}$ is irrelevant to channel coefficient $\beta$.

Since other parts except for $\Tr\left\{ {{{\bf{I}}_{i{p_2}}}(M,N){\bf{T}}(M,N,\beta)}\right\}$ in \eqref{eq_CMMSETemp} are also irrelevant to channel coefficient $\beta$, the minimum channel vector MSE $I_{\min}(\boldsymbol{\psi})$ is unrelated to $\beta$ and the optimal direction parameter offsets $\boldsymbol{\Delta}_{1}^{*},\boldsymbol{\Delta}_{2}^{*},\boldsymbol{\Delta}_{3}^{*}$ are invariant to channel coefficient $\beta$.

\emph{\textbf{Step 2}: We prove that $\boldsymbol{\Delta}_{1}^*,\boldsymbol{\Delta}_{2}^*,\boldsymbol{\Delta}_{3}^*$ are unrelated to direction parameter vector $\textbf{x}$.}

Since the analog beamforming vectors are steering vectors, i.e., $\textbf{w}_{k,i} = \frac{1}{\sqrt{MN}}\textbf{a}\left(\textbf{x}+\boldsymbol{\Delta}_{k,i}\right)$, where $\boldsymbol{\Delta}_{k,i}=\left[\delta_{k,i1},\delta_{k,i2}\right]^\text{T}$ denotes the $i$-th direction parameter offset, the $i$-th ($i=1,2,3$) element of ${{{\bf{g}}}_{k}}$ and ${{{\bf{\tilde g}}}_{k1}}$ defined in the Fisher information matrix \eqref{eq_fisher} can be rewritten as \eqref{eq_gki} and \eqref{eq_gkti1}:
\vspace{-1mm}
\begin{align}\label{eq_gki}
&{[{{\bf{g}}_k}]_i} = \frac{1}{{\sqrt {MN} }}\sum\limits_{m = 1}^M {\sum\limits_{n = 1}^N {{e^{-j2\pi \left[ {\frac{{(m - 1)\delta_{k,i1}}}{M} + \frac{{(n - 1)\delta_{k,i2}}}{N}} \right]}}} } \\
&= \frac{1}{{\sqrt {MN} }}\frac{{\sin (\pi \delta_{k,i1})}}{{\sin \left( { \frac{\pi \delta_{k,i1}}{M}} \right)}}\frac{\sin (\pi \delta_{k,i2})}{{\sin \left( {\frac{\pi \delta_{k,i2}}{N}} \right)}}{e^{-j\pi \left[ {\frac{{M - 1}}{M}\delta_{k,i1} + \frac{{N - 1}}{N}\delta_{k,i2}} \right]}}. \nonumber
\end{align}
\begin{figure*}[!t]
	\normalsize
	\vspace{-5mm}
	\begin{equation}\label{eq_gkti1}
	\begin{aligned}
	{[{{{\bf{\tilde g}}}_{k1}}]_i} &= \beta {\bf{w}}_{k,i}^{\text{H}}\frac{{\partial {\bf{a}}({\bf{x}})}}{{\partial {x_1}}} = \frac{\beta }{{\sqrt {MN} }}\left( {\sum\limits_{m = 1}^M {\sum\limits_{n = 1}^N {j2\pi \frac{{m - 1}}{M}{e^{-j2\pi \left[ {\frac{{(m - 1)\delta_{k,i1}}}{M} + \frac{{(n - 1)\delta_{k,i2}}}{N}} \right]}}} } } \right)\\
	&= \frac{{j2\pi \beta }}{{M\sqrt {MN} }}\left( {\frac{{\sin (\pi \delta_{k,i2})}}{{\sin \left( {\frac{\pi \delta_{k,i2} }{N}} \right)}}{e^{-j\pi \frac{{N - 1}}{N}\delta_{k,i2}}}\frac{{(M - 1){e^{-j2\pi \delta_{k,i1}}} - M{e^{-j2\pi \frac{{M - 1}}{M}\delta_{k,i1}}} + 1}}{{{{\left[ {1 - {e^{-j2\pi \frac{{\delta_{k,i1}}}{M}}}} \right]}^2}}}{e^{-j2\pi \frac{{\delta_{k,i1}}}{M}}}} \right).
	\end{aligned}
	\end{equation}
	\hrulefill
	\vspace*{0pt}
\end{figure*}

As shown in \eqref{eq_gki} and \eqref{eq_gkti1}, both $\textbf{g}_k$ and $\tilde{\textbf{g}}_{k1}$ have nothing to do with the direction parameter vector $\textbf{x}=\left[x_1,x_2\right]^\text{T}$, which is also feasible to  $\tilde{\textbf{g}}_{k,2}$. Therefore, the whole Fisher information matrix $\textbf{I}({\boldsymbol{\psi}}, \textbf{W})$ \eqref{eq_fisher} has nothing to do with $\textbf{x}$. In addition, $\textbf{T}(M,N,\beta)$ in \eqref{eq_sumrank1} is unrelated to $\textbf{x}$. Hence, the minimum CRLB in \eqref{eq_CMMSE} has nothing to do with $\textbf{x}$ and the optimal direction parameter offsets $\boldsymbol{\Delta}_{1}^{*},\boldsymbol{\Delta}_{2}^{*},\boldsymbol{\Delta}_{3}^{*}$ are invariant to the direction parameter vector $\textbf{x}=\left[x_1,x_2\right]^\text{T}$.

\emph{\textbf{Step 3}: We prove that $\boldsymbol{\Delta}_{1}^*,\boldsymbol{\Delta}_{2}^*,\boldsymbol{\Delta}_{3}^*$ converge to constant values as $M,N \to +\infty$.}

Let us go into the asymptotic features of \eqref{eq_CMMSE}. By \eqref{eq_gki} and \eqref{eq_gkti1}, when antenna number $\emph{M}$,\,$\emph{N} \to +\infty$, the limit of $i$-th ($i=1,2,3$) element of ${{{\bf{g}}}_{k}}$ and ${{{\bf{\tilde g}}}_{k,1}}$ are given as follows:
\vspace{-1mm}
%\begin{small}
\begin{align}\label{eq_gkil}
&\mathop {\lim }\limits_{M.N \to  + \infty } \frac{{{{[{{\bf{g}}_{\bf{k}}}]}_i}}}{{\sqrt {MN} }}
= \operatorname{Sa}\left(\pi \delta_{k,i1}\right)\operatorname{Sa}[\pi \delta_{k,i2}]{e^{{\rm{ - }}j\pi \left(\delta_{k,i1}+\delta_{k,i2}\right)}}.
\end{align}
\begin{align}\label{eq_gktil1}
&\mathop {\lim }\limits_{M,N \to  + \infty } \frac{{{{[{{{\bf{\tilde g}}}_{k1}}]}_i}}}{{\sqrt {MN} }}\\
&= j2\pi \beta Sa[\pi \delta_{k,i2}]{e^{{\rm{ - }}j\pi \delta_{k,i2}}}\frac{{{e^{{\rm{ - }}j2\pi \delta_{k,i1}}}\left(1{\rm{ + }}j2\pi \delta_{k,i1}\right) - 1}}{{{{\left(2\pi \delta_{k,i1}\right)}^2}}}.\nonumber
\end{align}By \eqref{eq_gkil}, we can obtain that
\begin{align}\label{eq_gkv}
\mathop {\lim}\limits_{M.N \to  + \infty } \frac{\left\|{\bf{g}_k}\right\|_2^2}{MN} = \sum\limits_{i = 1}^3 {\operatorname{Sa}^2\left(\pi \delta_{k,i1}\right)\operatorname{Sa}^2\left(\pi \delta_{k,i2}\right)}.
\end{align}

Hence, the first element of ${\textbf{I}}(\boldsymbol{\psi} ,{{\bf{W}}_k})/{MN}$ in \eqref{eq_fisher} converges when $M,\,N \to +\infty$. Similar to that, other elements of ${\bf{I}}(\psi ,{{\bf{W}}_k})/{MN}$ in \eqref{eq_fisher} also converge. Thus, the whole matrix ${\bf{I}}(\psi,{{\bf{W}}_k})/{MN}$ converge as $M,\,N \to +\infty$, the limit defined as follows:
\vspace{-1mm}
\begin{equation}\label{eq_IL}
\begin{aligned}
{\textbf{I}}_\text{L}(\boldsymbol{\psi} ,{{\bf{W}}_k}) \triangleq \lim\limits_{M,N \to  + \infty }  \frac{1}{MN}{\textbf{I}}(\boldsymbol{\psi} ,{{\bf{W}}_k}).
\end{aligned}
\end{equation}
\vspace{-1mm}

The limit of ${{\bf{T}}(M,N,\beta )}$ in \eqref{eq_sumrank1} is given as:
\vspace{-1mm}
\begin{equation}\label{eq_TL}
\begin{aligned}
\begin{array}{l}
\mathop {\lim }\limits_{M,N \to  + \infty } \frac{1}{MN}{{\bf{T}}(M,N,\beta)}\\
= \left[ {\begin{array}{*{20}{c}}
	1&j&{j\pi \beta }&{j\pi \beta }\\
	{ - j}&1&{\pi \beta }&{\pi \beta }\\
	{ - j\pi \bar \beta }&{\pi \bar \beta }&{\frac{4}{3}{\pi ^2}{{\left| \beta  \right|}^2}}&{{\pi ^2}{{\left| \beta  \right|}^2}}\\
	{ - j\pi \bar \beta }&{\pi \bar \beta }&{{\pi ^2}{{\left| \beta  \right|}^2}}&{\frac{4}{3}{\pi ^2}{{\left| \beta  \right|}^2}}
	\end{array}} \right]\\
\triangleq {{\bf{T}}_{\rm{L}}}(\beta )
\end{array}
\end{aligned}
\end{equation}
\vspace{-1mm}

Combine \eqref{eq_CMMSE}, \eqref{eq_IL}, and \eqref{eq_TL} , we obtain the limit of $I_{\min}(\boldsymbol{\psi})$ in \eqref{eq_CMMSE} as $M,\,N \to + \infty$:
\vspace{-1mm}
\begin{equation}\label{eq_CMMSEL}
\begin{aligned}
&\mathop {\lim }\limits_{M,N \to  + \infty } \left( {MN \times I_{\min}(\boldsymbol{\psi})} \right)\\
{\rm{ = }}&\mathop {\lim }\limits_{M,N \to  + \infty } \Tr\left\{ {{{(k{\bf{I}}(\boldsymbol{\psi} ,{{\bf{W}}^*}))}^{ - 1}}\sum\limits_{m = 1}^M {\sum\limits_{n = 1}^N {{\textbf{v}_{mn}^{\rm{H}}}\textbf{v}_{mn}} } } \right\}\\
{\rm{  = }}&\mathop {\lim }\limits_{M,N \to  + \infty } \Tr\left\{ {{{\left( {k M N{{\textbf{I}}_{\text{L}}(\boldsymbol{\psi} ,{{\bf{W}}^*})}} \right)}^{ - 1}}MN{{\bf{T}}_\text{L}}(\beta )} \right\}{\rm{ }}\\
{\rm{ = }}&\Tr\left\{ {{{\left( {k{\textbf{I}}_{\text{L}}(\boldsymbol{\psi} ,{{\bf{W}}^*})} \right)}^{ - 1}}{{\bf{T}}_\text{L}}(\beta )} \right\},
\end{aligned}
\end{equation}
\vspace{-1mm}which reveals that the optimal analog beamforming matrix converges, i.e, the optimal direction parameter offsets $\boldsymbol{\Delta}_{1}^*,\boldsymbol{\Delta}_{2}^*,\boldsymbol{\Delta}_{3}^*$ converge to constant values determined by \eqref{eq_CMMSEL}. 

Therefore, Lemma \ref{UnifiedOptShift} gets proved.

\section{Proof of Theorem~\ref{Converge to unique stable point}}\label{proof_Converge to unique stable point}
Recall the beam and channel tracking procedure in \eqref{eq_rtracking}. 
Since $\textbf{z}_k \triangleq \left[z_{k,1},z_{k,2},z_{k,3} \right]$ in \eqref{eq_z} is composed of three \emph{i.i.d.} circularly symmetric complex Gaussian random variables, the expectation of $\hat{\textbf{z}}_k$ is $\mathbb{E}\left[\hat{\textbf{z}}_k\right] = \textbf{0}$ and the covariance matrix is given by \eqref{eq_rzc}, where Step $(a)$ is obtained as follows:

\begin{figure*}[!t]
	\normalsize
	\vspace{-6mm}
	\begin{equation}\label{eq_rzc}
	\begin{aligned}
	\mathbb{E}\left[ \left(\hat{\mathbf{z}}_k - \mathbb{E}\left[ \hat{\mathbf{z}}_k \right] \right) \left(\hat{\mathbf{z}}_k - \mathbb{E}\left[ \hat{\mathbf{z}}_k \right] \right)^\text{T}\right] &~=\frac{4}{\sigma^4}\textbf{I}\left(\hat{\boldsymbol{\psi}}_{k-1}, \textbf{W}_k\right)^\text{-1} \mathbb{E}\!\left\{\!\!\left[\begin{matrix} \operatorname{Re}\{{s_p^\text{H}}{\mathbf{e}}_{k}^\text{H}\mathbf{z}_k\} \\
	\operatorname{Im}\{{s_p^\text{H}}{\mathbf{e}}_{k}^\text{H}\mathbf{z}_k\} \\
	\operatorname{Re}\{{s_p^\text{H}}\tilde{\mathbf{e}}_{k1}^\text{H}\mathbf{z}_k\}\\
	\operatorname{Re}\{{s_p^\text{H}}\tilde{\mathbf{e}}_{k2}^\text{H}\mathbf{z}_k\}
	\end{matrix}\right]\!\!\cdot\!\!\left[\begin{matrix} \operatorname{Re}\{{s_p^\text{H}}{\mathbf{e}}_{k}^\text{H}\mathbf{z}_k\} \\
	\operatorname{Im}\{{s_p^\text{H}}{\mathbf{e}}_{k}^\text{H}\mathbf{z}_k\} \\
	\operatorname{Re}\{{s_p^\text{H}}\tilde{\mathbf{e}}_{k1}^\text{H}\mathbf{z}_k\}\\
	\operatorname{Re}\{{s_p^\text{H}}\tilde{\mathbf{e}}_{k2}^\text{H}\mathbf{z}_k\} \end{matrix}\right]^\text{\!\!T} \!\right\} \textbf{I}\left(\hat{\boldsymbol{\psi}}_{k-1}, \textbf{W}_k\right)^\text{-1}\\
	&~ \buildrel  {(a)} \over =  \textbf{I}\left(\hat{\boldsymbol{\psi}}_{k-1}, \textbf{W}_k\right)^\text{-1}
	\end{aligned}
	\end{equation}
	\vspace{-2mm}
	\hrulefill
	\vspace*{-2mm}
\end{figure*}

\begin{itemize}
	\item Since $\mathbf{z}_k\!=\!\left[ z_{k,1}, z_{k,2}, z_{k,3} \right]^\text{T}$ consists of three  \emph{i.i.d.} circularly symmetric complex Gaussian random variables, we  get
	\begin{equation}\label{eq:variance}
	\left\{\begin{array}{*{20}{l}}
	s_p^\text{H}{\mathbf{e}}_{k}^\text{H}\mathbf{z}_k \sim \mathcal{CN}\left( 0, \left\|{s_p}{\mathbf{e}}_{k}\right\|_2^2 \sigma^2 \right)\\
	s_p^\text{H}\tilde{\mathbf{e}}_{k1}^\text{H}\mathbf{z}_k \sim \mathcal{CN}\left( 0, \left\|{s_p}\tilde{\mathbf{e}}_{k1}\right\|_2^2 \sigma^2 \right).\\
	s_p^\text{H}\tilde{\mathbf{e}}_{k2}^\text{H}\mathbf{z}_k \sim \mathcal{CN}\left( 0, \left\|{s_p}\tilde{\mathbf{e}}_{k2}\right\|_2^2 \sigma^2 \right)\end{array}\right.
	\end{equation}
	\item splitting the real part and imaginary part, we obtain
	\begin{equation}\label{eq:real_imaginary}
	\!\!\!\!\!\!\!\!\!\left\{
	\begin{aligned}
	&\operatorname{Re}\{{s_p^\text{H}}{\mathbf{e}}_{k}^\text{H}\mathbf{z}_k\} \!=\! \operatorname{Re}\{{s_p^\text{H}}{\mathbf{e}}_{k}^\text{H}\}\operatorname{Re}\{\mathbf{z}_k\} \!-\! \operatorname{Im}\{{s_p^\text{H}}{\mathbf{e}}_{k}^\text{H}\}\operatorname{Im}\{\mathbf{z}_k\}, \\
	&\operatorname{Im}\{{s_p^\text{H}}{\mathbf{e}}_{k}^\text{H}\mathbf{z}_k\} \!=\! \operatorname{Re}\{{s_p^\text{H}}{\mathbf{e}}_{k}^\text{H}\}\operatorname{Im}\{\mathbf{z}_k\} \!+\! \operatorname{Im}\{{s_p^\text{H}}{\mathbf{e}}_{k}^\text{H}\}\operatorname{Re}\{\mathbf{z}_k\}, \\
	&\operatorname{Re}\{{s_p^\text{H}}\tilde{\mathbf{e}}_{k1}^\text{H}\mathbf{z}_k\} \!=\! \operatorname{Re}\{{s_p^\text{H}}\tilde{\mathbf{e}}_{k1}^\text{H}\}\operatorname{Re}\{\mathbf{z}_k\} \!-\! \operatorname{Im}\{{s_p^\text{H}}\tilde{\mathbf{e}}_{k1}^\text{H}\}\operatorname{Im}\{\mathbf{z}_k\}, \\
	&\operatorname{Re}\{{s_p^\text{H}}\tilde{\mathbf{e}}_{k2}^\text{H}\mathbf{z}_k\} \!=\! \operatorname{Re}\{{s_p^\text{H}}\tilde{\mathbf{e}}_{k2}^\text{H}\}\operatorname{Re}\{\mathbf{z}_k\} \!-\! \operatorname{Im}\{{s_p^\text{H}}\tilde{\mathbf{e}}_{k2}^\text{H}\}\operatorname{Im}\{\mathbf{z}_k\}, \\
	&\operatorname{Re}\{{s_p^\text{H}}{\mathbf{e}}_{k}^\text{H}{s_p}\tilde{\mathbf{e}}_{k1}\} \!=\! |{s_p}|^2\operatorname{Re}\{{\mathbf{e}}_{k}^\text{H}\tilde{\mathbf{e}}_{k1}\} \\
	&~~~~~~~~~~\!= \operatorname{Re}\{{s_p^\text{H}}{\mathbf{e}}_{k}^\text{H}\}\operatorname{Re}\{{s_p}\tilde{\mathbf{e}}_{k1}\} \!+\! \operatorname{Im}\{{s_p^\text{H}}{\mathbf{e}}_{k}^\text{H}\}\operatorname{Im}\{{s_p}\tilde{\mathbf{e}}_{k1}\}, \\
	&\operatorname{Re}\{{s_p^\text{H}}{\mathbf{e}}_{k}^\text{H}{s_p}\tilde{\mathbf{e}}_{k2}\} \!=\! |{s_p}|^2\operatorname{Re}\{{\mathbf{e}}_{k}^\text{H}\tilde{\mathbf{e}}_{k2}\} \\
	&~~~~~~~~~~\!= \operatorname{Re}\{{s_p^\text{H}}{\mathbf{e}}_{k}^\text{H}\}\operatorname{Re}\{{s_p}\tilde{\mathbf{e}}_{k2}\} \!+\! \operatorname{Im}\{{s_p^\text{H}}{\mathbf{e}}_{k}^\text{H}\}\operatorname{Im}\{{s_p}\tilde{\mathbf{e}}_{k2}\}, \\
	&\operatorname{Im}\{{s_p^\text{H}}{\mathbf{e}}_{k}^\text{H}{s_p}\tilde{\mathbf{e}}_{k1}\} \!=\! |{s_p}|^2\operatorname{Im}\{{\mathbf{e}}_{k}^\text{H}\tilde{\mathbf{e}}_{k1}\} \\
	&~~~~~~~~~~\!= \operatorname{Re}\{{s_p^\text{H}}{\mathbf{e}}_{k}^\text{H}\}\operatorname{Im}\{{s_p}\tilde{\mathbf{e}}_{k1}\} \!+\! \operatorname{Im}\{{s_p^\text{H}}{\mathbf{e}}_{k}^\text{H}\}\operatorname{Re}\{{s}\tilde{\mathbf{e}}_{k1}\},\\
	&\operatorname{Im}\{{s_p^\text{H}}{\mathbf{e}}_{k}^\text{H}{s_p}\tilde{\mathbf{e}}_{k2}\} \!=\! |{s_p}|^2\operatorname{Im}\{{\mathbf{e}}_{k}^\text{H}\tilde{\mathbf{e}}_{k2}\} \\
	&~~~~~~~~~~\!= \operatorname{Re}\{{s_p^\text{H}}\hat{\mathbf{e}}_{k}^\text{H}\}\operatorname{Im}\{{s_p}\tilde{\mathbf{e}}_{k2}\} \!+\! \operatorname{Im}\{{s_p^\text{H}}{\mathbf{e}}_{k}^\text{H}\}\operatorname{Re}\{{s}\tilde{\mathbf{e}}_{k2}\},\\
	&\operatorname{Re}\{{s_p^\text{H}}\tilde{\mathbf{e}}_{k1}^\text{H}{s_p}\tilde{\mathbf{e}}_{k2}\} \!=\! |{s_p}|^2\operatorname{Re}\{\tilde{\mathbf{e}}_{k}^\text{H}\tilde{\mathbf{e}}_{k2}\} \\
	&~~~~~~~~~~\!= \operatorname{Re}\{{s_p^\text{H}}\tilde{\mathbf{e}}_{k1}^\text{H}\}\operatorname{Re}\{{s_p}\tilde{\mathbf{e}}_{k2}\} \!+\! \operatorname{Im}\{{s_p^\text{H}}\tilde{\mathbf{e}}_{k1}^\text{H}\}\operatorname{Im}\{{s_p}\tilde{\mathbf{e}}_{k2}\}
	\end{aligned}\right.\!\!\!\!\!.
	\end{equation}
	
	\item Combining \eqref{eq:variance} and \eqref{eq:real_imaginary}, we can obtain
	\begin{equation}
	\!\!\!\!\!\!\left\{
	\begin{aligned}
	&\mathbb{E}\left[\operatorname{Re}\{{s_p^\text{H}}{\mathbf{e}}_{k}^\text{H}\mathbf{z}_k\}^2\right] \!=\! \mathbb{E}\left[\operatorname{Im}\{{s_p^\text{H}}{\mathbf{e}}_{k}^\text{H}\mathbf{z}_k\}^2\right] \!=\! \frac{|{s_p}|^2 \sigma^2}{2}\left\|{\mathbf{e}}_{k}\right\|_2^2, \\
	&\mathbb{E}\left[\operatorname{Re}\{{s_p^\text{H}}{\mathbf{e}}_{k}^\text{H}\mathbf{z}_k\}\!\cdot\!\operatorname{Im}\{{s_p^\text{H}}{\mathbf{e}}_{k}^\text{H}\mathbf{z}_k\}\right] = 0, \\
	&\mathbb{E}\left[\operatorname{Re}\{{s_p^\text{H}}{\mathbf{e}}_{k}^\text{H}\mathbf{z}_k\}\!\cdot\!\operatorname{Re}\{{s_p^\text{H}}\tilde{\mathbf{e}}_{k1}^\text{H}\mathbf{z}_k\} \right] = \frac{|{s_p}|^2 \sigma^2}{2}\operatorname{Re}\{{\mathbf{e}}_{k}^\text{H}\tilde{\mathbf{e}}_{k1}\} , \\
	&\mathbb{E}\left[\operatorname{Re}\{{s_p^\text{H}}{\mathbf{e}}_{k}^\text{H}\mathbf{z}_k\}\!\cdot\!\operatorname{Re}\{{s_p^\text{H}}\tilde{\mathbf{e}}_{k2}^\text{H}\mathbf{z}_k\} \right] = \frac{|{s_p}|^2 \sigma^2}{2}\operatorname{Re}\{{\mathbf{e}}_{k}^\text{H}\tilde{\mathbf{e}}_{k2}\} , \\
	&\mathbb{E}\left[\operatorname{Im}\{{s_p^\text{H}}{\mathbf{e}}_{k}^\text{H}\mathbf{z}_k\}\!\cdot\!\operatorname{Re}\{{s_p^\text{H}}\tilde{\mathbf{e}}_{k1}^\text{H}\mathbf{z}_k\} \right] = \frac{|{s_p}|^2 \sigma^2}{2}\operatorname{Im}\{{\mathbf{e}}_{k}^\text{H}\tilde{\mathbf{e}}_{k1}\},  \\
	&\mathbb{E}\left[\operatorname{Im}\{{s_p^\text{H}}{\mathbf{e}}_{k}^\text{H}\mathbf{z}_k\}\!\cdot\!\operatorname{Re}\{{s_p^\text{H}}\tilde{\mathbf{e}}_{k2}^\text{H}\mathbf{z}_k\} \right] = \frac{|{s_p}|^2 \sigma^2}{2}\operatorname{Im}\{{\mathbf{e}}_{k}^\text{H}\tilde{\mathbf{e}}_{k2}\},  \\
	&\mathbb{E}\left[\operatorname{Re}\{{s_p^\text{H}}\tilde{\mathbf{e}}_{k1}^\text{H}\mathbf{z}_k\}^2\right] \!=\! \frac{|{s_p}|^2 \sigma^2}{2}\left\|\tilde{\mathbf{e}}_{k1}\right\|_2^2, \\
	&\mathbb{E}\left[\operatorname{Re}\{{s_p^\text{H}}\tilde{\mathbf{e}}_{k2}^\text{H}\mathbf{z}_k\}^2\right] \!=\! \frac{|{s_p}|^2 \sigma^2}{2}\left\|\tilde{\mathbf{e}}_{k2}\right\|_2^2, \\
	&\mathbb{E}\left[\operatorname{Re}\{{s_p^\text{H}}\tilde{\mathbf{e}}_{k1}^\text{H}\mathbf{z}_k\}\!\cdot\!\operatorname{Re}\{{s_p^\text{H}}\tilde{\mathbf{e}}_{k2}^\text{H}\mathbf{z}_k\} \right] = \frac{|{s_p}|^2 \sigma^2}{2}\operatorname{Re}\{\tilde{\mathbf{e}}_{k1}^\text{H}\tilde{\mathbf{e}}_{k2}\}.
	\end{aligned}\right.\!\!\!\!\!
	\end{equation}
	Hence, we have
	\begin{align}\label{eq_expectation_noise}
	\mathbb{E}\!\left\{\!\!\left[\begin{matrix} \operatorname{Re}\{{s_p^\text{H}}{\mathbf{e}}_{k}^\text{H}\mathbf{z}_k\} \\
	\operatorname{Im}\{{s_p^\text{H}}{\mathbf{e}}_{k}^\text{H}\mathbf{z}_k\} \\
	\operatorname{Re}\{{s_p^\text{H}}\tilde{\mathbf{e}}_{k1}^\text{H}\mathbf{z}_k\}\\
	\operatorname{Re}\{{s_p^\text{H}}\tilde{\mathbf{e}}_{k2}^\text{H}\mathbf{z}_k\}
	\end{matrix}\right]\!\!\cdot\!\!\left[\begin{matrix} \operatorname{Re}\{{s_p^\text{H}}{\mathbf{e}}_{k}^\text{H}\mathbf{z}_k\} \\
	\operatorname{Im}\{{s_p^\text{H}}{\mathbf{e}}_{k}^\text{H}\mathbf{z}_k\} \\
	\operatorname{Re}\{{s_p^\text{H}}\tilde{\mathbf{e}}_{k1}^\text{H}\mathbf{z}_k\}\\
	\operatorname{Re}\{{s_p^\text{H}}\tilde{\mathbf{e}}_{k2}^\text{H}\mathbf{z}_k\} \end{matrix}\right]^\text{\!\!T} \!\right\} \!=\!\frac{\sigma^4}{4} \mathbf{I}(\hat{\boldsymbol{\psi}}_{k\!-\!1},\!\mathbf{W}_k).
	\end{align}
	
	\item Substituting \eqref{eq_expectation_noise} into \eqref{eq_rzc} yields the result of Step $(a)$.
\end{itemize}

Assume $\{ \mathcal{G}_k: k \ge 0 \}$ is an increasing sequence of $\sigma$-fields of $\{ \hat{\boldsymbol{\psi}}_{0}, \hat{\boldsymbol{\psi}}_{1}, \hat{\boldsymbol{\psi}}_{2}, \ldots \}$, i.e., $\mathcal{G}_{k-1}\!\subset\!\mathcal{G}_k$, where $\mathcal{G}_{0} \!\overset{\Delta}{=}\! \sigma(\hat{\boldsymbol{\psi}}_{0})$ and $\mathcal{G}_k \!\overset{\Delta}{=}\! \sigma(\hat{\boldsymbol{\psi}}_{0}, \hat{\mathbf{z}}_{1}, \ldots, \hat{\mathbf{z}}_{k}) $ for $k \ge 1$. Because the $\hat{\mathbf{z}}_k$'s are composed of \emph{i.i.d.} circularly symmetric complex Gaussian random variables with zero mean, $\hat{\mathbf{z}}_k$ is independent of $\mathcal{G}_{k-1}$, and $\hat{\boldsymbol{\psi}}_{k-1} \!\in\! \mathcal{G}_{k-1}$. Hence, we have
\begin{align}\label{eq_fil2}
&~\mathbb{E} \left[ \left. \mathbf{f}\left(\hat{\boldsymbol{\psi}}_{k-1}, \boldsymbol{\psi}\right) + \hat{\mathbf{z}}_k \right| \mathcal{G}_{k-1} \right] \\
= &~\mathbb{E} \left[ \left. \mathbf{f}\left(\hat{\boldsymbol{\psi}}_{k-1}, \boldsymbol{\psi}\right) \right| \mathcal{G}_{k-1} \right] + \mathbb{E} \left[ \left. \hat{\mathbf{z}}_k \right| \mathcal{G}_{k-1} \right] = \mathbf{f}\left(\hat{\boldsymbol{\psi}}_{k-1}, \boldsymbol{\psi}\right), \nonumber
\end{align}
for $k \ge 1$.

Theorem 5.2.1 in \cite[Section 5.2.1]{kushner2003stochastic} gives the conditions that ensure $\hat{\textbf{x}}_k$ converges to a unique point when there are several stable points with probability one. Next, we will prove that if the step-size $b_k$ is given by \eqref{eq_stepsize} with any $\varepsilon > 0$ and $K_0 \ge 0$, the joint beam and channel tracking algorithm in \eqref{eq_Tracking} satisfies the corresponding conditions below:
\begin{itemize}
	\item[1)] Step-size requirements:
	\begin{equation}\left\{\begin{aligned}&b_k = \frac{\varepsilon}{k + K_0} \rightarrow 0, \\
	& \sum\limits_{k=1}^{+\infty} b_k = \sum\limits_{k=1}^{+\infty} \frac{\varepsilon}{k+K_0} = +\infty, \\
	& \sum\limits_{k=1}^{+\infty} b_k^2 = \sum\limits_{k=1}^{+\infty} \frac{\varepsilon^2}{(k+K_0)^2} \le \sum\limits_{l=1}^{+\infty} \frac{\varepsilon^2}{l^2} < {+\infty}.\end{aligned}\right.\end{equation}
	
	\item[2)] It is necessary to prove that $\sup\nolimits_k \mathbb{E} \left[ \left\|\mathbf{f}\left(\hat{\boldsymbol{\psi}}_{k-1}, \boldsymbol{\psi}\right) + \hat{\mathbf{z}}_k\right\|_2^2 \right] < +\infty$. \\
	From \eqref{eq_rtracking} and \eqref{eq_rzc}, we have
	\begin{align}\label{eq_expectation_yk} & \mathbb{E} \left[ \left\|\mathbf{f}\left(\hat{\boldsymbol{\psi}}_{k-1}, \boldsymbol{\psi}\right) + \hat{\mathbf{z}}_k\right\|_2^2 \right] \\
	= & \mathbb{E} \left[ \left\|\mathbf{f}\left(\hat{\boldsymbol{\psi}}_{k-1}, \boldsymbol{\psi}\right)\right\|_2^2 + 2 \mathbf{f}\left(\hat{\boldsymbol{\psi}}_{k-1}, \boldsymbol{\psi}\right)^\text{T} \hat{\mathbf{z}}_k + \left\|\hat{\mathbf{z}}_k\right\|_2^2 \right] \nonumber \\ \overset{(a)}{=} & \mathbb{E} \left[ \left\|\mathbf{f}\left(\hat{\boldsymbol{\psi}}_{k-1}, \boldsymbol{\psi}\right)\right\|_2^2 \right] + \operatorname{tr}\left\{\mathbf{I}(\hat{\boldsymbol{\psi}}_{k\!-\!1},\!\mathbf{W}_k)^{-1}\right\},\nonumber \end{align}
	where Step $(a)$ is due to \eqref{eq_rzc} and that $\hat{\mathbf{z}}_k$ is independent of $\mathbf{f}\left(\hat{\boldsymbol{\psi}}_{n-1}, \boldsymbol{\psi}\right)$. \\
	From \eqref{eq_f}, we have
	\begin{align}\label{eq_ub_fx}\left\|\mathbf{f}\left(\hat{\boldsymbol{\psi}}_{k-1}, \boldsymbol{\psi}\right)\right\|_2^2 \le &~\left\|\mathbf{I}(\hat{\boldsymbol{\psi}}_{k\!-\!1},\!\mathbf{W}_k)^{-1}\right\|_\text{F}^2 \\
	&~\!\!\!\!\!\!\!\!\!\!\!\!\!\!\!\!\!\!\!\!\!\!\!\!\!\!\!\!\!\!\!\!\!\cdot\!\left\|\frac{2|{s_p}|^2}{\sigma^2}\!\!\left[\begin{matrix} {\operatorname{Re}\left\{\textbf{e}_k^\text{H}\left(\beta_k \textbf{W}_k^\text{H} \textbf{a}\left(\textbf{x}_k\right) -\hat{\beta}_{k-1} \textbf{e}_k \right)\right\}}\\
	{\operatorname{Im}\left\{\textbf{e}_k^\text{H}\left(\beta_k \textbf{W}_k^\text{H} \textbf{a}\left(\textbf{x}_k\right)-\hat{\beta}_{k-1} \textbf{e}_k \right)\right\}}\\
	{\operatorname{Re}\left\{\tilde{\textbf{e}}_{k1}^\text{H}\left(\beta_k \textbf{W}_k^\text{H} \textbf{a}\left(\textbf{x}_k\right)-\hat{\beta}_{k-1} \textbf{e}_k \right)\right\}}\\
	{\operatorname{Re}\left\{\tilde{\textbf{e}}_{k2}^\text{H}\left(\beta_k \textbf{W}_k^\text{H} \textbf{a}\left(\textbf{x}_k\right)-\hat{\beta}_{k-1} \textbf{e}_k \right)\right\}} \end{matrix}\right]\right\|_2^2. \nonumber \end{align}
	As the Fisher information matrix is invertible, we get
	\begin{align}\label{eq:proof1_1}\left\|\mathbf{I}(\hat{\boldsymbol{\psi}}_{k\!-\!1},\!\mathbf{W}_k)^{-1}\right\|_\text{F}^2 < +\infty.\end{align}
	Besides, $\mathbf{W}_{k}\!=\!\left[\mathbf{w}_{k,1},\mathbf{w}_{k,2}, \mathbf{w}_{k,3}\right]$, ${\mathbf{e}}_{k}\!=\!\mathbf{W}_k^\text{H}\mathbf{a}(\hat{\textbf{x}}_{k\!-\!1})$, $\tilde{\textbf{e}}_{k1} = \hat{\beta}_{k-1}\textbf{W}_k^\text{H} \frac{\partial \textbf{a}\left(\hat{\textbf{x}}_{k-1}\right)}{\partial x_1}$, $\tilde{\textbf{e}}_{k2} =\hat{\beta}_{k-1} \textbf{W}_k^\text{H} \frac{\partial \textbf{a}\left(\hat{\textbf{x}}_{k-1}\right)}{\partial x_2}$, hence we have
	\begin{align}
	\begin{array}{l}
	\left| {{\bf{w}}_{k,i}^{\rm{H}}{\bf{a}}({\bf{x}})} \right|\\
	= \left| {\frac{1}{{\sqrt {MN} }}\sum\limits_{m = 1}^M {\sum\limits_{n = 1}^N {{e^{-j2\pi \left( {\frac{{(m - 1){\delta_{k,i1}}}}{M} + \frac{{(n - 1){\delta_{k,i2}}}}{N}} \right)}}} } } \right|\\
	\le \frac{1}{{\sqrt {MN} }}\sum\limits_{m = 1}^M {\sum\limits_{n = 1}^N {\left| {{e^{-j2\pi \left( {\frac{{(m - 1){\delta_{k,i1}}}}{M} + \frac{{(n - 1){\delta_{k,i2}}}}{N}} \right)}}} \right|} } \\
	= \sqrt {MN} < +\infty
	\end{array}
	\end{align}
	\begin{small}
		\begin{align}
		&\left| {{\bf{w}}_{k,i}^{\rm{H}}\frac{{\partial {\bf{a}}({\bf{x}})}}{{\partial {x_1}}}} \right|\nonumber\\
		=& \left| {\frac{1}{{\sqrt {MN} }}\sum\limits_{m = 1}^M {\sum\limits_{n = 1}^N {j2\pi \frac{{m - 1}}{M}{{e^{-j2\pi \left( {\frac{{(m - 1){\delta_{k,i1}}}}{M} + \frac{{(n - 1){\delta_{k,i2}}}}{N}} \right)}}}} } } \right|\nonumber\\
		\le & \frac{{2\pi }}{{M\sqrt {MN} }}\sum\limits_{m = 1}^M {\sum\limits_{n = 1}^N {(m - 1)\left| {{{e^{-j2\pi \left( {\frac{{(m - 1){\delta_{k,i1}}}}{M} + \frac{{(n - 1){\delta_{k,i2}}}}{N}} \right)}}}} \right|} } \nonumber\\
		=& \sqrt {MN} \left( {M - 1} \right) < +\infty,
		\end{align}
	\end{small}
	and
	\begin{small}
		\begin{align}
		&\left| {{\bf{w}}_{k,i}^{\rm{H}}\frac{{\partial {\bf{a}}({\bf{x}})}}{{\partial {x_2}}}} \right|\nonumber\\
		=& \left| {\frac{1}{{\sqrt {MN} }}\sum\limits_{m = 1}^M {\sum\limits_{n = 1}^N {j2\pi \frac{{n - 1}}{N}{{e^{-j2\pi \left( {\frac{{(m - 1){\delta_{k,i1}}}}{M} + \frac{{(n - 1){\delta_{k,i2}}}}{N}} \right)}}}} } } \right|\nonumber\\
		\le& \frac{{2\pi }}{{N\sqrt {MN} }}\sum\limits_{m = 1}^M {\sum\limits_{n = 1}^N {(n - 1)\left| {{{e^{-j2\pi \left( {\frac{{(m - 1){\delta_{k,i1}}}}{M} + \frac{{(n - 1){\delta_{k,i2}}}}{N}} \right)}}}} \right|} } \nonumber\\
		=& \sqrt {MN} \left( {N - 1} \right) < +\infty,
		\end{align}
	\end{small}for $i = 1, 2,3$ and all possible $\textbf{w}_{k,i}$ and $\textbf{x}$, thus we can get
	\begin{align}\label{eq:proof1_2}
	\left\|\frac{2|{s_p}|^2}{\sigma^2}\!\!\left[\begin{matrix} {\operatorname{Re}\left\{\textbf{e}_k^\text{H}\left(\beta_k \textbf{W}_k^\text{H} \textbf{a}\left(\textbf{x}_k\right) -\hat{\beta}_{k-1} \textbf{e}_k \right)\right\}}\\
	{\operatorname{Im}\left\{\textbf{e}_k^\text{H}\left(\beta_k \textbf{W}_k^\text{H} \textbf{a}\left(\textbf{x}_k\right)-\hat{\beta}_{k-1} \textbf{e}_k \right)\right\}}\\
	{\operatorname{Re}\left\{\tilde{\textbf{e}}_{k1}^\text{H}\left(\beta_k \textbf{W}_k^\text{H} \textbf{a}\left(\textbf{x}_k\right)-\hat{\beta}_{k-1} \textbf{e}_k \right)\right\}}\\
	{\operatorname{Re}\left\{\tilde{\textbf{e}}_{k2}^\text{H}\left(\beta_k \textbf{W}_k^\text{H} \textbf{a}\left(\textbf{x}_k\right)-\hat{\beta}_{k-1} \textbf{e}_k \right)\right\}} \end{matrix}\right]\right\|_2^2 < +\infty.
	\end{align}
	Combining \eqref{eq:proof1_1} and \eqref{eq:proof1_2}, we have
	\begin{align}\label{eq:expectation_fx}\mathbb{E} \left[ \left\|\mathbf{f}\left(\hat{\boldsymbol{\psi}}_{n-1}, \boldsymbol{\psi}\right)\right\|_2^2 \right] < +\infty. \end{align}
	According to  \eqref{eq:proof1_1}, it is clear that $ \operatorname{tr}\left\{\mathbf{I}(\hat{\boldsymbol{\psi}}_{k\!-\!1},\!\mathbf{W}_k)^{-1}\right\}  < +\infty$. Then, we can get that
	\begin{align}\sup\nolimits_k \mathbb{E} \left[ \left\|\mathbf{f}\left(\hat{\boldsymbol{\psi}}_{k-1}, \boldsymbol{\psi}\right) + \hat{\mathbf{z}}_k\right\|_2^2 \right] < +\infty.\end{align}
	
	\item[3)] The function $\mathbf{f}\left(\hat{\boldsymbol{\psi}}_{k-1}, \boldsymbol{\psi}\right)$ should be continuous with respect to $\hat{\boldsymbol{\psi}}_{k-1}$.\\
	By using \eqref{eq_f}, we know that each element of $\mathbf{f}\left(\hat{\boldsymbol{\psi}}_{k-1}, \boldsymbol{\psi}\right)$ is continuous with respect to $\hat{\boldsymbol{\psi}}_{k-1} = \left [\hat{\beta}_{re},\hat{\beta}_{im},\hat{x}_{1},\hat{x}_{2},\right]^\text{T}$. Therefore, $\mathbf{f}\left(\hat{\boldsymbol{\psi}}_{k-1}, \boldsymbol{\psi}\right)$ is continuous with respect to $\hat{\boldsymbol{\psi}}_{k-1}$.
	
	\item[4)] Let $\boldsymbol\gamma_k = \mathbb{E} \left[ \left.\mathbf{f}\left(\hat{\boldsymbol{\psi}}_{k-1}, \boldsymbol{\psi}\right) + \hat{\mathbf{z}}_k\right| \mathcal{G}_{k-1} \right] - \mathbf{f}\left(\hat{\boldsymbol{\psi}}_{k-1}, \boldsymbol{\psi}\right)$. We need to prove that $\sum_{k=1}^{+\infty} \left\| b_k \boldsymbol\gamma_k \right\|_2 < +\infty$ with probability one. \\
	From (\ref{eq_fil2}), we get $\boldsymbol\gamma_k = \mathbf{0}$ for all $k \ge 1$. So we have $\sum_{k=1}^{+\infty} \left\| b_k \boldsymbol\gamma_k \right\|_2 = 0 < +\infty$ with probability one.
	
	% \item[5)] The set of stable points for the ODE \eqref{eq_ODE} should be obtained. \\
	% According to \eqref{eq:stable_points}, $\mathcal{S}(x)$ contains the local optimal stable points for the ODE \eqref{eq_ODE}. What's more, the boundary point 1 (or $-1$) is a stable point when $f(1, x) \ge 0$ (or $f(-1, x) \le 0$). Hence, the set of stable points is a subset of $\mathcal{S}(x) \cup \{ -1\} \cup \{1 \}$.
\end{itemize}

By Theorem 5.2.1 in \cite{kushner2003stochastic}, $\hat{\textbf{x}}_k$ converges to a unique stable point within the stable points set with probability one. % This completes the proof of Theorem~\ref{th_convergence}.

%\section{Proof of Lemma~\ref{UnifiedOptShift}}\label{proof_UnifiedOptShift}

\section{Proof of Theorem~\ref{Converge to real beam direction}}\label{proof_Converge to real beam direction}

Theorem \ref{proof_Converge to real beam direction} is proven in three steps:

\emph{\textbf{Step 1:} Two continuous processes based on the discrete process $ \hat{\boldsymbol{\psi}}_k = [\hat{\beta}^\text{re}_{k}, \hat{\beta}^\text{im}_{k}, \hat{x}_{k,1},\hat{x}_{k,2}]^\text{\emph{T}}$ are established here, i.e., $\bar{\boldsymbol{\psi}}(t) \!\overset{\Delta}{=}\! [\bar{\beta}^\text{re}(t), \bar{\beta}^\text{im}(t), \bar{x}_{1}(t),\bar{x}_{2}(t)]^\text{\emph{T}}$ and $\tilde{\boldsymbol{\psi}}^k(t) \!\overset{\Delta}{=}\! [\tilde{\beta}^{\text{re},k}(t), \tilde{\beta}^{\text{im},k}(t), \tilde{x}_{1}^k(t), \tilde{x}_{2}^k(t)]^\text{\emph{T}}$.}

The discrete time parameters are defined as: $t_{0} \overset{\Delta}{=} 0$, $t_k \overset{\Delta}{=} \sum_{i=1}^k b_{i}$, $k \ge 1$. The first continuous process $\bar{\boldsymbol{\psi}}(t), t \ge 0$ is constructed as the linear interpolation of the sequence $\hat{\boldsymbol{\psi}}_k, k\ge 0$, where $\bar{\boldsymbol{\psi}}(t_k) = \hat{\boldsymbol{\psi}}_k, k \ge 0$. Therefore, $\bar{\boldsymbol{\psi}}(t)$ is given by
\begin{equation}\label{eq_continuous}
\begin{aligned}
\bar{\boldsymbol{\psi}}(t)\!=\!\bar{\boldsymbol{\psi}}(t_k)\!+\!\frac{(t\!-\!t_k)}{b_{k+1}}\left[\bar{\boldsymbol{\psi}}(t_{k+1})\!-\!\bar{\boldsymbol{\psi}}(t_k)\right], t\!\in\![t_k, t_{k+1}].
\end{aligned}
\end{equation}

The second continuous process $\tilde{\boldsymbol{\psi}}^k(t)$ is the solution of the following ordinary differential equation (ODE):
\begin{align}\label{eq_ODE}
\frac{d \tilde{\boldsymbol{\psi}}^k(t)}{dt} = \mathbf{f}\left(\tilde{\boldsymbol{\psi}}^k(t), \boldsymbol{\psi}\right),
\end{align}
for $t \in [t_k, \infty)$, where $\tilde{\boldsymbol{\psi}}^k(t_k) = \bar{\boldsymbol{\psi}}(t_k) = \hat{\boldsymbol{\psi}}_k, k \ge 0$. Thus, $\tilde{\boldsymbol{\psi}}^k(t)$ can be given as
\begin{equation}\label{eq_ODE_new}
\begin{aligned}
\tilde{\boldsymbol{\psi}}^k(t) & = \bar{\boldsymbol{\psi}}(t_k) + \int_{t_k}^t \mathbf{f}\left(\tilde{\boldsymbol{\psi}}^k(v), \boldsymbol{\psi}\right) dv, t \ge t_k.
\end{aligned}
\end{equation}

\emph{\textbf{Step 2:} By using the two continuous processes $\bar{\boldsymbol{\psi}}(t)$ and $\tilde{\boldsymbol{\psi}}^k(t)$ constructed in Step 1, a sufficient condition for the convergence of the discrete process $\hat{\textbf{x}}_k$ is provided here.}

We first construct a time-invariant set $\mathcal{I}$ that includes the direction parameter vector $\textbf{x}$ within the mainlobe, i.e., $\textbf{x} \in \mathcal{I} \subset \mathcal{B}(\textbf{x})$\footnote{The boundary of the set $\mathcal{B}(\textbf{x})$ is denoted by $\partial \mathcal{B}(\textbf{x})$.}. Define $\tilde{\textbf{x}}^0(t) \triangleq \left[\tilde{x}_1^0(t),\tilde{x}_2^0(t)\right]^\text{T}$ and denote $\hat{\textbf{x}}_{\text{b}} = \tilde{\textbf{x}}^0(t_{\text{b}})$ as the beam direction of the process $\tilde{\boldsymbol{\psi}}^0(t)$ that is closest to the boundary of the mainlobe, which is given by
\begin{align}\label{eq:x_b}
\inf_{\textbf{v} \in \partial \mathcal{B}(x), t \ge 0} \left\| \textbf{v} - \tilde{\textbf{x}}^0(t) \right\|_2 = \inf_{\textbf{v} \in \partial \mathcal{B}(x)} \left\| \textbf{v} - \hat{\textbf{x}}_{\text{b}} \right\|_2  > 0.
\end{align}Then we pick $\delta$ such that
\begin{align}\label{eq:delta}
\min \left\{\inf_{\textbf{v} \in \partial \mathcal{B}(x)} \left\| \textbf{v} - \hat{\textbf{x}}_{\text{b}} \right\|_{-\infty}, \left\|\hat{\textbf{x}}_b - \textbf{x}\right\|_{-\infty} \right\} > \delta > 0,
\end{align}where $\left\|\textbf{u}\right\|_{-\infty} = \underset{l=1,2}{\min}{\left[\textbf{u}\right]_l}$ denotes the minimum element of $\textbf{u}$. Note that when $t \ge t_b$, the solution $\tilde{\boldsymbol{\psi}}^0(t)$ of the ODE (\ref{eq_ODE}) will approach the real channel coefficient $\beta$ and direction parameter vector $\textbf{x}$ monotonically as time $t$ increases.
Hence, we construct the invariant set $\mathcal{I}$ as \eqref{main_lobe_b}.
\begin{figure*}[!t]
	\begin{equation}\label{main_lobe_b}
	\mathcal{I} = \Big( x_1 - |x_1 - \hat{x}_{1,\text{b}}| - \delta,~x_1 + |x_1 - \hat{x}_{1,\text{b}}| + \delta \Big)
	\times \Big( x_2 - |x_2 - \hat{x}_{2,\text{b}}| - \delta,~x_2 + |x_2 - \hat{x}_{2,\text{b}}| + \delta \Big) \subset \mathcal{B}(\textbf{x}).
	\end{equation}
	\hrulefill
	\vspace*{4pt}
\end{figure*}An example of the invariant set $\mathcal{I}$ is shown in Fig. \ref{fig_invariant_set}.

\begin{figure}[!t]
	%\vspace{-5mm}
	\centering
	\includegraphics[width=4.5cm]{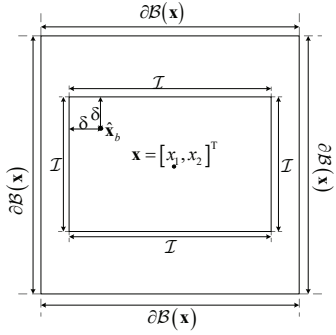}
	%\vspace{-3mm}
	\caption{An illustration of the invariant set $\mathcal{I}$.}
	\label{fig_invariant_set}
	%\vspace{-5mm}
\end{figure}

Then, a sufficient condition will be established in Lemma \ref{le_sufficient} that ensures $\hat{\textbf{x}}_k\!\in\!\mathcal{I}~\text{for}~k\!\ge\!0$, and hence from Corollary 2.5 in \cite{borkar2008stochastic}, we can obtain that $\hat{\textbf{x}}_k$ converges to $\textbf{x}$.
Before giving Lemma \ref{le_sufficient}, let us provide some definitions first:
\begin{itemize}
	\item Pick $T > 0$ such that the solution $\tilde{\boldsymbol{\psi}}^0(t), t\ge 0$ of the ODE (\ref{eq_ODE}) with $\tilde{\boldsymbol{\psi}}^0(0)\!=\![\hat{\beta}^\text{re}_{0}, \hat{\beta}^\text{im}_{0}, \hat{x}_{{0},1},\hat{x}_{{0},2}]^\text{T}$ satisfies $\inf_{\textbf{v} \in \partial \mathcal{B}}\left| \textbf{v}\!-\!\tilde{\textbf{x}}^0(t) \right| \ge 2\delta$ for $t \ge T$. Since when $t \ge t_b$, $\tilde{\textbf{x}}^0(t)$ will approach the direction parameter vector $\textbf{x}$ monotonically as time $t$ increases, one possible $T$ is given by
	\begin{align}\label{eq_T}
	T= \arg\min\limits_{t\in[t_\text{b}, \infty]}\left|~\!\!\left|\!\left[\int_{t_\text{b}}^t \mathbf{f}\left(\tilde{\boldsymbol{\psi}}^0(v), \boldsymbol{\psi}\right) dv\right]_3\right| - \delta\right|,
	\end{align}
	where $[\cdot]_{i}$ obtains the $i$-th element of the vector.
	
	\item Let $T_0 \overset{\Delta}{=} 0$ and $T_{l+1} \overset{\Delta}{=} \min \left\{ t_i: t_i \ge T_l + T, i \ge 0 \right\}$ for $l \ge 0$. Then $T_{l+1} - T_l \in [T, T+b_1]$ and $T_l = t_{\tilde{k}(l)}$ for some $\tilde{k}(l) \uparrow \infty$, where $\tilde{k}(0) = 0$. Let $\tilde{\boldsymbol{\psi}}^{\tilde{k}(l)}(t)$ denote the solution of ODE (\ref{eq_ODE}) for $t \in I_l \overset{\Delta}{=} \left[ T_l, T_{l+1} \right]$ with $\tilde{\boldsymbol{\psi}}^{\tilde{k}(l)}(T_l) = \bar{\boldsymbol{\psi}}(T_l)$, $l \ge 0$.
\end{itemize}
Hence, we can obtain the following lemma:
\begin{lemma}\label{le_sufficient}
	If $ \underset{t\in I_l}{\sup} \left\| \bar{\textbf{x}}(t) - \tilde{\textbf{x}}^{\tilde{k}(l)}(t)\right\|_2 \le \delta$ for all $l \ge 0$, then $\hat{\textbf{x}}_k \in \mathcal{I}~\text{for all}~k \ge 0$.
\end{lemma}
\begin{proof}
	If $\underset{t\in I_l}{\sup} \left\| \bar{\textbf{x}}(t) - \tilde{\textbf{x}}^{\tilde{k}(l)}(t)\right\|_2 \le \delta$ for all $l \ge 0$, then $\underset{t\in I_l}{\sup} \left| \bar{x}_1(t) - \tilde{{x}}_{1}^{\tilde{k}(l)}(t)\right|\le \delta$ and $\underset{t\in I_l}{\sup} \left| \bar{x}_2(t) - \tilde{{x}}_{2}^{\tilde{k}(l)}(t)\right| \le \delta$.
	
	According to Lemma 1 in \cite{JLiJoint2018}, $\hat{{x}}_{k,1} \in \mathcal{I}~\text{for all}~k \ge 0$ and $\hat{{x}}_{k,2} \in \mathcal{I}~\text{for all}~k \ge 0$. Hence, $\hat{\textbf{x}}_k \in \mathcal{I}~\text{for all}~k \ge 0$.
	% The proof is omitted due to space limitation.
\end{proof}

\emph{\textbf{Step 3:} We will derive the probability lower bound for the condition in Lemma \ref{le_sufficient}, which is also a lower bound for $P\left( \left. \hat{\textbf{x}}_k\!\rightarrow\!\textbf{x} \right| \hat{\textbf{x}}_0\!\in\!\mathcal{B}\left(\textbf{x}\right) \right)$. }

We will derive the probability lower bound for the condition in Lemma \ref{le_sufficient}, which results in the following lemma:
\begin{lemma}\label{le_lower_bound}
	If (i) the initial point satisfies $\hat{\textbf{x}}_0 \in \mathcal{B}(\textbf{x})$, (ii) $b_k$ is given by (\ref{eq_stepsize}) with any $\epsilon > 0$,
	then there exist  $K_0 \ge 0$ and $C>0$ such that
	\begin{equation}\label{eq_lock}
	\begin{aligned}
	P\left( \hat{\textbf{x}}_k \in \mathcal{I}, \forall k \ge 0 \right) \ge 1 - 8e^{-\frac{C|s_p|^2}{\epsilon^2\sigma^2}}.
	\end{aligned}
	\end{equation}
\end{lemma}
\begin{proof} See Appendix \ref{sec_proof_le_lower_bound}.
\end{proof}

Finally, applying Lemma \ref{le_lower_bound} and Corollary 2.5 in \cite{borkar2008stochastic}, we can obtain
\begin{align}\label{eq_lock10}
P\left( \left. \hat{\textbf{x}}_k \rightarrow \textbf{x} \right| \hat{\textbf{x}}_0 \in \mathcal{B} \right) \ge &~ P\left( \hat{\textbf{x}}_k \in \mathcal{I}, \forall k \ge 0 \right) \\
\ge &~ 1 - 8e^{-\frac{C|{s_p}|^2}{\epsilon^2\sigma^2}},\nonumber
\end{align}
which completes the proof of Theorem \ref{Converge to real beam direction}.

\section{Proof of Theorem~\ref{Converge to with minimum CRLB}}\label{proof_Converge to with minimum CRLB}
If the step-size $b_k$ is given by \eqref{eq_stepsize} with any $\varepsilon > 0$ and $K_0 \ge 0$, the sufficient conditions are provided by Theorem 6.6.1 \cite[Section 6.6]{nevel1973stochastic} to prove the asymptotic normality of $\sqrt{k} \left( \hat{\textbf{x}}_k - \textbf{x} \right)$, i.e., $\sqrt{k} \left( \hat{\textbf{x}}_k - \textbf{x} \right) \overset{d}{\rightarrow}\mathcal{N}\left( 0, \Sigma_\textbf{x} \right)$. With the the condition that $\hat{\boldsymbol{\psi}}_k \rightarrow \boldsymbol{\psi}$, we can prove that the beam and channel tracking algorithm satisfies the condition above and obtain the variance $\boldsymbol{\Sigma}$ as follows:
\begin{itemize}
	% \item[1)] The estimate $\hat{x}_n$ should be within $[-1, 1]$.\\
	% The projection operator in \eqref{eq_est} ensures that $\hat{x}_n \in [-1, 1]$.
	
	\item[1)] Equation \eqref{eq_rtracking} is supposed to satisfy: (i) there exists an increasing sequence of $\sigma$-fields $\{\mathcal{F}_{k}: k \ge 0\}$ such that $\mathcal{F}_{l} \!\subset\!\mathcal{F}_{k}$ for $l\!<\!k$, and (ii) the random noise $\hat{\mathbf{z}}_k$ is $\mathcal{F}_{k}$-measurable and independent of $\mathcal{F}_{k-1}$.\\
	As is shown in Appendix \ref{proof_Converge to unique stable point}, there exists an increasing sequence of $\sigma$-fields $\{ \mathcal{G}_k : k \ge 0 \}$, where $\hat{\mathbf{z}}_k$ is measurable with respect to $\mathcal{G}_{k}$, i.e., $\mathbb{E} \left[ \left. \hat{\mathbf{z}}_k \right| \mathcal{G}_{k} \right] = \hat{\mathbf{z}}_k$, and is independent of $\mathcal{G}_{k-1}$, i.e., $\mathbb{E} \left[ \left. \hat{\mathbf{z}}_k \right| \mathcal{G}_{k-1}\right]  = \mathbb{E} \left[ \hat{\mathbf{z}}_k \right] = \mathbf{0}$.
	
	\item[2)] $\hat{\textbf{x}}_k$ should converge to $\textbf{x}$ almost surely as $k \rightarrow +\infty$. \\
	We assume that $\hat{\boldsymbol{\psi}}_k \rightarrow \boldsymbol{\psi}$, hence $\hat{\textbf{x}}_k$ converges to $\textbf{x}$ almost surely when $k \rightarrow +\infty$.
	
	\item[3)] The stable condition:\\
	In \eqref{eq_f}, we rewrite $\mathbf{f}\left(\hat{\boldsymbol{\psi}}_{k-1}, \boldsymbol{\psi}\right)$ as follows:
	\begin{equation}
	\begin{aligned}
	\!\!\!\!\!\!\mathbf{f}\left(\hat{\boldsymbol{\psi}}_{k\!-\!1}, \boldsymbol{\psi}\right)\!=\!\mathbf{C}_1 \left( \hat{\boldsymbol{\psi}}_{k\!-\!1} \!-\! \boldsymbol{\psi} \right) \!\!+\!\! \left[\!\begin{matrix} o(\| \hat{\boldsymbol{\psi}}_{k-1} \!-\! \boldsymbol{\psi} \|_2) \\ o(\| \hat{\boldsymbol{\psi}}_{k-1} \!-\! \boldsymbol{\psi} \|_2) \\ o(\| \hat{\boldsymbol{\psi}}_{k-1} \!-\! \boldsymbol{\psi} \|_2) \\ o(\| \hat{\boldsymbol{\psi}}_{k-1} \!-\! \boldsymbol{\psi} \|_2) \end{matrix}\!\right]\!\!,\!\!\!\!
	%& = -\frac{M(M-1)^2|\beta x|^2}{2\sigma^2} \left( \hat{x}_n - u^* \right) + o\left(\hat{x}_n - u^*\right).
	\end{aligned}
	\end{equation}where $\mathbf{C}_1$ is given by
	\begin{equation}
	\begin{aligned}
	\!\!\!\mathbf{C}_1 \!=\!  \left.\frac{\partial \mathbf{f}\left(\hat{\boldsymbol{\psi}}_{k-1}, \boldsymbol{\psi}\right)}{\partial \hat{\boldsymbol{\psi}}_{k-1}^\text{T}} \right|_{\hat{\boldsymbol{\psi}}_{k-1} = \boldsymbol{\psi}} \!=\! \!-\! \left[\begin{matrix} 1 & 0 & 0 & 0 \\ 0 & 1 & 0  & 0 \\ 0 & 0 & 1 & 0 \\ 0 &0 &0 &1 \end{matrix} \right].\!\!
	\end{aligned}
	\end{equation}
	Then the stable condition is obtained that:
	\begin{align}
	\mathbf{E} \!=\! \mathbf{C}_1  \cdot \varepsilon + \frac{1}{2} \!=\! - \!\left[\begin{matrix} \varepsilon \!-\! \frac{1}{2} & 0 & 0 & 0 \\ 0 & \varepsilon \!-\! \frac{1}{2} & 0 & 0 \\ 0 & 0 & \varepsilon \!-\! \frac{1}{2} & 0\\ 0 & 0 & 0 & \varepsilon \!-\! \frac{1}{2}& \end{matrix} \right]  \prec 0,
	\end{align}
	which leads to $\varepsilon > \frac{1}{2}$.
	
	\item[4)] The noise vector $\hat{\mathbf{z}}_k$ satisfies:\\
	\begin{equation}
	\mathbb{E}\left[\left\|\hat{\mathbf{z}}_k\right\|_2^2\right] = \operatorname{tr}\left\{\mathbf{I}(\hat{\boldsymbol{\psi}}_{k\!-\!1},\!\mathbf{W}_k)^{-1}\right\} < +\infty,
	\end{equation}
	and
	\begin{equation}
	\underset{v\rightarrow\infty}{\lim}\ \ \underset{k\ge1}{\sup}\ \ \int\limits_{\left\| \hat{z}_k \right\|_2 > v} \left\| \hat{\mathbf{z}}_k \right\|_2^2 p(\hat{\mathbf{z}}_k) d\hat{\mathbf{z}}_k = 0.
	\end{equation}
	
\end{itemize}

Let
\begin{align}
\mathbf{F} = &~ \lim_{\begin{matrix} k \rightarrow +\infty \\ \hat{\boldsymbol{\psi}}_k \rightarrow \boldsymbol{\psi} \end{matrix}} \mathbb{E}\left[\hat{\mathbf{z}}_k \hat{\mathbf{z}}_k^\text{T}\right] \\
\overset{(a)}{=} &~ \lim_{\begin{matrix} k \rightarrow +\infty \\ \hat{\boldsymbol{\psi}}_k \rightarrow \boldsymbol{\psi} \end{matrix}}\mathbf{I}(\hat{\boldsymbol{\psi}}_{k},\!\mathbf{W}_{k\!+\!1})^{-1} = \mathbf{I}(\boldsymbol{\psi}, \mathbf{W}^*)^{-1}, \nonumber
\end{align}
where step $(a)$ is obtained from \eqref{eq_rzc}.

By Theorem 6.6.1 \cite[Section 6.6]{nevel1973stochastic}, we have
\begin{align*}\sqrt{k+K_0}\left( \hat{\boldsymbol{\psi}}_{k} - \boldsymbol{\psi} \right) \overset{d}{\rightarrow}\mathcal{N}\left( 0, \boldsymbol\Sigma \right), \end{align*}
where
\begin{equation}\label{eq_Sigma}\begin{aligned}
\boldsymbol\Sigma = &~\alpha^2  \cdot \int_0^\infty e^{\mathbf{E}v} \mathbf{F} e^{\mathbf{E}^\text{H}v} dv \\
= &~\frac{\varepsilon^2}{2\varepsilon - 1}\mathbf{I}(\boldsymbol{\psi}, \mathbf{W}^*)^{-1}.\\ %& = \left( \frac{\lambda^2 }{2M(M-1)^2\pi^2 d^2\rho} \right)^2 *  \frac{2M(M-1)^2\pi^2 d^2\rho}{\lambda^2} \\
%& \ge I_{\max}^{-1} = \frac{\lambda^2 }{2M(M-1)^2\pi^2 d^2\rho}.
\end{aligned}\end{equation}
Due to that $\lim_{k\rightarrow\infty}\sqrt{{(k+K_0)}/{k}} = 1$, we have
\begin{align*}
\sqrt{k}\left( \hat{\boldsymbol{\psi}}_{k} - \boldsymbol{\psi} \right) \rightarrow \sqrt{k}\cdot\sqrt{\frac{k+K_0}{k}}\left( \hat{\boldsymbol{\psi}}_{k} - \boldsymbol{\psi} \right) \overset{d}{\rightarrow}\mathcal{N}\left( 0, \boldsymbol\Sigma \right),
\end{align*}
if $k\rightarrow +\infty$. Thus, we can get
\begin{align}
\sqrt{k}\left( \hat{\boldsymbol{\psi}}_{k} - \boldsymbol{\psi} \right) \overset{d}{\rightarrow}\mathcal{N}\left( 0, \boldsymbol\Sigma \right).
\end{align}
By adapting $\epsilon=1$ in \eqref{eq_Sigma}, we can obtain
\vspace{-1mm}
\begin{equation}\label{eq_asypsi}
\begin{aligned}
\sqrt{k}\left( \hat{\boldsymbol{\psi}}_{k} - \boldsymbol{\psi} \right) \overset{d}{\rightarrow}\mathcal{N}\left( 0, \mathbf{I}(\boldsymbol{\psi}, \mathbf{W}^*)^{-1} \right).
\end{aligned}
\end{equation}
\vspace{-1mm}Since $\hat{\boldsymbol{\psi}}_{k} \to \boldsymbol{\psi}$ as $k \to +\infty$, $\hat{\textbf{h}}_k - \textbf{h}_k$ is linear to $ \hat{\boldsymbol{\psi}}_{k} - \boldsymbol{\psi}$. Hence, $\hat{\textbf{h}}_k - \textbf{h}_k$ is also asymptotically Gaussian.

Combining \eqref{eq_fLB}, \eqref{eq_asypsi} and \eqref{eq_CMMSE}, we can conclude that
\begin{align}
\mathop {\lim }\limits_{k \to +\infty } \frac{k}{MN} \mathbb{E} \left[{\left\| \hat{\textbf{h}}_k - \textbf{h}_k \right\|}_2^2 \big| \hat{\boldsymbol{\psi}}_k \to \boldsymbol{\psi} \right] = {I}_{\min}(\boldsymbol{\psi}).
\end{align}

\section{Proof of Lemma \ref{le_lower_bound}}\label{sec_proof_le_lower_bound}
The following lemmas are introduced to prove Lemma \ref{le_lower_bound}.

\begin{lemma}[Lemma 3 \cite{JLiJoint2018}]\label{le_gronwall}
	Given $T$ by \eqref{eq_T} and
	\begin{align}\label{eq_nT}
	k_T \overset{\Delta}{=} \inf \left\{i \in \mathbb{Z}: t_{k+i} \ge t_k + T \right\}.
	\end{align}
	If there exists a constant $C>0$, which satisfies
	\begin{equation}\label{eq_gronwall1}
	\begin{aligned}
	&~\left\| \bar{\boldsymbol{\psi}}(t_{k+l}) - \tilde{\boldsymbol{\psi}}^k(t_{k+l})\right\|_2 \\
	\le &~L \sum_{i=1}^{l} a_{k+i} \left\| \bar{\boldsymbol{\psi}}(t_{k+i-1}) - \tilde{\boldsymbol{\psi}}^k(t_{k+i-1}) \right\|_2  + C,
	\end{aligned}
	\end{equation}
	for all $k \ge 0$ and $1 \le l \le k_T$, then
	\begin{equation}\label{eq_gronwall2}
	\begin{aligned}
	\underset{t\in\left[ t_k, t_{k+k_T} \right]}{\sup} \left\| \bar{\boldsymbol{\psi}}(t) - \tilde{\boldsymbol{\psi}}^k(t)\right\|_2 \le  \frac{C_{\mathbf{f}} b_{k+1}}{2} + C e^{L (T+b_1)},
	\end{aligned}
	\end{equation}
	where $L$ and $C_{\mathbf{f}}$ are defined in \eqref{eq_Lip} and \eqref{eq_CT} separately.
\end{lemma}

%\begin{lemma}\label{le_cheb}
%If a $D$-dimensional process
%$$\{\mathbf{M}_n = [M_{n,1}, \ldots, M_{n,D}]^\text{T}: n = 1, 2, \ldots\},$$
%statisfies that: (i) $\mathbf{M}_n$ is Gaussian distributed with zero mean, i.e., $\mathbf{M}_n \sim \mathcal{N}(\mathbf{0}, \mathbf{V}_n)$, and (ii)~$\sum_{i = 1}^D \left(M_{n,i}^2 - \mathbb{E}\left[M_{n,i}^2\right]\right)$ is a martingale in $n$, then
%\begin{equation}\label{eq_lock5}
%\begin{aligned}
%	&~P\left( \underset{0\le n \le k}{\sup}\sum_{i = 1}^D \left(M_{n,i}^2 - \mathbb{E}\left[M_{n,i}^2\right]\right) > \eta \right) \\
%	\le &~\frac{e^{-C\left(\sum_{i=1}^D \lambda_{k,i} + \eta\right)} }{\sqrt{\prod_{i = 1}^D \left(1 - 2C\lambda_{k,i}\right)}},
%\end{aligned}
%\end{equation}
%for any $\eta > 0$, where $\lambda_{k,1} \le \lambda_{k,2} \le \cdots \le \lambda_{k,D}$ are the eigenvalues of $\mathbf{V}_k$, and $0 < C < \frac{1}{2\lambda_{k,D}}$.
%\end{lemma}
\begin{lemma}[Lemma 4 \cite{JLiSuper2017}]\label{le_cheb}
	If $\{M_i: i = 1, 2, \ldots\}$ satisfies that: (i)  $M_i$ is Gaussian distributed with zero mean, and (ii) $M_i$ is a martingale in $i$, then
	\begin{equation}\label{eq_lock5}
	\begin{aligned}
	& P\left( \underset{0\le i \le k}{\sup}\left|M_i\right| > \eta \right) \le 2\exp\left\{-\frac{\eta^2}{2\operatorname{Var}\left[M_k\right]}\right\},
	\end{aligned}
	\end{equation}
	for any $\eta > 0$.
\end{lemma}
%\begin{proof}
%See Appendix \ref{sec_proof_le_cheb}.
%% The proof is omitted due to space limitation.
%\end{proof}

\begin{lemma}[Lemma 5 \cite{JLiSuper2017}]\label{le_sum}
	If given a constant $C > 0$, then
	\begin{equation}\label{eq_increasing}
	\begin{aligned}
	G(v) = \frac{1}{v}\exp\left[-\frac{C}{v}\right],
	\end{aligned}
	\end{equation}
	is increasing for all $0 < v < C$.
\end{lemma}
Let $\boldsymbol{\xi}_0 \overset{\Delta}{=} \mathbf{0}$ and $\boldsymbol{\xi}_k \overset{\Delta}{=} \sum_{l=1}^{k} b_{l} \mathbf{\hat{z}}_{l} $, $k \ge 1$, where $\mathbf{\hat{z}}_{l}$ is given in \eqref{eq_z}. With \eqref{eq_continuous} and \eqref{eq_ODE_new}, we have for $t_{k+l}, 1 \le l \le k_T$,
\begin{align}\label{eq_seq_trace} \bar{\boldsymbol{\psi}}(t_{k+l}) = &~\bar{\boldsymbol{\psi}}(t_k)  + \sum_{i=1}^{l} a_{k+i} \mathbf{f}\left(\bar{\boldsymbol{\psi}}(t_{n+i-1}), \boldsymbol{\psi}\right) \\
&~+ (\boldsymbol{\xi}_{k+l} - \boldsymbol{\xi}_{k}), \nonumber
\end{align}
and
\begin{align}\label{eq_ode_trace}
\tilde{\boldsymbol{\psi}}^n(t_{k+l}) = &~\tilde{\boldsymbol{\psi}}^k(t_k) + \int_{t_k}^{t_{k+l}} \mathbf{f}\left(\tilde{\boldsymbol{\psi}}^k(v), \boldsymbol{\psi}\right) dv \\
= &~\tilde{\boldsymbol{\psi}}^k(t_k) + \sum_{i=1}^{l} b_{k+i} \mathbf{f}\left(\tilde{\boldsymbol{\psi}}^k(t_{k+i-1}), \boldsymbol{\psi}\right) \nonumber \\
&~+ \int_{t_k}^{t_{k+l}} \left[\mathbf{f}\left(\tilde{\boldsymbol{\psi}}^k(v), \boldsymbol{\psi}\right) - \mathbf{f}\left(\tilde{\boldsymbol{\psi}}^k(\underline{v}), \boldsymbol{\psi}\right) \right]dv,  \nonumber
\end{align}
where $\underline{v} \overset{\Delta}{=} \max \left\{ t_k: t_k \le v, k \ge 0 \right\}$ for $v \ge 0$.
%Note that we only care about $\hat{x}_n \in \mathcal{I} \subset \mathcal{B}$, so the projection operator does not take effect in \eqref{eq_seq_trace} and we omit it.

To bound $\int_{t_k}^{t_{k+l}} \left[\mathbf{f}\left(\tilde{\boldsymbol{\psi}}^k(v), \boldsymbol{\psi}\right) - \mathbf{f}\left(\tilde{\boldsymbol{\psi}}^k(\underline{v}), \boldsymbol{\psi}\right) \right]dv$ on the RHS of (\ref{eq_ode_trace}), we obtain the Lipschitz constant of function $\mathbf{f}(\mathbf{v}, \boldsymbol{\psi})$ considering the first varible $\mathbf{v}$, given by
\begin{equation}\label{eq_Lip}
L \overset{\Delta}{=} \underset{\mathbf{v}_1 \ne \mathbf{v}_2}{\sup} \frac{\left\| \mathbf{f}(\mathbf{v}_1, \boldsymbol{\psi}) -  \mathbf{f}(\mathbf{v}_2, \boldsymbol{\psi}) \right\|_2}{\left\| \mathbf{v}_1 - \mathbf{v}_2 \right\|_2}.
\end{equation}
Similar to \eqref{eq_ub_fx}, for any $t \ge t_k$, we can obtain that there exists a constant $0 < C_{\mathbf{f}} < \infty$ such that
\begin{equation}\label{eq_CT}
\begin{aligned}
\left\| \mathbf{f}\left(\tilde{\boldsymbol{\psi}}^k(t), \boldsymbol{\psi}\right) \right\|_2 \le C_{\mathbf{f}}.
\end{aligned}
\end{equation}
Hence, we have
\begin{equation}\label{eq_int}
\begin{aligned}
& \left\| \int_{t_k}^{t_{k+m}} \left[\mathbf{f}\left(\tilde{\boldsymbol{\psi}}^k(v), \boldsymbol{\psi}\right) - \mathbf{f}\left(\tilde{\boldsymbol{\psi}}^k(\underline{v}), \boldsymbol{\psi}\right) \right]dv \right\|_2 \\
\le & \int_{t_k}^{t_{k+l}} \left\| \mathbf{f}\left(\tilde{\boldsymbol{\psi}}^k(v), \boldsymbol{\psi}\right) - \mathbf{f}\left(\tilde{\boldsymbol{\psi}}^k(\underline{v}), \boldsymbol{\psi}\right) \right\|_2 dv \\
\overset{(a)}{\le} & \int_{t_k}^{t_{k+l}} L \left\| \tilde{\boldsymbol{\psi}}^k(v) - \tilde{\boldsymbol{\psi}}^k(\underline{v}) \right\|_2 dv \\
\overset{(b)}{\le} & \int_{t_k}^{t_{k+l}} L \left\| \int_{\underline{v}}^{v} \mathbf{f}\left(\tilde{\boldsymbol{\psi}}^k(s), \boldsymbol{\psi}\right) ds \right\|_2 dv \\
\le & \int_{t_k}^{t_{k+l}} \int_{\underline{v}}^{v} L \left\| \mathbf{f}\left(\tilde{\boldsymbol{\psi}}^k(s), \boldsymbol{\psi}\right) \right\|_2 ds dv \\
\overset{(c)}{\le} & \int_{t_k}^{t_{k+l}} \int_{\underline{v}}^{v} C_{\mathbf{f}} L ds dv =   \int_{t_k}^{t_{k+l}}C_{\mathbf{f}} L (v - \underline{v}) dv \\
=  & \sum_{i=1}^{l} \int_{t_{k+i-1}}^{t_{k+i}}C_{\mathbf{f}} L (v - t_{k+i-1}) dv \\
= & \sum_{i=1}^{l} \frac{C_{\mathbf{f}} L (t_{k+i} - t_{k+i-1})^2}{2}= \frac{C_{\mathbf{f}} L}{2} \sum_{i=1}^{l} b_{k+i}^2,\end{aligned}
\end{equation}
where Step $(a)$ is due to (\ref{eq_Lip}), Step $(b)$ is due to the definition in (\ref{eq_ODE_new}), and Step $(c)$ is due to (\ref{eq_CT}). Then, by subtracting $\tilde{\boldsymbol{\psi}}^k(t_{k+l})$ in \eqref{eq_ode_trace} from $\bar{\boldsymbol{\psi}}(t_{k+l})$  in \eqref{eq_seq_trace} and taking norms, the following inequality can be obtained from (\ref{eq_Lip}) and (\ref{eq_int}) for $k \ge 0, 1 \le l \le k_T$:
\begin{equation}\label{eq_lock2}
\begin{aligned}
& \left\| \bar{\boldsymbol{\psi}}(t_{k+l}) - \tilde{\boldsymbol{\psi}}^k(t_{k+l})\right\|_2 \\
\le & L \sum_{i=1}^{l} b_{k+i} \left\| \bar{\boldsymbol{\psi}}(t_{k+i-1}) - \tilde{\boldsymbol{\psi}}^k(t_{k+i-1}) \right\|_2 \\
&  + \frac{C_{\mathbf{f}} L}{2} \sum_{i=1}^{l} b_{k+i}^2+ \left\|\boldsymbol{\xi}_{k+l} - \boldsymbol{\xi}_{k}\right\|_2 \\
\le & L \sum_{i=1}^{l} b_{k+i} \left\| \bar{\boldsymbol{\psi}}(t_{k+i-1}) - \tilde{\boldsymbol{\psi}}^k(t_{k+i-1}) \right\|_2 \\
&  + \frac{C_{\mathbf{f}} L}{2} \sum_{i=1}^{k_T} b_{k+i}^2+ \underset{1 \le l\le k_T}{\sup}\left\|\boldsymbol{\xi}_{k+l} - \boldsymbol{\xi}_{k}\right\|_2.
%	& \left| \bar{x}(t_{n+m}) - \tilde{x}^n(t_{n+m})\right| \\
%	\le & L \sum_{i=1}^{m} a_{n+i} \left| \bar{x}(t_{n+i-1}) - \tilde{x}^n(t_{n+i-1}) \right| \\
%	&  + \frac{\sqrt{M} L}{2} \sum_{i=1}^{m} a_{n+i}^2+ |\xi_{n+m} - \xi_{n}|.
\end{aligned}
\end{equation}

Applying Lemma \ref{le_gronwall} to (\ref{eq_lock2}) and letting
\begin{align*}C = \frac{C_{\mathbf{f}} L}{2} \sum_{i=1}^{k_T} b_{k+i}^2+ \underset{1\le l\le k_T}{\sup}\left\|\boldsymbol{\xi}_{k+l} - \boldsymbol{\xi}_{k}\right\|_2,\end{align*}
yields
\begin{equation}\label{eq_lock3}
\begin{aligned}
&~\underset{t\in\left[ t_k, t_{k+k_T} \right]}{\sup} \left\| \bar{\boldsymbol{\psi}}(t) - \tilde{\boldsymbol{\psi}}^k(t)\right\|_2	 \\
\le &~C_e \left\{ \frac{C_{\mathbf{f}} L}{2} \big[c(k) - c(k+k_T)\big] \right.  \\
& \left. + \underset{1 \le l\le k_T}{\sup}\left\|\boldsymbol{\xi}_{k+l} - \boldsymbol{\xi}_{k}\right\|_2 \right\} + \frac{C_{\mathbf{f}} c_{k+1}}{2} ,
\end{aligned}
\end{equation}
where $C_e \overset{\Delta}{=} e^{L (T+b_1)}$, and $c(k) \overset{\Delta}{=} \sum_{i > k} b_{i}^2$.
% Suppose the initial point $\hat{x}_0 \in \mathcal{B}$.
% Then, we will derive the lower bound of the probability that the sequence $\{\hat{x}_n: n \ge 0\}$ remains inside this invariant set.
% Let $T_0 \overset{\Delta}{=} 0$ and $T_{m+1} \overset{\Delta}{=} \min \left\{ t_i: t_i \ge T_n + T, i \ge 0 \right\}$ for $m \ge 0$. Then $T_{m+1} - T_m \in [T, T+a_1]$ and $T_m = t_{\tilde{n}(m)}$ for some $\tilde{n}(m) \uparrow \infty$, where $\tilde{n}(0) = n_0$. Let $\tilde{x}^{\tilde{n}(m)}(t)$ denote the solution of ODE (\ref{eq_ODE}) for $t \in I_m \overset{\Delta}{=} \left[ T_m, T_{m+1} \right]$ with $\tilde{x}^{\tilde{n}(m)}(T_m) = \bar{x}(T_m)$, $m \ge 0$.
Letting $k = \tilde{k}(l)$ in (\ref{eq_lock3}), we have $k + k_T = \tilde{k}(l+1)$ due to the definition of $T_{l+1} = t_{\tilde{k}(l+1)}$ in \emph{Step 2} of Appendix \ref{proof_Converge to real beam direction} and
\begin{equation}\label{eq_lock3-2}
\begin{aligned}
&~\underset{t\in I_l}{\sup} \left\| \bar{\boldsymbol{\psi}}(t) - \tilde{\boldsymbol{\psi}}^{\tilde{k}(l)}(t)\right\|_2	 \\
\le &~C_e \left\{ \frac{C_{\mathbf{f}} L}{2} \big[ c(\tilde{k}(l)) - c(\tilde{k}(l+1)) \big]\right. \\
& \left.  + \underset{\tilde{k}(l) \le p \le \tilde{k}(l+1)}{\sup}\left\|\boldsymbol{\xi}_{p} - \boldsymbol{\xi}_{\tilde{k}(l)}\right\|_2 \right\} + \frac{C_{\mathbf{f}} b_{\tilde{k}(l)+1}}{2}.
\end{aligned}
\end{equation}

Suppose that the step size $\{b_k: k > 0\}$ satisfies
\begin{equation}\label{eq_lock_constr}
C_e \frac{C_{\mathbf{f}} L}{2} \big[c(\tilde{k}(l)) - c(\tilde{k}(l+1))\big] +  \frac{C_{\mathbf{f}} b_{\tilde{k}(l)+1}}{2} < \frac{\delta}{2},
\end{equation}
for $l \ge 0$.

Given $\underset{t\in I_l}{\sup} \left\| \bar{\textbf{x}}(t) - \tilde{\textbf{x}}^{\tilde{k}(l)}(t)\right\| \!>\! \delta$, we can obtain from \eqref{eq_lock3-2} and \eqref{eq_lock_constr} that
\begin{equation*}
\begin{aligned}
&~\underset{\tilde{k}(l)\le p \le \tilde{k}(l+1)}{\sup}\left\|\boldsymbol{\xi}_{p} - \boldsymbol{\xi}_{\tilde{k}(l)}\right\|_2 \\
\ge &~\frac{1}{C_e} \left( \underset{t\in I_l}{\sup} \left\| \bar{\boldsymbol{\psi}}(t) - \tilde{\boldsymbol{\psi}}^{\tilde{k}(l)}(t)\right\|_2 - \frac{C_{\mathbf{f}} L}{2} \big[ c(\tilde{k}(l)) \right. \\
& \left.  - c(\tilde{k}(l+1)) \big] - \frac{C_{\mathbf{f}} a_{\tilde{k}(l)+1}}{2}\right) \\
> &~\frac{1}{C_e}\left( \underset{t\in I_l}{\sup} \left| \bar{\textbf{x}}(t) - \tilde{\textbf{x}}^{\tilde{k}(l)}(t)\right| - \frac{\delta}{2} \right) \\
> &~\frac{\delta}{2C_e}.
\end{aligned}
\end{equation*}
Then, we get
\begin{equation}\label{eq_lock4}
\begin{aligned}
&~P\left( \left. \underset{t\in I_m}{\sup} \left\| \bar{\textbf{x}}(t) - \tilde{\textbf{x}}^{\tilde{k}(l)}(t)\right\| > \delta \right|\right. \\
&~~~~~~\left.\underset{t\in I_i}{\sup} \left\| \bar{\textbf{x}}(t) - \tilde{\textbf{x}}^{\tilde{k}(i)}(t)\right\| \le \delta, 0 \le i < l \right) \\
{\le} & P\left( \left. \underset{\tilde{k}(l)\le p \le \tilde{k}(l+1)}{\sup}\left\|\boldsymbol{\xi}_{p} - \boldsymbol{\xi}_{\tilde{k}(l)}\right\|_2 > \frac{\delta}{2C_e} \right| \right. \\
&~~~~~~\left.  \underset{t\in I_i}{\sup} \left\| \bar{\textbf{x}}(t) - \tilde{\textbf{x}}^{\tilde{k}(i)}(t)\right\| \le \delta, 0 \le i < l \right) \\
\overset{(a)}{=} &~P\left( \underset{\tilde{k}(l)\le p \le \tilde{k}(l+1)}{\sup}\left\|\boldsymbol{\xi}_{p} - \boldsymbol{\xi}_{\tilde{k}(l)}\right\|_2 > \frac{\delta}{2C_e} \right),
\end{aligned}
\end{equation}
where Step $(a)$ is due to the independence of noise, i.e., $ \boldsymbol{\xi}_{p}-\boldsymbol{\xi}_{\tilde{k}(l)} , \tilde{k}(l) \le p \le \tilde{k}(l+1)$ are independent of $\hat{\textbf{x}}_k, 0 \le k \le \tilde{k}(l)$.

%\begin{lemma}[Lemma 4.2 in \cite{borkar2008stochastic}]\label{le_42}
%
%\end{lemma}

The lower bound of the probability that the sequence $\{\hat{\textbf{x}}_k: k \ge 0\}$ remains in the invariant set $\mathcal{I}$ is given by
\begin{align}\label{eq_p_invariant}
& P\left( \hat{\textbf{x}}_k \in \mathcal{I}, \forall k \ge 0 \right) \nonumber\\
\overset{(a)}{\ge} & P\left( \underset{t\in I_m}{\sup} \left\| \bar{\textbf{x}}(t) - \tilde{\textbf{x}}^{\tilde{k}(l)}(t)\right\| \le \delta, \forall l \ge 0 \right)  \nonumber\\
\overset{(b)}{\ge} & 1 - \sum_{l\ge 0} P\left( \left. \underset{t\in I_m}{\sup} \left\| \bar{\textbf{x}}(t) - \tilde{\textbf{x}}^{\tilde{k}(l)}(t)\right| > \delta \right\| \right.  \\
&~~~~~~~~~~~~~~~\left. \underset{t\in I_i}{\sup} \left\| \bar{\textbf{x}}(t) - \tilde{\textbf{x}}^{\tilde{k}(i)}(t)\right\| \le \delta, 0 \le i < l \right)  \nonumber\\
\overset{(c)}{\ge} & 1 - \sum_{l\ge 0} P\Bigg( \underset{\tilde{k}(l)\le p \le \tilde{k}(l+1)}{\sup}\left\|\boldsymbol{\xi}_{p} - \boldsymbol{\xi}_{\tilde{k}(l)}\right\|_2 > \frac{\delta}{2C_e} \Bigg),\nonumber
\end{align}
where Step $(a)$ is due to Lemma \ref{le_sufficient}, Step $(b)$ is due to Lemma 4.2 in \cite{borkar2008stochastic}, and Step $(c)$ is due to \eqref{eq_lock4}. Let $\left\|\cdot\right\|_{\infty}$ denote the max-norm, i.e., $\left\|\mathbf{u}\right\|_{\infty} = \max_{l} |[\mathbf{u}]_l|$. Note that for $\mathbf{u} \in \mathbb{R}^{D}$, $\left\|\mathbf{u}\right\|_{2} \le \sqrt{D} \left\|\mathbf{u}\right\|_{\infty}$. Hence we have
\begin{align}\label{eq_lock6-0}
&~ P\left( \underset{\tilde{k}(l)\le p \le \tilde{k}(l+1)}{\sup}\left\|\boldsymbol{\xi}_{p} - \boldsymbol{\xi}_{\tilde{k}(l)}\right\|_2 > \frac{\delta}{2C_e} \right) \nonumber\\
\le &~ P\left( \underset{\tilde{k}(l)\le p \le \tilde{k}(l+1)}{\sup} \left\|\boldsymbol{\xi}_{p} - \boldsymbol{\xi}_{\tilde{k}(l)}\right\|_{\infty} > \frac{\delta}{4C_e} \right)  \\
= &~ P\left( \underset{\tilde{k}(l)\le p \le \tilde{k}(l+1)}{\sup} \max_{1 \le s \le 4} \left|\big[\boldsymbol{\xi}_{p}\big]_s - \big[\boldsymbol{\xi}_{\tilde{k}(l)}\big]_s \right| > \frac{\delta}{4C_e} \right) \nonumber\\
= &~ P\left( \max_{1 \le s \le 4} \underset{\tilde{k}(l)\le p \le \tilde{k}(l+1)}{\sup} \left|\big[\boldsymbol{\xi}_{p}\big]_s - \big[\boldsymbol{\xi}_{\tilde{k}(l)}\big]_s \right| > \frac{\delta}{4C_e} \right) \nonumber\\
\le &~ \sum_{s = 1}^4 P\left( \underset{\tilde{k}(l)\le p \le \tilde{k}(l+1)}{\sup} \left|\big[\boldsymbol{\xi}_{p}\big]_s - \big[\boldsymbol{\xi}_{\tilde{k}(l)}\big]_s \right| > \frac{\delta}{4C_e} \right). \nonumber
\end{align}

With the increasing $\sigma$-fields $\{\!\mathcal{G}_k\!:\!k\!\ge\!0\!\}$ defined in Appendix \ref{proof_Converge to unique stable point}, we have for $k \ge 0$,
\begin{itemize}
	\item[1)] $\boldsymbol{\xi}_k \!=\! \sum_{l=1}^{k} b_{l} \hat{\mathbf{z}}_{l} \sim \mathcal{N}(0, \sum_{l=1}^k b_k^2 \mathbf{I}(\hat{\boldsymbol{\psi}}_{l\!-\!1},\!\mathbf{W}_l)^{-1})$,
	
	\item[2)] $\boldsymbol{\xi}_k$ is $\mathcal{G}_k$-measurable, i.e., $\mathbb{E} \left[ \left. \boldsymbol{\xi}_k \right| \mathcal{G}_k \right] = \boldsymbol{\xi}_k$,
	
	\item[3)] $\mathbb{E} \left[ \left\| \boldsymbol{\xi}_k \right\|^2_2 \right] = \sum_{l=1}^k b_k^2 \operatorname{tr}\left\{\mathbf{I}(\hat{\boldsymbol{\psi}}_{l\!-\!1},\!\mathbf{W}_l)^{-1}\right\}< +\infty$,
	
	\item[4)] $\mathbb{E} \left[ \left. \boldsymbol{\xi}_k \right| \mathcal{G}_l \right] = \boldsymbol{\xi}_l$ for all $0 \le l < k$.
\end{itemize}
Therefore, $\left[\boldsymbol{\xi}_k\right]_s, s = 1,2,3,4$ is a Gaussian martingale with respect to $\mathcal{G}_k$, and satisfies
\begin{align}\label{eq_lock6-1}
\operatorname{Var}\left[\big[\boldsymbol{\xi}_{k+l}\big]_s - \big[\boldsymbol{\xi}_{k}\big]_s\right] = &~\sum_{i = k+1}^{k+l} b_i^2 \left[\mathbf{I}(\hat{\boldsymbol{\psi}}_{i\!-\!1},\!\mathbf{W}_i)^{-1}\right]_{s,s} \nonumber\\
\le &~\sum_{i = k+1}^{k+l} b_i^2 \frac{C_{\mathbf{I}}\sigma^2}{|{s_p}|^2} \\
= &~\frac{C_{\mathbf{I}}\sigma^2}{|{s_p}|^2} \big[c(k) - c(k+l)\big],\nonumber
\end{align}
where $C_{\mathbf{I}} \!\overset{\Delta}{=}\! \max_{s} \max_{i \ge 1} \frac{|{s_p}|^2}{\sigma^2}\big[\mathbf{I}(\hat{\boldsymbol{\psi}}_{i\!-\!1},\!\mathbf{W}_i)^{-1}\big]_{s,s}$. Let $\eta \!=\! \frac{\delta}{4C_e}$, $M_i \!=\! \big[\boldsymbol{\xi}_{\tilde{k}(l)+i}\big]_s - \big[\boldsymbol{\xi}_{\tilde{k}(l)}\big]_s, s\!=\! 1, 2, 3 ,4$ and $p = {\tilde{k}(l+1) - \tilde{k}(l)}$ in Lemma \ref{le_cheb}, then from \eqref{eq_lock6-0} and \eqref{eq_lock6-1}, we can obtain
\begin{align}\label{eq_lock6}
%P\left( \left. \underset{t\in I_m}{\sup} \left| \bar{x}(t) - x^n(t)\right| > \delta \right|   \underset{t\in I_m}{\sup} \left| \bar{x}(t) - x^m(t)\right| \le \delta, 0 \le m < n \right)
&~ P\left( \underset{\tilde{k}(l)\le p \le \tilde{k}(l+1)}{\sup} \left|\big[\boldsymbol{\xi}_{p}\big]_s - \big[\boldsymbol{\xi}_{\tilde{k}(l)}\big]_s \right| > \frac{\delta}{4C_e} \right) \nonumber \\
\le & ~2\exp\left\{-\frac{\delta^2}{32C_e^2\operatorname{Var}\left[\big[\boldsymbol{\xi}_{\tilde{k}(l)+i}\big]_s - \big[\boldsymbol{\xi}_{\tilde{k}(l)}\big]_s\right]}\right\} \\
\le & ~2\exp\left\{-\frac{\delta^2|{s_p}|^2}{32C_{\mathbf{I}}C_e^2\big[c(\tilde{k}(l)) - c(\tilde{k}(l+1))\big]\sigma^2}\right\}.\nonumber
\end{align}
Combining \eqref{eq_p_invariant}, \eqref{eq_lock6-0} and \eqref{eq_lock6}, we have
%\begin{equation}
\begin{align}\label{eq_lock7}
&~P\left( \hat{\textbf{x}}_k \in \mathcal{I}, \forall k \ge 0 \right) \\
\ge &~1 -  8\sum_{l \ge 0} \exp\left\{-\frac{\delta^2|{s_p}|^2}{32C_{\mathbf{I}}C_e^2\big[c(\tilde{k}(l)) - c(\tilde{k}(l+1))\big]\sigma^2}\right\}. \nonumber
%	& P\left( \underset{t\in I_m}{\sup} \left| \bar{x}(t) - x^n(t)\right| > \delta \ \text{for some}\  n \ge 0 \right) \\
%	\le & \sum_{n\ge 0}P\left( \left. \underset{t\in I_m}{\sup} \left| \bar{x}(t) - x^n(t)\right| > \delta \right|   \underset{t\in I_m}{\sup} \left| \bar{x}(t) - x^m(t)\right| \le \delta, 0 \le m < n \right) \le 2 \sum_{n\ge 0} e^{-\frac{\rho\delta_1^2}{b(\tilde{n}(m)) - b(\tilde{n}(m+1))}}.
\end{align}
%\end{equation}

To use Lemma \ref{le_sum}, we assume that the step-size $b_k$ satisfies
\begin{equation}\label{eq_lock_constr2}
c(0) = \sum_{i > 0} b_i^2 \le \frac{\delta^2|{s_p}|^2}{32C_{\mathbf{I}}C_e^2\sigma^2}.
\end{equation}
Then, from Lemma \ref{le_sum}, we can obtain
\begin{equation*}%\label{eq_lock8}
\begin{aligned}
&~\frac{\exp\left\{-\frac{\delta^2|{s_p}|^2}{32C_{\mathbf{I}}C_e^2\big[c(\tilde{k}(l)) - c(\tilde{k}(l+1))\big]\sigma^2}\right\}}{c(\tilde{k}(l)) - c(\tilde{k}(l+1))} \\
\le &~\frac{\exp\left\{-\frac{\delta^2|{s_p}|^2}{32C_{\mathbf{I}}C_e^2 c(0)\sigma^2}\right\} }{c(0)},
\end{aligned}
\end{equation*}
for $c(\tilde{k}(l)) - c(\tilde{k}(l+1)) < c(\tilde{k}(l)) \le c(0)$. Hence, we have
%\begin{equation*}%\label{eq_lock9}
%\begin{aligned}
%	&~\exp\left\{-\frac{\rho\delta^2}{4C_e^2\big[b(\tilde{n}(m)) - b(\tilde{n}(m+1))\big]}\right\} \\
%	= &~\big[ b(\tilde{n}(m)) - b(\tilde{n}(m+1)) \big] \\
%	&~\cdot \frac{\exp\left\{-\frac{\rho\delta^2}{4C_e^2\big[b(\tilde{n}(m)) - b(\tilde{n}(m+1))\big]}\right\}}{b(\tilde{n}(m)) - b(\tilde{n}(m+1))} \\
%	\le &~\big[ b(\tilde{n}(m)) - b(\tilde{n}(m+1)) \big] \cdot \frac{\exp\left\{-\frac{\rho\delta^2}{4C_e^2b(0)}\right\} }{b(0)},
%\end{aligned}
%\end{equation*}
%and
\begin{align}\label{eq_lock7-2}
&\sum_{l\ge 0} \exp\left\{-\frac{\delta^2|{s_p}|^2}{32C_{\mathbf{I}}C_e^2\big[c(\tilde{k}(l)) - c(\tilde{k}(l+1))\big]\sigma^2}\right\} \\
\le &\sum_{l \ge 0} \left[ c(\tilde{k}(l)) - c(\tilde{k}(l+1))\right] \cdot \frac{\exp\left\{-\frac{\delta^2|{s_p}|^2}{32C_{\mathbf{I}}C_e^2 c(0)\sigma^2}\right\}}{c(0)} \nonumber \\
= &c(0) \cdot \frac{\exp\left\{-\frac{\delta^2|{s_p}|^2}{32C_{\mathbf{I}}C_e^2c(0)\sigma^2}\right\} }{c(0)} = \exp\left\{-\frac{\delta^2|{s_p}|^2}{32C_{\mathbf{I}}C_e^2c(0)\sigma^2}\right\}.\nonumber
%e^{-\frac{3\rho\delta^2}{8 \pi^2 \alpha^2 e^{2\delta}}},
\end{align}
As $C_e = e^{L (T+b_1)}$, $c(0) = \sum_{i > 0} b_{i}^2$, and $b_k, T, L$ are given by \eqref{eq_stepsize}, \eqref{eq_T}, \eqref{eq_Lip} separately, we can obtain
\begin{equation}\label{eq_exponential}
\begin{aligned}
\frac{\delta^2|{s_p}|^2}{32C_{\mathbf{I}}C_e^2c(0)\sigma^2} = &~ \frac{\delta^2|{s_p}|^2}{32C_{\mathbf{I}} e^{2L (T+\frac{\alpha}{K_0+1})} \sigma^2\sum\limits_{i \ge 1} \frac{\epsilon^2}{(i+K_0)^2}} \\
=&~\frac{\delta^2}{\sum\limits_{i \ge 1} \frac{32C_{\mathbf{I}} e^{2L (T+\frac{\epsilon}{K_0+1})}}{(i+K_0)^2}} \!\cdot\!\frac{|{s_p}|^2}{\epsilon^2 \sigma^2}.
\end{aligned}
\end{equation}
In \eqref{eq_exponential}, $0 < \delta < \inf_{\textbf{v} \in \partial \mathcal{B}} \left\| \textbf{v} - \hat{\textbf{x}}_\text{b} \right\|$, (\ref{eq_lock_constr}) and (\ref{eq_lock_constr2}) should be satisfied, where a sufficiently large $K_0 \ge 0$ can make both (\ref{eq_lock_constr}) and (\ref{eq_lock_constr2}) true.

To ensure that $\hat{\textbf{x}}_0 + b_{1}\left[\mathbf{f}\left(\hat{\boldsymbol{\psi}}_0, \boldsymbol{\psi}\right)\right]_{3,4}$ does not exceed the mainlobe $\mathcal{B}(\textbf{x})$, i.e., the first step-size $b_{1}$ satisfies
\begin{align*}\left|\hat{x}_{0,1} + b_{1}\left[\mathbf{f}\left(\hat{\boldsymbol{\psi}}_0, \boldsymbol{\psi}\right)\right]_3 - x_1\right| < 1\\
\left|\hat{x}_{0,2} + b_{1}\left[\mathbf{f}\left(\hat{\boldsymbol{\psi}}_0, \boldsymbol{\psi}\right)\right]_4 - x_2\right| < 1\\
\end{align*}
we can obtain the maximum $\epsilon$ as follows
\begin{small}
	\begin{align}%\label{eq_alpha_range}
	\epsilon_{\max} &= {\min}\frac{(K_0+1)}{\left|\!\left[\!\mathbf{f}\left(\hat{\boldsymbol{\psi}}_0, \boldsymbol{\psi}\right)\!\right]_3\!\right|}\left\{1 - \lvert x_1-\hat{x}_{0,1} \rvert, 1 - \lvert x_2-\hat{x}_{0,2}\rvert \right\}\nonumber\\
	&\leq \frac{(K_0+1)}{\left|\!\left[\!\mathbf{f}\left(\hat{\boldsymbol{\psi}}_0, \boldsymbol{\psi}\right)\!\right]_3\!\right|}\\
	& \triangleq \epsilon_b.\nonumber
	\end{align}
\end{small}
Hence, from \eqref{eq_exponential}, we have
\begin{align}\label{eq_exponential-2}
\frac{\delta^2|{s_P}|^2}{32C_{\mathbf{I}}C_e^2c(0)\sigma^2} \!\cdot\! \frac{\epsilon^2\sigma^2}{|{s_p}|^2}  \ge \frac{\delta^2}{\sum\limits_{i \ge 1} \frac{32C_{\mathbf{I}} e^{2L (T+\frac{\epsilon_{b}}{K_0+1})}}{(i+K_0)^2}} \overset{\Delta}{=} C.
\end{align}

Combining \eqref{eq_lock7}, \eqref{eq_lock7-2} and \eqref{eq_exponential-2}, yields
\begin{equation*}
\begin{aligned}
P\left( \hat{\textbf{x}}_k \in \mathcal{I}, \forall k \ge 0 \right) \ge 1 - 8e^{-\frac{C|{s_p}|^2}{\epsilon^2\sigma^2}},
\end{aligned}
\end{equation*}
which completes the proof.

\IEEEpeerreviewmaketitle

\end{document}